\documentclass[11pt]{article}
\usepackage{titlesec}
\usepackage{ntheorem}
\titlespacing*{\section}
{0pt} % Left margin
{1pt} % Space before the section
{1pt}  % Space after the section

\titlespacing*{\subsection}
{0pt} % Left margin
{1pt} % Space before the subsection
{1pt}  % Space after the subsection

\titlespacing*{\subsubsection}
{0pt} % Left margin
{0pt} % Space before the subsubsection
{0pt}  % Space after the subsubsection

\usepackage{pdfpages}
\usepackage{subcaption}
\usepackage{colortbl}
\usepackage{booktabs}
\usepackage{bbm}
\usepackage[normalem]{ulem}
\usepackage{ragged2e}

\usepackage{charter}
\usepackage{pgf}
\usepackage{tikz}
\usetikzlibrary{arrows,automata}
\usepackage{booktabs} 
\usepackage[latin1]{inputenc}
\usepackage{verbatim}
\usepackage{subcaption}

\usepackage{xcolor}

\usepackage{booktabs}
\usepackage{bbm}

\usepackage{natbib}

\colorlet{drkblue}{blue!61.8!black}

\usepackage[bookmarks=false]{hyperref}
\hypersetup{
     pdfborder = {0 0 0},
  colorlinks=true,
 citecolor=drkblue,
linkcolor=drkblue,
urlcolor=drkblue,
pdfstartview={XYZ null null 0.80}
}

\usepackage[hyperpageref]{backref}

\usepackage[left=2.5cm,top=3.5cm,bottom=3.5cm,right=3cm]{geometry}

\usetikzlibrary{math}

\usepackage{amssymb,amsmath,xcolor}

\newtheorem{theorem}{Theorem}

\newtheorem{lemma}{Lemma}[section]
\newtheorem{definition}{Definition}

\newtheorem{remark}{Remark}

\newtheorem{example}{Example}
\newtheorem{axiom}{Axiom}

\theorembodyfont{\rmfamily}
\newdimen\dummy
\dummy=\oddsidemargin
\usetikzlibrary{math}
\usepackage{appendix}

\usepackage{amssymb,amsmath,xcolor}

\newtheorem{prop}{Proposition}[section]

\theorembodyfont{\rmfamily}
\theorembodyfont{\rmfamily}
\newenvironment{proof}[1][Proof:]{\noindent\textbf{#1} }

\theorembodyfont{\rmfamily}
\newdimen\dummy
\dummy=\oddsidemargin
\newcommand{\xyzt}{\{x,y,z,t\}}
\newcommand{\xyz}{\{x,y,z\}}

\newcommand{\yz}{\{y,z\}}
\newcommand{\xz}{\{x,z\}}
\newcommand{\xy}{\{x,y\}}

\newcommand{\x}{\{x\}}
\newcommand{\y}{\{y\}}

\DeclareMathOperator*{\argmax}{argmax} 
\DeclareMathOperator*{\argmin}{argmin}

\title{Decision Making under Dual-System Thinking\thanks{We thank Christopher Chambers, Ruichao Pan, Emel Filiz-Ozbay, Erkut Ozbay, Daniel Vincent, and Siming Ye for helpful discussions and comments.}}

 \author{Yusufcan Masatlioglu\thanks{\protect\linespread{1}\protect\selectfont  Department of Economics, University of Maryland, 3114 Tydings Hall, 7343 Preinkert Dr., College Park, MD 20742. Email: \texttt{yusufcan@umd.edu.}}
\and Tri Phu Vu\thanks{\protect\linespread{1}\protect\selectfont  Department of Economics, University of Maryland, 3114 Tydings Hall, 7343 Preinkert Dr., College Park, MD 20742. Email: \texttt{tvuphu@umd.edu.}}}
\date{May 25, 2025}

\begin{document}

\maketitle
\begin{abstract}
This paper introduces the Dual-System Thinking (DST) model, a decision-theoretic framework that integrates psychological dual-process theories into economic modeling. A single cognitive weight parameter governs the relative influence of the automatic and deliberate cognitive systems. Even the simplest form of DST exhibits distinct behavioral patterns, suggesting that the psychological insights of dual-system theory offer a distinct and valuable approach to modeling choice behavior. Empirically, we show that the model can accommodate several empirical findings and outperform well-known models in discrete choice analysis across various contexts. We also apply the model to study optimal list design and rationality in stochastic environments. 

\end{abstract}


\newpage

\section{Introduction}

Over the past fifty years, dual-system theories have emerged as a cornerstone in psychology, offering a unifying framework for understanding decision-making, reasoning, and social cognition \citep{wason1974dual,evans1989bias, reber1989implicit, epstein1994integration, evans1996rationality, sloman1996empirical, stanovich1999rational}. While these theories vary in detail, they share a foundational premise: human cognition is governed by two distinct systems-one fast, automatic, intuitive, and unconscious (System 1), and the other slow, analytical, deliberate, and conscious (System 2). Following \cite{stanovich1999rational}, we adopt this terminology in this paper.\footnote{Alternative terminologies include Type 1 vs. Type 2, heuristic vs. analytic, associative vs. rule-based, automatic vs. reflective, and intuitive vs. deliberate thinking.}

The influence of dual-system theories has expanded far beyond cognitive psychology, permeating domains as diverse as economics \citep{kahneman2003maps,ilut2023economic}, intelligence \citep{kaufman2011intelligence}, creativity \citep{Barr02012015}, advertising \citep{buttner2014dual}, misinformation \citep{bago2020fake}, interview bias \citep{derous2016your}, and even artificial intelligence \citep{bonnefon2020machine,li2025system}. The influence of these theories is not confined to academia. Notably, institutions like the National Academy of Medicine and the World Bank have incorporated these insights into policy recommendations \citep{ball2015improving,worldbank2015wdr}, underscoring their real-world relevance.

While dual-system models are intuitive, it remains an open question whether they yield observed behaviors that are empirically distinct from those of existing models. One might expect that with sufficient complexity, any model could generate distinct predictions. In contrast, this paper demonstrates that even a simple, parsimonious decision-theoretic model incorporating dual-system thinking predicts behavioral patterns that diverge meaningfully from those produced by existing parametric models. This suggests that the psychological insights of dual-system theory are not only descriptively plausible but also normatively and predictively relevant, offering a distinct and valuable approach to modeling choice behavior.

Any decision-theoretic choice model that incorporates dual-system theories must address three fundamental components: i) Modeling System 1 -- How should automatic, intuitive processing be represented? ii) Modeling System 2 -- How should deliberate, analytical reasoning be captured? iii) Integration Mechanism -- How do the two systems interact to produce a final decision? In this paper, guided by the principle of simplicity, we adopt the most natural and intuitive modeling choices -- those that would first come to mind.

Consistent with its characterization as slow, deliberate, and conscious, we model System 2 as a rational decision-maker. When engaged, System 2 performs exhaustive evaluation of all available alternatives, carefully weighing costs and benefits to arrive at an optimal choice. This follows the classical economic framework of utility maximization, where behavior is governed by optimization of a strict utility function $u$. The rational choice paradigm serves dual roles here: as a normative benchmark for optimal decisions and as a descriptive model of deliberate reasoning. We later relax the assumption of System 2 being perfectly rational.

In contrast, System 1 operates through fast, automatic processes that are inherently variable and sometimes prone to both systematic errors (bias) and random errors (noise). We formalize System 1 using the Luce choice rule \citep{luce1959individual}, which provides a psychologically plausible representation of intuitive decision-making through two key features.  First, probabilistic choice reflects automaticity and variability, capturing random errors in automatic processing (e.g., attentional lapses, noise in perception) and stochasticity in intuitive judgments (consistent with psychological evidence on heuristics). On this point, \cite{kahneman2021noise} provide empirical evidence arguing that many operations and heuristics that System 1 relies on extensively cause ``predictable, directional errors ... as well as noise," highlighting the stochastic nature of choice under System 1. Second, it aligns with System 1's reliance on associative memory, affective responses, and salience-driven attention. Unlike System 2's utility maximization, System 1 can be context-dependent and non-rational. 

The Luce model's probabilistic, context-sensitive, and parsimonious structure makes it an ideal formalization of System 1, complementing System 2's rational utility maximization. By combining the two, the dual-process model can explain both optimal and biased decisions within a unified framework. In the Luce model, choice probabilities are proportional to the attractiveness of each alternative, represented by the function $w$, which captures utility, salience, familiarity, etc. The stochastic nature of choice in the Luce model is driven by alternative-specific noise terms $\varepsilon$ that follow an i.i.d. Gumbel distribution.\footnote{In economics, this distribution typically has cdf $G(\varepsilon) = \exp(-\exp(-\varepsilon))$ and is also referred to as Type I extreme value distribution.} We allow the attractiveness function $w$ (System 1) to differ from the utility function $u$ (System 2), capturing potential conflicts between intuitive and deliberative processes. We also discuss the case where $u$ and $w$ induce the same ordering among
alternatives.

Finally, there are two main theories on how these two systems interact. In the parallel-competitive account, two systems generate independent judgments, with one eventually dominating behavior \citep{sloman1996empirical,smith2000dual}. In the default-interventionist account, System 1 processes produce rapid intuitions, which System 2 processes minimally monitor--either endorsing or overriding them through deliberate reasoning \citep{Evans2007A,kahneman2011thinking,evans2013dual}.  We unify these perspectives through a simple parameter $\alpha \in (0,1)$, representing the probability that System 2 controls the final decision. In other words, System 2 is active with probability $\alpha$.\footnote{We restrict attention to $\alpha \in (0,1)$ so that both cognitive systems are active. The extreme cases $\alpha=1$ (always deliberative) and $\alpha=0$ (always automatic) have been extensively analyzed in the literature, and excluding them has virtually no effect on the model's identification or characterization. Moreover, as shown in Section \ref{sec:identification}, the observed data allow these boundary cases to be readily distinguished.}

We are now ready to define our model formally. The choice can be written as an $\alpha-$mixture of the decisions of Systems 1 and 2. Formally, the choice probability of $x$ in $S$ is given by 

\[
\rho(x,S)=(1-\alpha)  \underbrace{\frac{w(x)}{\sum\limits_{y\in S}w(y)}}_{\text{System 1}} +\alpha  \underbrace{\vphantom{\frac{w(x)}{\sum\limits_{y\in S}w(y)}} \mathbbm{1}_{\{\text{$\arg\max\limits_{y \in S} u(y)$\}}}(x)}_{\text{System 2}}
\]

This parsimonious yet powerful formulation captures both optimal (System 2) and biased, random (System 1) decisions within a unified choice framework. It links the abstract cognitive architecture to observable choice behavior. We call this model Dual-System Thinking (DST).

In addition to its simplicity, the DST model enjoys other desirable properties that facilitate theoretical and empirical analysis. First, the model has a unique representation. In terms of identification, Section \ref{sec:identification} demonstrates that all unobserved primitives of the model, $(\alpha,u,w)$, can be uniquely recovered based on observed choice behaviors in only binary and tripleton menus, thus minimizing data requirements.\footnote{More precisely, utility is ordinal, and hence $u$ is identified only up to strictly increasing transformations. In the main text, we represent $u$ by the linear order it induces, which is unique. Given the normalization the summation of all $w$s is $1$, $w$ is uniquely identified as well.} These uniqueness and identification results facilitate welfare analysis and comparative statics, while also allowing an outside analyst to make out-of-sample predictions (Section \ref{sec:comparative statics}).

Second, the DST model has nice behavioral foundations. Section \ref{sec:characterization} shows that the model is characterized by five simple behavioral postulates. These axioms link the unobserved primitives of the model, $(\alpha,u,w)$, to observed behavior, $\rho$. Our characterization is valid with limited data, as often found in empirical settings. Through the characterization of the DST model, we also uncover a new characterization for the Luce model. This illustrates that our model is distinct from other generalizations of the Luce model.

Third, the DST model can explain well-known empirical patterns and is useful for empirical estimation in discrete choice settings. Section \ref{sec:anomaly} demonstrates how the model accommodates violations of stochastic transitivity \citep{Rieskamp_Busemeyer_Mellers_2006_JEL}, persistence of market shares (illustrated by the famous red bus/blue bus example in \cite{Debreu1960}), and nonzero market shares following a large number of product introductions \citep{bajari2001discrete}. Section \ref{sec:empirical_estimation} shows that estimating the model by maximum likelihood is straightforward both technically and computationally. As an application, we evaluate the DST model on the intertemporal choice data in \cite{roelofsma2000intransitive}, and show that the DST model outperforms various well-known models in discrete choice analysis, including logit, nested logit, independent probit, and covariance probit.

Finally, when additional elements are introduced into the choice environment, many desirable properties of the model are preserved despite the resulting increase in complexity. For example, in an environment with random product availability investigated in \cite{Brady-Rehbeck_2016_ECMA}, the model remains uniquely identified with its behavioral foundations largely unchanged (Section \ref{sec:DST-PA}). Furthermore, the model can accommodate violations of Regularity in this setting, which allows it to explain other empirical patterns such as attraction effects \citep{Huber:1982}.

 Section \ref{sec:application} demonstrates the empirical promise of our framework for firms' decision making by applying the model to an optimal list design problem. In recent years, online platforms such as Amazon, Target, and Apple's App Store are increasingly assuming a dual role, functioning both as marketplaces for third-party sellers and as sellers by offering their own products on their marketplaces \citep{hagiu2022should,padilla2022self,farronato2023self,etro2024commerce}. Empirical studies have indicated that these platforms often manipulate customer search results to promote their own items \citep{chen2019steering,farronato2023self}. This practice, commonly referred to as self-preferencing, has triggered heated policy debates in the United States and European countries due to concerns about antitrust violations and potential negative effects on consumer welfare. Motivated by this phenomenon, we study a platform facing a population of customers with dual-system thinking that aims to design an optimal list of products. Our main finding is that the optimal list may not perfectly align with the platform's priority order. Consequently, products that provide moderate benefits to both the platform and customers may occupy the top positions in the optimal list. This result is primarily driven by potential conflicts between intuitive and deliberative processes in customer choice: the intuitive process tends to favor products at or near the top of the list, while the deliberative process searches exhaustively for the best option.

In our framework, we model System 2 as a rational decision-maker. As this system is engaged in the decision-making process with probability $\alpha$, this parameter measures the degree to which the observed behavior aligns with the deterministic choices of a preference-maximizing agent. Consequently, $\alpha$ provides a natural and conservative measure of stochastic rationality in our setting. In Section \ref{sec:rationality}, we compare this measure with the elegant swaps index introduced in \cite{apesteguia2015measure}, which provides a model-free measure of stochastic rationality. We first identify a simple sufficient condition under which these two measures align with each other and produce the same ordering of rationality. We then argue that the swaps index may indicate that choice behavior becomes less rational even as the likelihood of making rational decisions increases. Our analysis suggests that using a model-free index of stochastic rationality, like the swaps index, can be misleading under certain circumstances.

We aim to develop a parsimonious model of decision-making that incorporates the psychological notion of dual-system thinking. Our modeling choices are motivated by the intrinsic properties of the two systems and their implications for choice behavior. Section \ref{sec:micro_foundation} shows that the DST model can be alternatively interpreted as the solution to a constrained optimization problem. Following other models of deliberate randomization \citep{machina1985stochastic}, the decision maker in our framework intentionally contemplates to determine the likelihood that each option will be chosen. The decision maker is boundedly rational and operates
under dual-system thinking. Drawing on the cognitive miser and default-interventionist theories from psychology, System 1 is activated first and seeks to minimize an expected cost. Each option has an attractiveness level, which may reflect its ease of selection, familiarity, utility, etc. The inverse of an option's attractiveness serves as a metric for the selection cost in our framework. System 2 intervenes when System 1's outcome is unsatisfactory. Once it intervenes, System 2 decides on the magnitude of improvement it will enact upon System 1's outcome. The intervention by System 2 acts as a constraint on System 1's optimization problem, and the magnitude of improvement chosen by System 2 naturally gives rise to the degree of intervention in our model (parameter $\alpha$).

Other sections of the paper are organized as follows. Section \ref{sec:model} introduces our dual-system thinking model (DST), providing identification, characterization, and comparative statics. Section \ref{sec:extensions} generalizes our baseline model to accommodate random product variation and menu-dependent cognitive weights. Section \ref{sec:compare} reviews the literature and compares DST to related models. Section \ref{sec:conclusion} concludes. Additional results and all the proofs are in the appendices.

%%%%%%%%%%%%%%%%%%%%%%%%%%%%%%%%%%%%%%%%%%%%%%%%%%%%%%%%%%%%%
%%%%%%%%%%%%%%%%% Models %%%%%%%%%%%%%%%%%%%%%%%%%%%%%%%%%%%
%%%%%%%%%%%%%%%%%%%%%%%%%%%%%%%%%%%%%%%%%%%%%%%%%%%%%%%%%%%%%

\section{Model}\label{sec:model}
Let $X$ be a finite non-empty set of alternatives. We assume that there are at least three distinct elements in $X$. Let $\mathcal{X}$ be the set of all nonempty subsets of $X$. We refer to each element of $\mathcal{X}$ as a choice set or a menu. For notational simplicity, given a set $S$ and an element $x$, we write $S\cup x$ instead of $S\cup \{x\}$ and $S\setminus x$ instead of $S\setminus \{x\}$. Let the sets of non-negative and positive real numbers be $\mathbb{R}_+$ and $\mathbb{R}_{++}$, respectively.

Let $\mathcal{D}\subseteq \mathcal{X}$ be a collection of menus for which the researcher has data. Our model can accommodate situations in which data are scarce (i.e., $\mathcal{D}\ne \mathcal{X}$), as is often the case in empirical settings. When data are limited, following \cite{Manzini-Mariotti_2014_ECMA}, we assume that $\mathcal{D}$ satisfies the following richness conditions: (1) if $S\in \mathcal{D}$ then $T\in \mathcal{D}$ for all non-empty $T\subseteq S$ and (2) $\{x, y, z\} \in \mathcal{D}$ for all distinct $x,y,z\in X$. The richness conditions imply that the researcher has data from all binary and tripleton menus. A random choice function (RCF) is a mapping $\rho\colon X\times \mathcal{D}\rightarrow [0,1]$ such that $ \sum_{x\in S} \rho(x,S)=1$ for all $S\in \mathcal{D}$; $\rho(x,S) \in [0,1]$ for all $x\in S\in \mathcal{D}$; and $\rho(x,S)=0$ for all $x\notin S$. Here, $\rho(x,S)$ represents the choice probability of $x$ in $S$. Define $\rho(S',S)= \sum_{x\in S'}\rho(x,S)$ as the probability that the chosen alternative belongs to $S'\subseteq S$ when the menu is $S$.

Let  $\succ$ be a linear order over $X$, which is complete, transitive, and antisymmetric. Let $b_{\succ}(S)$ and $l_{\succ}(S)$ denote the best and worst options in $S$ according to $\succ$. We also use the terms $\succ$-best and $\succ$-worst options in $S$ as synonyms for $b_{\succ}(S)$ and $l_{\succ}(S)$, respectively.

To build our model, we introduce three key ingredients. First, the \emph{cognitive weight} $\alpha$, which captures how much influence the deliberate-thinking system has relative to heuristic thinking. Second, the \emph{strict utility function} $u(\cdot)$, which guides choices when decisions are made through deliberate reasoning. Third, the \emph{Luce weight} $w(\cdot)$, which determines choices under fast, intuitive thinking. Because utility is only ordinal---that is, it is identified only up to monotonic transformations---we represent it by a linear order $\succ$, reflecting the strict ranking of alternatives. With these primitives in place, we can now formally define the model.

\begin{definition} A random choice function $\rho$ has a dual-system thinking (DST) representation if there exists a linear order $\succ$ over $X$, a constant $\alpha \in (0,1)$, and a weight function $w\colon X\to \mathbb{R}_{++}$ such that for all $x\in S\in \mathcal{D}$,
\[
\rho(x,S)=(1-\alpha)  \frac{w(x)}{\sum\limits_{y\in S}w(y)} +\alpha \mathbbm{1}(\text{$x$ is $\succ$-best in S)} \ \ \ \tag_{\textbf{DST}}
\]

\end{definition}
Every $(\alpha,\succ,w)$ tuple induces an RCF, which is denoted by $\rho_{\alpha,\succ,w}$. We call $\rho$ a DST with representation $(\alpha,\succ,w)$ if $\rho=\rho_{\alpha,\succ,w}$.  We also say $\succ$ represents $\rho$ if there exist $\alpha$ and $w$ such that $\rho=\rho_{\alpha,\succ,w}$. 

An important special case of our model arises when System 1 and System 2 are aligned,  in the sense that they induce the same ranking over alternatives.  Formally, two cognitive systems are consistent if the preference $\succ$ and the Luce weight $w$ agree on the ranking of every pair $(x,y)$:
\[
\tag_{\textbf{Consistency}} x\succ y \quad \text{ if and only if } \quad w(x)>w(y).
\]
This form of consistency is plausible in certain environments, as System 1, despite making decisions heuristically and intuitively, can be powerful and accurate on average \citep[p.1450]{kahneman2003maps}. 

The DST model allows for more general situations in which the two cognitive systems are not aligned with each other. That is, $x$ may be preferred to $y$ under deliberate thinking, but the weight of $x$ is smaller than the weight of $y$ under heuristic thinking ($w(x)<w(y)$). This scenario can happen when the decision maker selects options based on how salient they are or how easily they come to mind under intuitive thinking. In this case, $w(x)$ may measure the salience of $x$, which in general may deviate from its utility.

When the cognitive weight $\alpha$ goes to 1, the DST approximates the utility maximization model. In contrast, when $\alpha\to 0$, the DST approximates the Luce model. When the Luce weights are the same for all alternatives, DST corresponds to a special case of trembling hand models where choice under intuitive thinking is uniform \citep{harless1994predictive}.

Throughout the paper, we write $w(S)$ for $\sum_{x\in S}w(x)$, and, without loss of generality, we assume that $w(X)=1$. Besides the dual-process thinking interpretation, we show below that the DST model can also be applied to other environments with different interpretations.

\begin{example}[Imperfect recall memory] In a repeated-choice environment, the DM may have experience with all available alternatives in the past. The DM relies on memories to make choices. With probability $\alpha$, the decision maker successfully retrieves the relevant memory and chooses the best available option. With probability $1-\alpha$, the DM unsuccessfully retrieves the relevant memory and randomizes via a Luce rule.
\end{example}

\begin{example}[Objective recommendation] Suppose a ranking (a linear order) $\succ$ over alternatives is recommended to the DM. The recommended information is objective and the DM follows the suggested ranking with probability $\alpha$. When she trusts the ranking, she acts according to the ranking and chooses the best-ranked option. With probability $1-\alpha$, she does not take the ranking into account and acts according to her Luce randomization rule. 
\end{example}

\begin{example}[Dual selves] The decision maker has two selves. One self always chooses a single alternative by deterministic preference maximization. The other self prefers a mixture of alternatives to a single alternative \citep{machina1985stochastic} and does so via a Luce rule.
\end{example}

Before we move to the next section, we state the uniqueness property of DST. Suppose $\rho$ has a DST representation. Proposition \ref{prop:identification} states that such a representation is unique. 
\begin{prop}\label{prop:identification}
    Suppose $\rho$ has a DST representation $(\alpha,\succ,w)$. Then $(\alpha,\succ,w)$ is unique.
\end{prop}

To appreciate the unique identification of our model, note that there are three sources of variation in a DST. First, we can vary the relative importance of the two cognitive systems by varying $\alpha$. Second, we can change the preference of the decision maker. Note that the number of possible preference orderings is large, even when there are few options.\footnote{For instance, there are $6,24,120,$ and $720$ potential preferences for $3,4,5,$ and $6$ distinct options, respectively.} Third, variations can occur in the weight function that governs choice under intuitive thinking. Proposition \ref{prop:identification} consolidates these three types of variations and asserts that the representation is unique.

%%%%%%%%%%%%%%%%%%%%%%%%%%%%%%%%%%%%%%%%%%%%%%%%%%%%%%%%%%%
%%%%%%%%%%%%%%%%%%% Identification section%%%%%%%%%%%%%%%%
%%%%%%%%%%%%%%%%%%%%%%%%%%%%%%%%%%%%%%%%%%%%%%%%%%%%%%%%%%%
\subsection{Identification and Out-of-Sample Prediction}\label{sec:identification}
We investigate how to infer the model's primitives from observed data in this subsection. Proposition \ref{prop:identification} establishes that \((\alpha,\succ,w)\) is, in principle, identifiable, but it is silent on how to recover these parameters in practice. We show below that the parameters are uniquely identified by using only choice data from binary and tripleton menus, thereby minimizing data requirements.

\smallskip
\noindent\textbf{Identifying the preference.} A distinctive feature of our model lies in how the underlying preference is identified from choice data. In many models, identification of preference between two alternatives are intrinsically pairwise. Our identification strategy departs from this template. Because both cognitive systems are simultaneously active, the pairwise construction by itself cannot disentangle the deliberate ranking, due to the interaction between two systems. To recover the ranking between $x$ and $y$, one must observe how the relative odds of $x$ against $y$ respond to the introduction of a third alternative $z$, and this response must be read off jointly across all three pairings within $\{x,y,z\}$.

This three-way structure is what makes our identification novel and more involved than the usual pairwise construction. The identification strategy relies on the absolute difference in the relative choice probabilities between tripleton and binary menus. For arbitrary pairwise distinct $x,y,z$, let $D_{xy,z}$ be the difference in the relative choice probability between $x$ and $y$ when adding $z$ to the binary menu $\{x,y\}$. Formally,
\[
D_{xy,z} \;=\; \frac{\rho(x,\{x,y,z\})}{\rho(y,\{x,y,z\})}
              \;-\; \frac{\rho(x,\{x,y\})}{\rho(y,\{x,y\})}.
\]
The sign of $D_{xy,z}$ indicates how the presence of $z$ affects the relative odds of $x$ versus $y$:  
\begin{itemize}
    \item If $D_{xy,z} > 0$, then adding $z$ improves the choice probability of $x$ relative to $y$.  
    \item If $D_{xy,z} < 0$, then adding $z$ improves the choice probability of $y$ relative to $x$.  
\end{itemize}
By construction, $D_{xy,z}  D_{yx,z} \le 0$. Note also that in the Luce model $D_{xy,z}= 0$ for all $x,y,z$, whereas in our model $D_{xy,z} \ne 0$ under any preference ordering.\footnote{We assume that $\alpha\in (0,1)$ in a DST so that both systems will be active. Our framework can also capture the extreme cases of $\alpha$ by looking at $D_{xy,z}$. If $\alpha=1$, $D_{xy,z}$ is not well-defined for some $x,y,z$ because certain choice probabilities are zero. For $\alpha=0$ (the Luce model), $D_{xy,z}$ is well-defined and equals zero for all $x,y,z$.} 

\medskip

By observing the sign pattern of the three differences $(D_{xy,z}, D_{yz,x}, D_{zx,y})$, we can identify the complete strict ranking among $x,y,$ and $z$. In our model, the improvement relation cannot contain a cycle, i.e.\ the cases $(+,+,+)$ or $(-,-,-)$ are ruled out. Hence, there are six admissible sign patterns, each of which corresponds to one particular strict preference ordering as shown in the following table:

\begin{center}
\begin{tabular}{ccccl}
\toprule
 $D_{xy,z}$ & $D_{yz,x}$ & $D_{zx,y}$ & $\;\;\;\;$ & Implied Ranking \\
\midrule
$+$ & $-$ & $-$ & & $x \succ_R y \succ_R z$ \\
$-$ & $+$ & $-$ & & $y \succ_R z \succ_R x$ \\
$-$ & $-$ & $+$ & & $z \succ_R x \succ_R y$ \\
$+$ & $+$ & $-$ & & $x \succ_R z \succ_R y$ \\
$-$ & $+$ & $+$ & & $y \succ_R x \succ_R z$ \\
$+$ & $-$ & $+$ & & $z \succ_R y \succ_R x$ \\
\bottomrule
\end{tabular}
\end{center}

Motivated by this observation, for pairwise distinct $x,y,z$, we define the revealed ranking relation $\succ_R$ as
\begin{equation}\label{eq:revealed_preference_definition}
x \succ_R y \, \text{ and } \, y \succ_R z \, \text{ and } \, x \succ_R z \quad \text{if} \quad 
D_{xy,z} > 0, \;\; D_{yz,x} < 0, \text{ and } D_{zx,y} < 0.
\end{equation}

In our model, the relative odds of the best alternative always improve in the presence of a new alternative. Hence, $x$ being preferred to both $y$ and $z$ implies that $D_{xy,z} > 0$ and $D_{xz,y} > 0$ (equivalently $D_{zx,y} < 0$ ). On the other hand, the presence of the best alternative improves the relative odds of the worst alternative; hence, $x$ is preferred to $y$ and $y$ is preferred to $z$ imply $D_{yz,x} < 0$. One interesting observation is that our revealed preferences are based on three options as opposed to binary comparisons. Hence, one needs to show that $\succ_R$ is a well-defined linear order. The next result shows that $\succ_R$ is indeed well-defined and represents $\rho$ if $\rho$ admits a DST representation.

\begin{prop}\label{prop:identification_1} \it{(Revealed Preference)} Suppose $\rho$ has a DST representation. Then $\succ_R$ represents $\rho$.  
\end{prop}

Proposition \ref{prop:identification_1} allows us to identify the underlying preference of a DST model. With the preference recovered, the task of identifying $\alpha$ and $w$ becomes straightforward through simple algebra. Appendix \ref{sec:Appendix_identification_DST} provides explicit formulas of $\alpha$ and $w$ using choices from binary and tripleton menus.

\medskip
\noindent\textbf{Out-of-sample prediction}. A useful consequence of our identification result is that the analyst, having observed choices only from binary and tripleton menus, can recover the primitives $(\alpha,\succ,w)$ and thereby predict choice behavior on every larger menu, which need not belong to the data $\mathcal{D}$. We now use this implication to make a cautionary point. A natural, model-free instinct is to read welfare directly off observed choice frequencies: an alternative that is chosen rarely must be undesirable, and one that is chosen rarely on \emph{every} observed menu must be the worst. We show that this instinct is not merely imprecise but can be exactly inverted. Thus, a welfare ranking based on choice probabilities can be misleading.

Suppose the analyst observes choices from certain binary and tripleton menus and would like to use the observed data to predict choices from $X=\xyzt$, which is not observed. As an example, suppose the data satisfy the following conditions (b1)-(b3), which govern choices from binaries, and condition (t1), which specifies choices from tripletons:
\begin{itemize}
    \item[(b1)] $x$ is chosen with probability 2/5 in any pair including $x$.
    \item[(b2)] $y$ is chosen with probability 3/5 in any pair including $y$.
    \item[(b3)] $z$ is chosen with probability 3/5 in any pair including $z$ but not $y$.    
    \item [(t1)] $x$ is chosen with probability 11/35 in either $\{x,y,z\}$ or $\{x,z,t\}$. In each menu, the other two members are each chosen with probability 12/35. 
\end{itemize}
Note that in every scenario for which data are available, the decision maker chooses $x$ with the lowest probability. A model-free approach based on choice frequencies would unambiguously conclude that $x$ is the least preferred option. Consequently, it would predict that $x$ must also be chosen with the smallest probability in the grand set.\footnote{Some model-based approaches also share this property. The Luce model and the Additive Perturbed Utility (APU) model of \cite{Fudenberg_Iijima_Strzalecki_2015_ECMA}, for instance, require that if $\rho(x,S)\le \rho(y,S)$ for some $S\supseteq\xy $ then $\rho(x,S')\le \rho(y,S')$ for all menus $S'\supseteq\xy $.} As we show below, these simplistic interpretations of observed data can be deeply misleading.

Assume that the decision maker follows our DST model. Applying the identification strategies from our model, the analyst obtains $\alpha=1/5$, the revealed preference of $x\succ_R y \succ_R z \succ_R t$, and the Luce weights of $w(x)=1/10, w(y)=w(z)=w(t)=3/10$. Hence, far from being the least preferred, $x$ is identified as the best option, yielding a starkly different conclusion compared to the conventional model-free approach. Using these parameters, the analyst can also predict that, from $X$, $\rho(x,X)=7/25$ and $\rho(y,X)=\rho(z,X)=\rho(t,X)=6/25$. Consequently, the alternative with the lowest choice probability in every smaller menu where it is available (i.e., $x$) is predicted to have the highest choice probability in the grand set. This surprising prediction arises from the inconsistency between the two systems: $x$ is the best option under System 2, but its weight under System 1 is smallest. An analyst who fails to account for this inconsistency and follows a model-free approach would not only misidentify the best alternative but also make qualitatively incorrect out-of-sample predictions. Our model provides an analytical framework for disentangling true preference from biases, thus reducing potential errors in both prediction and welfare analysis.

%%%%%%%%%%%%%%%%%%%%%%%%%%%%%%%%%%%%%%%%%%%%%%%%%%%%%%%%%%%
%%%%%%%%%%%%%%%%%%% Characterization %%%%%%%%%%%%%%%%%%%%%
%%%%%%%%%%%%%%%%%%%%%%%%%%%%%%%%%%%%%%%%%%%%%%%%%%%%%%%%%%%

\subsection{Behavioral Characterization}\label{sec:characterization}
We now discuss the behavioral implications of our model. The behavioral implications are the bridge between the abstract functional form and observed behavior. The first axiom says that the probability of choosing a feasible option is always non-zero.\footnote{This assumption will make $D_{xy,z}$ well-defined for all $x,y,$ and $z$, which enables us to define the revealed preferences.} 

\begin{axiom}[Positivity]\label{axiom:positivity} $\rho(x,S)>0$ for all $(x,S)$ such that $x\in S\in \mathcal{D}$.
\end{axiom}

The second axiom formalizes the rationality of the DM in terms of the revealed preference defined in $(\ref{eq:revealed_preference_definition})$. Proposition \ref{prop:identification_1} states that the revealed preference corresponds to the underlying preference ordering in the DST. Hence, Axiom \ref{axiom:rationality} requires that this revealed preference constitutes a linear order over alternatives in $X$.

\begin{axiom}[Rationality]\label{axiom:rationality} $\succ_R$ is complete, transitive, and antisymmetric.
\end{axiom}

 The third axiom is a well-known condition in the stochastic choice literature. It states that choice
probabilities strictly decrease as competition among alternatives increases. 
\begin{axiom}[Strict Regularity]\label{ax:Strict_regularity} $\rho(x,S) < \rho(x,S \setminus y)$ for all $x,y\in S, x\ne y$ and $S\in \mathcal{D}$. 
\end{axiom}

In the DST model, choice by System 1 satisfies strict regularity, whereas choice by System 2 satisfies only regularity, which is obtained by replacing the strict inequality in Axiom \ref{ax:Strict_regularity} with a weak inequality. Since System 1 is activated with a positive probability, the DST model satisfies strict regularity.

To state our next axiom, we ask a simple counterfactual question: when a competitor $y$ is removed from a menu $S$, by how much does the choice probability of a remaining option $x$ rise? Intuitively, the answer should track how popular $y$ was to begin with. Removing a frequently chosen alternative frees up a large mass of probability to be redistributed among the survivors, while removing a rarely chosen one barely moves anything. We therefore expect the gain in the choice probability of $x$ to scale with $\rho(y,S)$.

We first provide a measure for the gain to $x$ not in absolute terms but relative to $x$'s own post-removal probability,
\[
A(x,y,S) = \frac{\rho(x,S \setminus y) - \rho(x,S)}{\rho(x,S \setminus y)},
\qquad x\neq y,\ x,y\in S\in\mathcal{D}.
\]
Normalizing by $\rho(x,S\setminus y)$ rather than reporting the raw difference $\rho(x,S\setminus y)-\rho(x,S)$ is what frees the measure from context: the absolute amount of probability flowing to $x$ depends on how many other competitors remain to share in the redistribution, whereas the \emph{rate} $A(x,y,S)$ does not. The quantity is a discrete analogue of a substitution effect---it records how much probability ``flows'' from $y$ to $x$ once $y$ is gone.

Our next axiom imposes that $A(x,y,S)$ is proportional to $\rho(y,S)$ as long as neither of $x$ or $y$ are the best options. Probability mass attached to the best option is partly the product of System 2's $\alpha$-bonus rather than System 1's heuristic competition, so flows into or out of the best option (denoted by $b_{\succ_R}(S)$) are governed by a different mechanism and do not obey the same proportionality. Our axiom accordingly restricts attention to the substitution patterns among non-best alternatives, where only the heuristic system is at work.
\begin{axiom}[Constant Proportional Gain]\label{axiom: Constant Gain}
$A(x,y,S) \propto \rho(y,S)$ for all $(x,y,S)$ with $x,y\in S\in\mathcal{D}$, $x\neq y$, and $x,y\neq b_{\succ_R}(S)$, where the constant of proportionality is uniform over all such $(x,y,S)$.
\end{axiom}
The axiom says that, except the best option, the rate at which $x$ gains from the removal of $y$ is a fixed multiple of how often $y$ was chosen. The Luce model is the special case in which that constant equals one: there $A(x,y,S)=\rho(y,S)$ exactly, a direct consequence of Luce's Choice Axiom. In our model the constant is instead determined by the cognitive weight $\alpha$, and its departure from one is precisely the trace left by the deliberate system on the heuristic substitution pattern.

For the last axiom, suppose the revealed preference satisfies $x\succ_R y\succ_R z$. Our model then satisfies
\begin{equation}\label{eq:relationship}
  A(z,x,\xyz)+A(z,y,\xyz)+A(y,z,\xyz)=1.
\end{equation}
Each term on the left-hand side of equation \eqref{eq:relationship} measures the gain to the \emph{worst} remaining option after one alternative is deleted. Specifically, removing $x$ from $\xyz$ leaves $\yz$, in which $z$ is $\succ_R$-worst, so $A(z,x,\xyz)$ is the gain to the worst survivor. Similarly, removing $y$ again leaves $z$ as the worst survivor, and removing $z$ leaves $y$ as the worst. Equation \eqref{eq:relationship} thus says that the gains accruing to the worst-ranked survivor, summed over which alternative is deleted, total exactly one.

The content of equation \eqref{eq:relationship} is easiest to see by contrast with the Luce model. Because $A(x,y,S)=\rho(y,S)$ under Luce's Choice Axiom, the analogous sum equals one there as well for any order. Within the Luce model the identity is therefore order-free. Our model breaks this invariance: the sum equals one when the worst survivor is taken with respect to the revealed preference $\succ_R$, and in general fails for any other order. Hence equation \eqref{eq:relationship} can be seen as an identifying restriction: among all orders, it singles out the one that represents $\rho$. The same identity extends to menus of any size, yielding our final axiom. In Axiom \ref{axiom:5}, recall that $l_{\succ_R}(S\setminus x)$ denotes the worst option in $S\setminus x$ with respect to $\succ_R$.

\begin{axiom}\label{axiom:5}
$\displaystyle \sum_{x\in S} A\bigl(l_{\succ_R}(S\setminus x),\,x,\,S\bigr)=1$ for all menus $S\in\mathcal{D}$ with $|S|\ge 3$.
\end{axiom}
We are now ready to characterize our model.

\begin{theorem}\label{thm:characterization} $\rho$ has a DST representation if and only if it satisfies Axioms \ref{axiom:positivity}-\ref{axiom:5}.
\end{theorem}

Theorem \ref{thm:characterization} allows the researcher to verify and falsify the DST model by checking simple conditions. This theorem can also be used to characterize a particular case of the DST model under Consistency, i.e., when the preference and Luce weights induce the same ordering among alternatives. In this situation, in addition to Axioms \ref{axiom:positivity}-\ref{axiom:5}, Remark \ref{remark:characterization consistency} states that the characterization requires an additional simple condition:  if $x\succ_R y\succ_R z$ then $\rho(x,\xz)>\max\{\rho(x,\xy),\rho(y,\yz)\}$. This condition is analogous to the strong stochastic transitivity property, except that the antecedent uses the revealed preference relation $\succ_R$ explicitly, rather than the order induced by choice probabilities in binary menus. Under a DST representation, the condition is equivalent to requiring $w(x)>w(y)>w(z)$ whenever $x\succ_R y\succ_R z$, and therefore exactly aligns the Luce weights with the revealed preference.

\begin{remark}\label{remark:characterization consistency} $\rho$ has a DST representation where the two systems are consistent if and only if it satisfies Axioms \ref{axiom:positivity}-\ref{axiom:5} and $\rho(x,\xz)>\max\{\rho(x,\xy),\rho(y,\yz)\}$ whenever $x\succ_R y\succ_R z$.
\end{remark}

The proof of Theorem \ref{thm:characterization} uses induction based on the number of alternatives in a choice set. It also helps uncover a novel characterization of the Luce model. To elaborate, consider a more restrictive version of Axiom \ref{axiom:5}. In Axiom \ref{axiom:5}, the ranking among alternatives is the revealed preference. Axiom \ref{axiom:5} can be extended to Axiom \ref{axiom:Luce-2} that accommodates an arbitrary ranking. Axiom \ref{axiom:Luce-2} states that the total gains of the worst-ranked option sum to $1$ under arbitrary orders; hence, it constitutes a stronger version of Axiom \ref{axiom:5}.

\begingroup
  \setcounter{axiom}{4}      
  \renewcommand{\theaxiom}{\arabic{axiom}$'$}  
\begin{axiom}\label{axiom:Luce-2} For each linear order $\succ$ over $X$, $\displaystyle \sum_{x\in S} A(l_{\succ}(S\setminus x),x,S)=1$ for all menus $S\in\mathcal{D}$ with at least three alternatives.
\end{axiom}
\endgroup

Proposition \ref{prop:characterization-Luce} below states that Axiom \ref{axiom:Luce-2} is both necessary and sufficient for a positive random choice function to admit a Luce representation, thereby providing a novel characterization of the Luce model. 

\begin{prop}\label{prop:characterization-Luce} A positive $\rho$ has a Luce representation in $\mathcal{D}$ if and only if it satisfies Axiom \ref{axiom:Luce-2}.
\end{prop}

Proposition \ref{prop:characterization-Luce} characterizes the Luce model in its classical setting from a different perspective, using limited data. Note that the standard Luce model does not assume any ordering. Proposition \ref{prop:characterization-Luce} shows that one can introduce an arbitrary order in the Luce model without altering its structure. One can then incorporate that order into the properties the Luce model satisfies, and then require them to hold under every possible order. Consequently, Proposition \ref{prop:characterization-Luce} offers a new perspective on the Luce model and its underlying properties.\footnote{See \cite{dougan2021odds} for an alternative characterization of the Luce model using a different approach.}

Together with Theorem \ref{thm:characterization}, Proposition \ref{prop:characterization-Luce} implies that our model can be interpreted as a Luce model with an underlying preference.

%%%%%%%%%%%%%%%%%%%%%%%%%%%%%%%%%%%%%%%%%%%%%%%%%%%%%%%%%%%%%%%%%%%%%
%%%%%%%%%%%%%%%%%%%%%%%% Accomodating Choice Anomalies %%%%%%%%%%%%%%
%%%%%%%%%%%%%%%%%%%%%%%%%%%%%%%%%%%%%%%%%%%%%%%%%%%%%%%%%%%%%%%%%%%%%

\subsection{Comparative Statics}\label{sec:comparative statics}
In this subsection, we study the comparative statics of two DST models characterized by different parameters. To investigate the link between observed choice and the model's primitives, we need an order on random choice functions. Following \cite{ABL2017SCRUM} and \cite{masatlioglu2024growing}, we define a partial order between two random choice functions as follows.

\begin{definition}[First-order stochastic dominance, FOSD]\label{defn:FOSD} Let $\succ$ be a linear order over $X$.  $\rho$ first-order stochastically dominates  $\rho'$ with respect to $\succ$ (write $\rho$ FOSD$_{\succ}$ $\rho'$)  if $\sum\limits_{y:y\succsim x}\rho(y,S)\ge \sum\limits_{y:y\succsim x}\rho'(y,S)$ for all $x\in S\in \mathcal{D}$. Here, $y\succsim x$ means $y=x$ or $y\succ x$.
\end{definition} 

Definition \ref{defn:FOSD} states that $\rho$ first-order stochastically dominates $\rho'$ if, for every choice set, better-ranked options are selected with higher probabilities under $\rho$. Proposition \ref{prop:comparative static-1} below establishes the relationship between two DST models with the same underlying preference.

\begin{prop}\label{prop:comparative static-1} Suppose $\rho$ and $\rho'$ have DST representations $(\alpha,\succ,w)$ and $(\alpha',\succ,w')$, respectively. 
\begin{itemize}
    \item [(i)] Suppose $w=w'$.  Then $\rho$ FOSD$_{\succ}$ $\rho'$ if and only if $\alpha\ge \alpha'$. 
    \item [(ii)] Suppose $\alpha= \alpha'$. Then $\rho$ FOSD$_{\succ}$ $\rho'$ if and only if $\frac{w(x)}{w'(x)}\ge \frac{w(y)}{w'(y)}$ for all pairs $(x,y)$ such that $x\succ y$.    
\end{itemize}
\end{prop}

Part (i) of Proposition \ref{prop:comparative static-1} applies to the situation in which $\rho$ and $\rho'$ differ only in cognitive weights. It states that deliberate thinking benefits the decision maker: the more she deliberates $(\alpha\ge \alpha')$, the better choices she makes ($\rho$ FOSD$_{\succ}$ $\rho'$). Intuitively, this result follows from the fact that the DM makes no mistake under System 2. Part (ii) of Proposition \ref{prop:comparative static-1} corresponds to the case where $\rho$ and $\rho'$ differ solely in Luce weights. It indicates that choices are better if the Luce weights satisfy the well-known monotone likelihood ratio property, which in our setting means that preferred options receive higher relative weights. As the intuitive-thinking system in a DST model follows the Luce rule, this result also holds for the comparative statics of two Luce models.\footnote{When $\rho$ and $\rho'$ may differ in both cognitive weights and the Luce weights, it can be shown that $\rho$ FOSD$_{\succ}$ $\rho'$ if $\alpha\ge \alpha'$ and $\frac{w(x)}{w'(x)}\ge \frac{w(y)}{w'(y)}$ for all pairs $(x,y)$ such that $x\succ y$.}

Proposition \ref{prop:comparative static-2} compares two DST models that differ only in preferences. In the absence of a natural ranking among preference orders, the comparative statics become less meaningful as it is unclear how to determine if one preference is better than another. To address this issue, we introduce a benchmark order over $X$. Call this linear order $\succ_b$. This order can be interpreted as a normative preference or a preference of a representative agent. Given $\succ_b$, one can assess how closely any preference order approximates $\succ_b$ and thereby rank preference orders accordingly. One way to measure this proximity is via the single-crossing property \citep{milgrom1994monotone,ABL2017SCRUM}. Specifically, for two preferences $\succ$ and $\succ'$, the ordered collection $\{\succ',\succ\}$ is single-crossing with respect to $\succ_b$ if for every pair $(x,y)$ such that $x\succ_b y$, the relation $x\succ y$ holds whenever $x\succ' y$. Intuitively, this condition implies that $\succ$ is more closely aligned with $\succ_b$ than $\succ'$.

Proposition \ref{prop:comparative static-2} below states that the decision maker makes better choices with respect to the benchmark order when her preference is closer to it and vice versa. 

\begin{prop}\label{prop:comparative static-2} Suppose $\rho$ and $\rho'$ have DST representations $(\alpha,\succ,w)$ and $(\alpha,\succ',w)$, respectively.  Then $\rho$ FOSD$_{\succ_b}$ $\rho'$ if and only if $\{\succ',\succ\}$ is single-crossing with respect to $\succ_b$.
\end{prop}

%%%%%%%%%%%%%%%%%%%%%%%%%%%%%%%%%%%%%%%%%%%%%%%%%%%%%%%%%%%%%%%%%%%%%
%%%%%%%%%%%%%%%%%%%%%%%% DST as a tool %%%%%%%%%%%%%%%%%%%%%%
%%%%%%%%%%%%%%%%%%%%%%%%%%%%%%%%%%%%%%%%%%%%%%%%%%%%%%%%%%%%%%%%%%%%%

\section{DST as a Tool in Discrete Choice Analysis}
In this section, we demonstrate the empirical usefulness of the DST model as a tool in analyzing discrete choice data. We first illustrate how our model accounts for key empirical patterns in random choice and market demand. We then present empirical estimations and show that the DST model outperforms well-known models in discrete choice analysis, including (multinomial) logit, nested logit and probit, across a wide range of choice contexts.

\subsection{Accommodating Empirical Patterns}\label{sec:anomaly}
Our model is a minor deviation from the Luce model. Yet, its explanatory power surpasses that of the Luce model. Deliberate and automatic thinking play pivotal roles in generating the additional explanatory power of our approach. 

\noindent \textbf{Stochastic intransitivity}. \cite{tversky1969intransitivity}, among others, argue that transitivity may fail in risky environments; \cite{Rieskamp_Busemeyer_Mellers_2006_JEL} provides a detailed review of related empirical evidence where stochastic transitivity is frequently violated. Three well-known properties are weak stochastic transitivity (WST), moderate stochastic transitivity (MST), and strong stochastic transitivity (SST). WST states that for all pairwise distinct $x,y,z \in X,$ $\rho(x,\xy)\ge 1/2$ and $ \rho(y,\yz)\ge 1/2$ implies $\rho(x,\xz)\ge 1/2.$\footnote{MST and SST keep the antecedent part in WST unchanged but replace the consequent part with $\rho(x,\xz)\ge \min\{ \rho(y,\yz),\rho(x,\xy)\}$ and $\rho(x,\xz)\ge \max\{ \rho(y,\yz),\rho(x,\xy)\}$, respectively.} 

It is straightforward to verify that the Luce model rules out any violations of SST--a natural consequence of its highly restrictive structure. What is more surprising is that even more flexible frameworks, such as the Additive Perturbed Utility (APU) model of \cite{Fudenberg_Iijima_Strzalecki_2015_ECMA}, also impose SST. By contrast, our model can be viewed as only a modest departure from Luce: both System 1 and System 2, when considered separately, still satisfy SST. Yet, somewhat unexpectedly, their combination in DST generates additional flexibility, allowing for WST violations and thus richer patterns of choice.

\begin{remark}\label{rem:WST_violated} DST accommodates violations of weak stochastic transitivity.
\end{remark}

To illustrate how DST can accommodate WST violations, consider $\rho$ with DST representation $(\alpha,\succ,w)$, where $\alpha=1/4$, $x\succ y \succ z$, $w(x)=1/6, w(y)=1/3,$ and $w(z)=1/2$. Note that the two cognitive systems induce opposite orderings over the three alternatives. These parameters generate $\rho(x,\xy)= 1/2$, $ \rho(y,\yz)=11/20$, and $\rho(x,\xz)=7/16$. As $7/16<1/2$, this is a violation of weak stochastic transitivity. The main driver of this result is that Systems 1 and 2 are not aligned. One can show that the DST model predicts SST if the preference of System 2 and the Luce weight of System 1 induce the same ordering among alternatives; this result is stated in Remark \ref{rem:SST_satisfied}.

\begin{remark}\label{rem:SST_satisfied} DST satisfies strong stochastic transitivity if the two cognitive systems are consistent.
\end{remark}

\smallskip
\noindent \textbf{Persistence of market share and violation of Luce's IIA}. We now turn to showing how our model can more naturally capture the response of market shares to the introduction of new products. \cite{Debreu1960} introduces the Red-Bus/Blue-Bus paradox as a serious limitation of the Luce framework. The Red-Bus/Blue-Bus paradox shows that the Luce model unrealistically assumes too much substitutability between alternatives. When a new product identical to an existing one is introduced, the model takes share equally from all alternatives instead of only from the closest substitute. We now show that DST eliminates this restrictive prediction of the Luce model. 

Suppose the decision maker chooses, with equal probabilities, a train ($t$) and a red bus ($r$) as a means of transport. Then the Luce model predicts that introducing a blue bus that is a close substitute for the red bus (but not for the train) does not alter the relative likelihood of choosing the train. \cite{Debreu1960} argued that adding the blue bus should increase the odds of choosing $t$ over either bus, thus violating Luce's IIA. We now show that the DST model accommodates any choice probability of the train (when all options are available) between $0.33$ and $0.5$.

To see this, assume, in every binary set, the choice probability ratios in pairwise choices for any two alternatives are equal to 1:
\begin{equation}\label{eq:equal binary choice}
   \frac{\rho(t,\{t,r\})}{\rho(r,\{t,r\})}=\frac{\rho(t,\{t,b\})}{\rho(b,\{t,b\})}=\frac{\rho(b,\{r,b\})}{\rho(r,\{r,b\})}=1 
\end{equation}

 To capture these binary choices, suppose that the preference relation is now a weak order; this is to capture the DM's genuine indifference between the two buses. Assume that all options that tie for the best are selected with equal probabilities under deliberate thinking; the model is unchanged otherwise. Suppose the decision maker follows a DST model with preference $t\succ b\sim r$. Let $w(r)=w(b)$. This, together with $r\sim b$, reflects the DM's genuine indifference between the two buses. We find $\alpha$ and weight function $w$ such that (\ref{eq:equal binary choice}) holds. As we assume the DM picks all options that tie for the best with equal probabilities under slow thinking, it is straightforward that $\frac{\rho(b,\{r,b\})}{\rho(r,\{r,b\})}=1$. Let $0<\alpha<1/2$ and $w(t)/w(r)=1-2\alpha$. Under these parameters, (\ref{eq:equal binary choice}) holds. Additionally, when all options are available, the red bus and blue bus are still chosen with the same probability. Consequently, the odds of choosing the train over either bus increase in the tripleton menu as the choice probability of the train is given by $\rho(t,\{t,r,b\})=1/(3-2\alpha)$, which could be any number between $0.33$ and $0.5$, depending on $\alpha$. If $\alpha$ approaches $1/2$, the choice probability of the train converges to $0.5$; in this case, the likelihood of choosing the train remains the same as in binary menus. If $\alpha$ approaches $0$, we get the same prediction as the Luce model in this setting.\footnote{Our explanation of the red bus/blue bus example presented here is somewhat similar to that in \cite{Manzini-Mariotti_2014_ECMA}; they also require $t\succ b\sim r$ to explain the phenomenon.}

\smallskip
\noindent \textbf{Market shares of dominant products.} As established by \cite{bajari2001discrete}, almost all popular discrete choice models (multinomial logit, nested logit, random coefficients, etc.) predict that the market share of an individual product must approach zero as the number of available products in the market grows large. This prediction leads to the counterintuitive implication that a superior product could see its market share drop to zero following the introduction of numerous inferior alternatives. We demonstrate that the DST model is not subject to this limitation, as a dominant product can maintain a non-negligible market share even as the number of alternatives becomes arbitrarily large.

To illustrate, let $S$ be the set of all existing products in the market. Suppose $S$ contains an inferior product $x$ and a superior product $y$, which is also the best product in $S$. Consider a set $R_n= \{x_1,\dots, x_n \}$ that consists of $n$ replicas of $x$, where each replica $x_i$ has the same utility and Luce weight as $x$. Suppose the replicas in $R_n$ are introduced to the market. As the number of replicas grows large, the market share of $y$ converges to $\alpha$:
\[
\lim_{n\to \infty} \rho(y,S\cup R_n) = \alpha>0.
\]
The result above follows directly from the DST model's structure: when System 2 is engaged in the decision-making process, $y$ is always chosen regardless of the number of inferior replicas in the market. Meanwhile, when System 1 is activated, the likelihood of selecting $y$ approaches zero as the number of replicas grows to infinity.

%%%%%%%%%%%%%%%%%%%%%%%%%%%%%%%%%%%%%%%%%%%%%%%%%%%%%%%%%%
%%%%%%%%%%%%%%%% This is the new section %%%%%%%%%%%%%%%%%%
%%%%%%%%%%%%%%%%%%%%%%%%%%%%%%%%%%%%%%%%%%%%%%%%%%%%%%%%%%
\subsection{Empirical Estimation}\label{sec:empirical_estimation}

In Section \ref{sec:identification} we describe how to identify the model's primitives from observed choice frequencies when the data admit a DST representation. In empirical studies, an analyst is often interested in finding the best-fitting models. This section shows that the analyst can estimate our model easily by using maximum likelihood (MLE).

For a fixed linear order among available alternatives, estimating our model by MLE is arguably no different from estimating a (multinomial) logit model both in terms of technical difficulty and computational cost. Let $f(x,S)\in [0,1]$ denote the observed frequency of choosing $x$ in menu $S\in \mathcal{D}$. Let $\sigma(S)\in [0,1]$ be the probability that menu $S$ is realized. Given linear order $\succ$ among the alternatives, the DST model can be estimated by maximizing the following log-likelihood function:\footnote{In many cases the number of choice observations does not vary across menus. In these situations, $\sigma(S)$ can be removed from the log-likelihood function without changing the optimal solutions.}
\[
\max_{(\alpha,w)} L((\alpha,w))=\sum_{(x,S)} \sigma(S) f(x,S) \log(\rho_{\alpha,\succ,w}(x,S)).
\]
This log-likelihood function is well-behaved and converges rapidly in most cases. Note that the preference ordering in the log-likelihood is exogenous. In principle, the analyst can search over all possible linear orders to identify the best-fit ordering. In practice, to minimize computational cost, the analyst can opt for a plausible ordering guided either by the observed choice frequencies, by theoretical considerations, or by prior findings in the literature. As we will describe, identifying such an ordering is usually straightforward, and in many cases, estimating the DST model using that ordering already yields significant improvements in model-fit relative to (multinomial) logit, nested logit and probit models.

As an application, we analyze the intertemporal choice data in \cite{roelofsma2000intransitive}. In their experimental study, 88 subjects made deterministic choices in six binary-menu questions involving four options $x$, $y$, $z$, and $t$ given by 
\[
x=(7 \text{ dfl}, 1 \text{ week}),y=(8 \text{ dfl}, 2 \text{ weeks}), z=(9 \text{ dfl}, 4 \text{ weeks}), \text{and } t=(10 \text{ dfl}, 7 \text{ weeks}),
\]
where ($a$ \text{ dfl}, $b$ \text{ weeks}) denotes the option that pays $a$ Dutch guilders $b$ weeks after the experiment. Table \ref{tab:intertemporal_choice_1} reports the aggregated choice frequencies.

We estimate logit, nested logit, independent probit, covariance probit and DST models using MLE.\footnote{In independent probit, the error terms are independent and identically distributed: $\varepsilon_i \sim \text{i.i.d}\, N(0,1)$. In covariance probit, the error terms may have different variances and nonzero covariances across alternatives.} When estimating our model, we assume that the preference among options follows the ordering $x\succ y\succ z \succ t$. This ordering naturally arises from a preference for earlier payoff. It is the modal choice pattern in individual choice data and is also consistent with both exponential and hyperbolic discounting utility models across a wide range of parameters \citep{roelofsma2000intransitive}.\footnote{Since there are six binary choice questions, there are $2^6=64$ deterministic choice patterns. The pattern $x\succ y\succ z \succ t$ accounts for $18\%$ of the data. No other deterministic choice pattern accounts for more than 9\% of the data.} The estimated parameters of the DST model are given by 

\[
\hat{\alpha} = 0.625 \quad \text{ and } \quad \hat{w}(x)=0.005,\, \hat{w}(y)=0.209,\, \hat{w}(z)=0.268, \, \hat{w}(t)=0.517. 
\]

Table \ref{tab:intertemporal_choice_1} reports estimated choice probabilities and goodness-of-fit statistics of various models. We use the adjusted $R^2$ and the McFadden pseudo $R^2$ to evaluate model fit. As the empirical choice frequencies violate moderate stochastic transitivity, the logit and independent probit models fit the data poorly. Their generalizations (nested logit and covariance probit) provide a better fit. Our model fits the data very well: its performance is comparable to that of covariance probit despite having only half as many free parameters (4 in DST versus 8 in covariance probit). The DST model also performs slightly better than the order-dependent Luce model (ODLM) of \cite{tserenjigmid2021order}.

\begin{table}[ht!]
    \centering
    \renewcommand{\arraystretch}{1.1}
    \setlength{\tabcolsep}{3pt} 
    \footnotesize
    \setlength{\tabcolsep}{2pt} 
    \begin{tabular}{l|
        >{\centering\arraybackslash}p{1.5cm}%
        >{\centering\arraybackslash}p{1.5cm}%
        >{\centering\arraybackslash}p{1.5cm}%
        >{\centering\arraybackslash}p{1.5cm}%
        >{\centering\arraybackslash}p{1.5cm}%
        >{\centering\arraybackslash}p{1.5cm}|
        ccc}
    \hline\hline
         & \multicolumn{6}{c|}{Choice frequencies} & \multicolumn{3}{c}{Goodness-of-fit measures} \\
    \cline{2-10}
         & $\rho(x,\xy)$ & $\rho(x,\xz)$ & $\rho(x,\{x,t\})$ & $\rho(y,\yz)$ & $\rho(y,\{y,t\})$ & $\rho(z,\{z,t\})$ & Adjusted $R^2$ & Pseudo $R^2$ & \# free para. \\
    \hline
         Data & 0.68 & 0.59 & 0.60 & 0.83 & 0.68 & 0.80 & & & \\
         \hline
         Logit & 0.51 & 0.63 & 0.73 & 0.62 & 0.72 & 0.62 & 0.45 & 0.07 & 3\\
         Ind. probit & 0.51 & 0.62 & 0.73 & 0.62 & 0.72 & 0.62 & 0.44 & 0.07 & 3\\
         Nested logit & 0.68  & 0.67 & 0.68 & 0.67 & 0.68 & 0.80 & 0.87 & 0.12 & 5 \\
         Cov. probit & 0.63 & 0.61 & 0.63 & 0.83 & 0.74& 0.74 & 0.89 & 0.14 & 8 \\
         ODLM & 0.68 & 0.64 & 0.55 & 0.80 & 0.73 & 0.80 & 0.81 & 0.14 & 5 \\
         \hline
         Our model & 0.63 & 0.63 & 0.63 & 0.79 & 0.73 & 0.75 & 0.95 & 0.14 & 4\\ 
    \hline
    \hline
    \end{tabular}
    \caption{Original data and estimated choice frequencies for logit, nested logit, independent probit, covariance probit, DST, and ODLM.}
    \label{tab:intertemporal_choice_1}
    \vspace{0.1cm}
    \parbox{\textwidth}{\footnotesize \justifying
    \noindent \textit{Notes:} Ind. probit: Independent probit; Cov. probit: Covariance probit; \# free para.: The number of free parameters in each model. Choice frequencies and adjusted $R^2$ of ODLM are from \cite{tserenjigmid2021order}. In the nested logit model, we use at most two nests since there are only four alternatives.}
\end{table}

\noindent \textbf{Other choice contexts.} To assess whether the DST model's strong performance is specific to the dataset in \cite{roelofsma2000intransitive}, we evaluate it on a diverse set of choice data described in Appendix \ref{sec:appendix_datasets}. Each of the experimental datasets in Appendix \ref{sec:appendix_datasets} has more observations than alternatives $(|\mathcal{D}|>|X|)$, facilitating stable estimation and meaningful goodness-of-fit evaluation. These datasets vary across three dimensions: (i) the domains in which choices are made, (ii) the number of options and types of menus, and (iii) whether the data satisfy certain rationality properties (regularity and stochastic transitivity). As before, for each dataset, we estimate (multinomial) logit, nested logit, independent probit, covariance probit, and DST models using MLE and compare model fit using the adjusted $R^2$ and the McFadden pseudo $R^2$. Across this heterogeneous collection of datasets, a common picture emerges: the DST model is able to outperform other well-known models in discrete choice estimation. This result highlights the model's broad applicability and usefulness across a wide range of choice environments.

%%%%%%%%%%%%%%%%%%%%%%%%%%%%%%%%%%%%%%%%%%%%%%%%%%%%%%%%%%
%%%%%%%%%%%%%%%% This is the new section %%%%%%%%%%%%%%%%%%
%%%%%%%%%%%%%%%%%%%%%%%%%%%%%%%%%%%%%%%%%%%%%%%%%%%%%%%%%%
\section{DST and Firm Decision-Making: An Application to Optimal List Design}\label{sec:application}

In recent years, there has been a growing body of literature investigating the manipulation
of online behaviors. Several studies have documented that online platforms such as Amazon,
Walmart, Google Shopping, and Apple's App Store can intentionally manipulate their offerings to favor their own products \citep{hagiu2022should,padilla2022self,farronato2023self,motta2023self}. This manipulation involves tactics such as strategically recommending
their own products or prominently featuring their offerings. For example, on Amazon, \cite{farronato2023self} and \cite{waldfogel2024amazon} find that Amazon-labeled products consistently receive higher rankings in consumer search results relative to observably comparable products.

What makes such tactics effective is that they exploit the way consumers actually make choices. Prominent placement increases a product's salience, biasing the fast, heuristic component of decision-making even when consumers' deliberative preferences point elsewhere. This wedge between the two systems is precisely what our model formalizes, making the platform's design problem a natural application. We therefore use our framework to study the optimal list-design problem faced by a platform serving consumers with dual-system thinking. Consider a setting in which a platform wants to present a set $S$ of products to a heterogeneous population of consumers by ordering them in a list. A list $l$ is a bijective function from $S$ to $\{1,\dots,|S|\}$, with $l(x)$ denoting the position of product $x$. 

\smallskip
\noindent \textbf{The customers.} The customers are indexed by $1,2,\dots,K$, each following a DST model. To keep the analysis simple, we assume that (1) all customers observe the same list chosen by the platform, and (2) all customers have the same preference, but they differ in the degree of deliberate thinking and Luce weights. The former assumption is justified when the platform cannot observe customers' characteristics or identities and therefore cannot offer customized lists.\footnote{Alternatively, the platform may possess customers' personal data but be legally prohibited from using it. Recent regulations, such as the European Union's General Data Protection Regulation and the California Consumer Privacy Act, place significant restrictions on the use of personal data.} The common-preference assumption can be interpreted as the preference of a representative consumer.

When System 1 is engaged in the decision-making process, customers make choices based on the perceived salience of items. The perceived salience of a product depends only on its position in the list and is allowed to be heterogeneous. For individual $i$, the salience of product $x$ under list $l$ is $w_i^l(x):= w_i({l(x)})$. We assume that the function $w_i\colon \{1,2,\dots,|S|\}\to \mathbb{R}_{++}$ is strictly decreasing so that salience decreases with position in the list.

Under System 2, customers make decisions based on the perceived utility of items. Each product $x$ has an observed utility $u(x)$. We assume that the list chosen by the platform can potentially impact a product's perceived utility. Instead of merely drawing attention, top positions in the list can act as a signal of unobserved quality, influencing a consumer's evaluation of the product. The perceived utility of product $x$ under list $l$, denoted by $u^l(x)$, is the sum of its observed utility and an unobserved utility component associated with its position:
\[
u^l (x)=u(x)+\beta(l(x)) \quad \text{ for all $x\in S$}.
\]
Here, $\beta\colon \{1,2,\dots,|S|\}\to \mathbb{R}_+$ is a weakly decreasing function representing the position-dependent utility boost. As $\beta(\cdot)$ is weakly decreasing, options listed earlier in the list have weakly higher utility boosts. We further assume that the ordering of perceived utilities is strict regardless of the list $l$ chosen by the platform; this technical assumption simplifies the analysis.

In summary, customer $i$ follows a DST model with representation $(\alpha_i,u^l,w^l_i)$. The platform knows the intrinsic utility ($u$) and how list positions affect System 2's preferences ($u^l$), which are the same for all customers, but does not know $\alpha_i$ or $w^l_i$, which are individual-specific. Additionally, the platform does not observe the distribution of customers, denoted by $(\eta_1,\dots,\eta_K)$. Let $\rho^l$ be the choice of the population when observing list $l$. Mathematically,
\[
\rho^l(x,S)=\sum_{i=1}^K \eta_i\rho_{\alpha_i,u^l,w^l_i}(x,S) \quad \text{for all $x\in S$.}
\]

\smallskip
\noindent \textbf{The platform.} Each product $x \in S$ yields a payoff $\tau_x>0$ to the platform. Depending on the platform's business model and objectives, $\tau_x$ may be interpreted as the platform's margin, commission, revenue, or profit from product $x$. Let $\tau=(\tau_x)_{x\in S}$ denote the payoff vector. We assume that these payoffs are pairwise distinct; that is, $\tau_x\ne \tau_y$ when $x\ne y$. Given payoff vector $\tau$, let $f_\tau(\rho^l)$ denote the platform's utility when choosing list $l$. We assume that the non-parametric utility function $f_\tau$ satisfies the following monotonicity property. Let $\succ_\tau$ be the ordering of products implied by their payoffs, i.e., for two products $x$ and $y$, if the payoff of $x$ is higher than that of $y$, then $x\succ_\tau y$. We assume that
\[
\tag_{Monotonicity} f_\tau(\rho^l) \ge f_\tau(\rho^{l'}) \quad \quad \text{if $\quad$ $\rho^l \, \, \text{FOSD}_{\succ_\tau} \, \, \rho^{l'}$},
\]
with strict inequality when $\rho^l\ne \rho^{l'}$. In the condition above, $\rho^l \,\, \text{FOSD}_{\succ_\tau} \,\, \rho^{l'}$ means that $\rho^l$ assigns weakly greater cumulative probability to higher-payoff products than $\rho^{l'}$ does (see Definition \ref{defn:FOSD} for a formal definition). Monotonicity, therefore, simply states that the platform prefers the list in which products with greater payoffs are selected with higher probabilities. The class of functions satisfying monotonicity is large. It includes, for example, the well-known family of functions with constant elasticity of substitution ($f_\tau(\rho^l)=(\sum_{x\in S} \rho^l(x,S)\tau_x ^{\frac{\delta-1}{\delta}})^{\frac{\delta}{\delta-1}}$, where $\delta>1$ is the elasticity of substitution).\footnote{The expected payoff utility function ($f_\tau(\rho^l)=\sum_{x\in S} \tau_x \rho^l(x,S)$) and Cobb-Douglas utility function ($f_\tau(\rho^l)=\prod_{x\in S} \tau_x^{\rho^l(x,S)})$ are special cases of this family as $\delta\rightarrow \infty$ and $\delta\rightarrow 1$, respectively.} Before presenting the analysis, we formally define an optimal list. 

\begin{definition}A list is optimal if it solves $\max_{l\in \mathcal{L}}f_\tau(\rho^l)$, where $\mathcal{L}$ is the set of all possible lists.
\end{definition} 

\medskip
\noindent \textbf{Result on optimal list.} In what follows, let $l^*$ denote the optimal list. Abusing the notation slightly, we also use $[x,y,z,\dots]$ to denote a list where $x$ is ranked first, $y$ is ranked second, $z$ is ranked third, and so on. Let $[x_1,x_2,\dots,x_{|S|}]$ denote the payoff-ordered list, where $\tau_{x_1}>\tau_{x_2}>\cdots>\tau_{x_{|S|}}$. Let $x_h$, where $h\in \{1,\dots,{|S|}\}$, denote the product with the highest observed utility.

A candidate for the optimal list is the payoff-ordered list. As choice in System 1 depends entirely on product salience and salience values decrease with positions, listing products in descending order of their payoffs maximizes the platform's utility under System 1. Consequently, the payoff-ordered list is optimal for the platform in certain circumstances. One specific circumstance is when the platform's choice of list is irrelevant to the customer's decisions under System 2.\footnote{This can happen when, for instance, $u(x_h)+\beta(|S|)> \max_{k\ne h} u(x_k)+ \beta(1)$. Under this condition, the product with the highest observed utility also has the highest perceived utility under every list chosen by the platform; hence, the platform's choice of list does not affect customers' decisions under System 2. This may arise, for example, when list positions do not affect utilities $(\beta(k)=0$ for all $k$), when they affect all products identically $(\beta(k)$ is a constant for all $k$), or when the position-induced utility boosts are small. In these cases, all customers select the product $x_h$ under System 2 in any list chosen by the platform. Consequently, $l^*=[x_1,x_2,\dots,x_{|S|}]$.}

In more general cases, an optimal list can significantly deviate from the payoff-ordered list. While listing products according to their payoff maximizes the platform's utility under System 1, customers may choose a product with a small payoff under System 2 when seeing this list, thereby hurting the platform's aggregate utility. Consequently, the platform may have an incentive to depart from the payoff-ordered list in order to induce customers to choose a product that yields a higher payoff when System 2 is activated. This potential trade-off between System 1 and System 2 determines which list is optimal.

As the platform's optimal list design is a finite optimization problem, it does not admit a closed-form solution in general. Proposition \ref{prop:optimal_list_1} below presents the structure of an optimal list and shows that an optimal list can be partitioned into three blocks with different properties. Note that Proposition \ref{prop:optimal_list_1} applies to any distribution of customer types and any strict preference profile of customers.

\begin{prop} \label{prop:optimal_list_1} In any optimal list $l^*$, there exists a unique product $x_a$ with $a\le h$ such that
\begin{itemize}
    \item[(i)] \text{[First block]} \[
l^*(x_i)=
\begin{cases}
m,   & \text{if} \quad  i=a, \quad \quad \quad \quad \text{where} \quad  m\le a \\[4pt]
i,   & \text{if} \quad 1\le i \le m-1, \\[4pt]
i+1, & \text{if} \quad m\le i \le a-1.
\end{cases}
\]
    \item [(ii)] \text{[Last block]} $l^*(x)<l^*(y)$ whenever $\tau_x>\tau_y$ and $\min\{l^*(x),l^*(y)\}>l^*(x_h)$. Here, $l^*(x_h)\ge m$.
\end{itemize}
\end{prop}

The first block in an optimal list comprises the first $a$ products with the highest payoffs. In this block, products are ordered by their payoffs except possibly for $x_a$ (part (i) of Proposition \ref{prop:optimal_list_1}). The proof of Proposition \ref{prop:optimal_list_1} identifies $x_a$ as the product with the highest perceived utility under the optimal list. The second (possibly empty) block contains options ranked immediately after the first block until product $x_h$. The ordering within this block may deviate from the payoff ordering. The last block consists of products placed after $x_h$. In this final block, the optimal list again respects the payoff ordering (part (ii) of Proposition \ref{prop:optimal_list_1}).\footnote{One corollary of Proposition \ref{prop:optimal_list_1} is that if $x_a\equiv x_1$ then the optimal list is unique and identical to the payoff-ordered list, i.e., $l^*=[x_1,x_2,\dots,x_{|S|}]$ when $x_a\equiv x_1$.}

Proposition \ref{prop:optimal_list_1} shows that the platform's optimal strategy depends on both the customers' underlying preferences and the platform's payoffs. It also indicates that listing products based on their payoffs may be suboptimal for the platform due to the potential conflict between System 1 and System 2 described above. Example \ref{ex:optimal_list} below illustrates the trade-off between the two systems and highlights three interesting features that may occur in an optimal list. First, a product that provides only a moderate payoff to the platform may occupy the top position. Second, the product with the highest observed utility may be placed at the bottom. Finally, a product that is least valuable to the platform and least appealing to customers in terms of observed utility may nonetheless appear before other products.

\begin{example} \label{ex:optimal_list} Suppose all customers are identical and have $\alpha=1/2$. There are 5 products $x_1,\dots,x_5$ with payoffs $\tau_{x_j}=1-j/6$. Let observed utilities be $u(x_1)=1.5$, $u(x_2)=1$, $u(x_3)=2$, $u(x_4)=2.5$, and $u(x_5)=0.5$. The position-dependent utility boost and the salience value associated with the $j$th position in the list are given by $\beta(j)=3/2-3j/20$ and $w(j)=2/5-j/15$. The platform maximizes the expected payoff, i.e., $f_\tau(\rho^l)=\sum_{i=1}^5 \tau_{x_i} \rho^l(x_i,S)$, where $S=\{x_1,x_2,\dots,x_5\}$.
\begin{itemize}
    \item Under the payoff-ordered list, $[x_1,x_2,x_3,x_4,x_5]$, the platform's utility under System 1 is $11/18$. The product with the highest perceived utility is $x_4$; hence, the platform's utility under System 2 is $\tau_{x_4}=1/3$. The platform's utility under the payoff-ordered list is $17/36$.
    \item The optimal list is $[x_3,x_1,x_2,x_5,x_4]$. Here, $x_3$ occupies the top position despite providing only a moderate payoff to the platform; $x_4$ appears last despite having the highest observed utility; and $x_5$ precedes $x_4$ even though it is least appealing to both the customers and the platform. Under the optimal list, the platform's utility under System 1 is  $17/30$, which is smaller than that under the payoff-ordered list. This decrease in utility, however, is offset by the increase in utility when System 2 is activated. The product with the highest perceived utility under the optimal list is $x_3$; hence, the platform's utility under System 2 is $\tau_{x_3}=1/2$, which is higher than that under the payoff-ordered list. The platform's utility under the optimal list is $8/15>17/36$.
\end{itemize}
\end{example}

Proposition \ref{prop:optimal_list_1} and Example \ref{ex:optimal_list} are empirically useful in understanding the platform's behavior and identifying list manipulation. Suppose an analyst observes the list chosen by the platform (such a list is often public information) and seeks to identify the product that yields the highest benefit to the platform. Proposition \ref{prop:optimal_list_1} indicates that such a product must be ranked either first or second in the observed list, conditional on the observed list being optimal. Moreover, the platform can strategically position items of moderate value to them near or at the top of the list, as shown in Example \ref{ex:optimal_list}. This observation indicates that the presence of low-to-middle-value products (to the platform) near or at the top of the list does not automatically eliminate the potential issue of information manipulation.

We conclude this section by outlining an algorithmic approach for solving the optimal list design problem. In general, the problem depends on the functional form of the platform's objective function, $f_\tau$. In a particular situation where $f_\tau$ corresponds to the expected payoff and the platform can observe the customers' choices, we show in Appendix \ref{sec:Appendix_linear_programming} that the platform's optimization reduces to a 0-1 linear program similar to an assignment problem. In this case, one can use some well-known algorithms, such as branch-and-bound algorithm, to identify optimal solutions; see Appendix \ref{sec:Appendix_linear_programming} for details.

%%%%%%%%%%%%%%%%%%%%%%%%%%%%%%%%%%%%%%%%%%%%%%%%%%%%%%%%%%
%%%%%%%%%%%%%%%% This is the new section %%%%%%%%%%%%%%%%%%
%%%%%%%%%%%%%%%%%%%%%%%%%%%%%%%%%%%%%%%%%%%%%%%%%%%%%%%%%%
\section{Cognitive Weight ($\alpha$) as a Measure of Stochastic Rationality} \label{sec:rationality}

In our framework, System 2 corresponds to the classical rational utility maximizer. Accordingly, the cognitive weight $\alpha$ on System 2 can be interpreted as a measure of rationality: a larger $\alpha$ indicates behavior more closely aligned with utility maximization. In this section, we formalize this interpretation by relating $\alpha$ to the well-known \emph{swaps index} of \cite{apesteguia2015measure}, which is a model-free measure of stochastic rationality. 

\cite{apesteguia2015measure} raises two fundamental questions: (i) How severe are the deviations from rationality? (ii) What are the underlying preferences for the purpose of welfare analysis? The paper provides a novel measure of stochastic rationality using a model-free approach. For a linear order $\succ$, the inconsistency of a single observation 
 $x$ chosen from $S$ with respect to $\succ$ is the number of items in $S$ ranked above $x$. That is equivalent to the number of swaps needed to make $x$ the top-ranked item under that ranking. Summing these swaps (option by option and menu by menu, weighted by how often each menu appears and how often each option is chosen) yields a total inconsistency score; the swaps index is the minimum such score over all rankings. The ranking that attains it is the swaps preference--the welfare ranking closest to the data.

Formally, let $\sigma\colon \mathcal{D}\to (0,1)$ denote a fixed probability distribution over all menus where data are available, with $\sigma(S)$ being the frequency the decision maker faces menu $S$.\footnote{$\sigma$ has full support on $\mathcal{D}$ because $\mathcal{D}$ is the set of menus where data are available by definition.} For an arbitrary RCF $\rho$ and a linear order $\succ$ on $X$, let
\[
I(\rho,\succ)=\sum_{S\in \mathcal{D}} \sum_{x\in S} \sigma(S) \rho(x,S) |\{y: y\in S, y\succ x\}|
\]
be the total amount of ``mistake" observed in $\rho$ w.r.t. $\succ$. In the functional form of $I(\rho,\succ)$, the cardinality of the set $\{y: y\in S, y\succ x\}$ measures the distance from $x$ to the best available option and thus represents the magnitude of the mistake when choosing $x$. This magnitude is then weighted by the probability that $S$ is realized and $x$ is chosen in the computation of $I(\rho,\succ)$. The swaps index of $\rho$, denoted by $I_{swaps}(\rho)$, is defined as the minimal observed mistake, i.e., it is the minimum of $I(\rho,\succ)$ taken over all linear orders on $X$. They also refer to the linear order(s) that attains this minimum as the swaps preference. Denote the swaps preference by $\succ_{swaps}$. Mathematically,
\[
I_{swaps}(\rho)=\min_{\succ} I(\rho,\succ) \quad  \text{ and } \quad \succ_{swaps}(\rho)\in \argmin_{\succ} I(\rho,\succ).
\]
Note that the swaps index is a measure of \textit{irrationality}; hence, $\rho$ is more rational when $I_{swaps}(\rho)$ is smaller. Additionally, the swaps index and the swaps preferences are independent of the underlying model. Therefore, this approach is a model-free measure of stochastic rationality. 

One might wonder what happens if we apply this measure to data generated by our model. For example, it is curious to see i) whether the preference of the decision maker (in System 2) coincides with the swaps preference, and ii) the cognitive weight $\alpha$ is consistent with the swaps index. 
Proposition \ref{prop:swaps index} below presents these relationships among swaps index, swaps preference, and the primitives of a DST satisfying Consistency. In Proposition \ref{prop:swaps index}, the monotone likelihood ratio condition ($w(x)/w'(x)\ge w(y)/w'(y)$ whenever $x\succ y$) implies that, relative to $w'$, $w$ places more weight on preferred options; hence, the condition indicates that $w$ is more consistent with (closer to) $\succ$ than $w'$.\footnote{This condition also appears in comparative statics in Proposition \ref{prop:comparative static-1}. If $w$ and $w'$ are proportional, then they induce the same Luce choice
probabilities and are equally close to $\succ$.}

\begin{prop}\label{prop:swaps index} Consider $\rho$ and $\rho'$ with consistent DST representations $(\alpha,\succ,w)$ and $(\alpha',\succ,w')$, respectively. 
\begin{itemize}
    \item[(i)] Suppose $\frac{w(x)}{w'(x)}\ge\frac{w(y)}{w'(y)}$ whenever $x\succ y$. If $\alpha>\alpha'$ then $I_{swaps}(\rho)<I_{swaps}(\rho')$.
    \item[(ii)] The swaps preferences of $\rho$ and $\rho'$ are unique and both identical to  $\succ$.
\end{itemize}
\end{prop}

Proposition \ref{prop:swaps index} holds under all possible distributions over the menus. When System 1 and System 2 are consistent, part (i) of Proposition \ref{prop:swaps index} states that the swaps index and our model-based cognitive weight produce the same ordinal ranking of rationality. Consequently, greater deliberation by the decision maker corresponds to higher rationality according to the swaps index. The proof of Proposition \ref{prop:swaps index} shows that the converse also holds when $\rho$ and $\rho'$ share the same Luce weights (i.e., $w\equiv w'$). In that case, higher rationality (a lower swap index) implies more deliberate thinking.

Part (ii) of Proposition \ref{prop:swaps index} indicates that the underlying preference of the DM is indeed the unique order that results in minimum mistakes. Hence, the swaps preference manages to accurately capture the true preference of the DM. Overall, Proposition \ref{prop:swaps index} lends support to the use of the swaps index as a measure of stochastic rationality under Consistency. The proof of Proposition \ref{prop:swaps index} uses Theorem 1 in \cite{apesteguia2015measure}.

So far, we have demonstrated that the swaps index successfully reflects rationality when the two systems are aligned. We now turn to cases in which the systems diverge, and show that the index can fail to provide the correct comparative ranking of agents' rationality. Put differently, the swaps index may paradoxically conclude that choice behavior becomes less rational even as the probability of rational choice increases. We also show that swaps-based preferences may fail to recover the decision maker's underlying preference relation.

In what follows, we assume that System 2 is dominant; that is, $\alpha > 0.5$. This ensures that each agent selects the best available alternative in any menu at least half of the time. Both Examples \ref{ex:swaps_index_alternative} and \ref{ex:swaps_index_alternative_1} below illustrate how the swaps index can produce results directly at odds with our measure $\alpha$. In both cases, we assume that the agents share the same true preference ordering. In Example \ref{ex:swaps_index_alternative}, they also share identical System 1 weights, differing only in how often they rely on the rational system. We show that not only is the relative ranking inverted, but the swaps preferences also fail to reproduce the true preferences.

\begin{example}\label{ex:swaps_index_alternative} Consider three alternatives with  $x\succ y \succ z$. Suppose $w(x)=0.25, w(y)=0.05$, and $w(z)=0.70$, so that the ranking induced by the Luce weight is different from the preference ranking. Suppose $\mathcal{D}$ consists of all menus of sizes 2 and 3. The likelihood of observing each menu in $\mathcal{D}$ is given by $\sigma(\xyz)=\sigma(\yz)=0.45$ and $\sigma(\xy)=\sigma(\xz)=0.05$. Let $\rho$ and $\rho'$ have DST representations $(\alpha=0.60, \succ, w)$ and $(\alpha'=0.55, \succ, w)$, respectively. Since $\alpha>\alpha'>0.5$, the propensity to think deliberately under $\rho$ is higher than that under $\rho'$. However, the swaps index ranks $\rho'$ more rational than $\rho$. Moreover, neither $\succ_{swaps}(\rho)$ nor $\succ_{swaps}(\rho')$ reflects the DM's true preference. The two are identical and represent a ``compromise" order between the ordering implied by the Luce weights and the DM's preference: $\succ_{swaps}(\rho) \equiv \, \,\succ_{swaps}(\rho') \equiv x\succ z \succ y$.
\end{example}

Example \ref{ex:swaps_index_alternative} also illustrates that the swaps index fails to respect stochastic dominance. The comparative statics in Proposition \ref{prop:comparative static-1} imply that if $\rho$ and $\rho'$ admit DST representations $(\alpha,\succ,w)$ and $(\alpha',\succ,w)$ with $\alpha>\alpha'$, then $\rho$ first-order stochastically dominates $\rho'$ with respect to $\succ$. This is the case in Example \ref{ex:swaps_index_alternative}, yet the swaps index nonetheless ranks $\rho'$ as more rational than $\rho$.

In Example \ref{ex:swaps_index_alternative_1} below, we demonstrate that the swaps index produces a contradictory ranking even when both agents are nearly fully rational (with $\alpha = 0.95$ and $\alpha' = 0.90$, respectively). Moreover, although the swaps-based preferences coincide with the true ordering, the index nevertheless assigns the opposite ranking.

\begin{example} \label{ex:swaps_index_alternative_1} Assume $\succ$, $\sigma$, and $w$ are as in Example \ref{ex:swaps_index_alternative}. Let $\rho$ and $\rho'$ have DST representations $(\alpha=0.95, \succ, w)$ and $(\alpha'=0.90, \succ, w')$, respectively. Here, $w'$ is obtained from $w$ by swapping the values of $y$ and $z$: $w'(x)=0.25, w'(y)=0.70$, and $w'(z)=0.05$. Then, the swaps index still ranks $\rho'$ as more rational than $\rho$, despite i) $\alpha>\alpha'$ and ii) both $\succ_{swaps}(\rho)$ and $\succ_{swaps}(\rho')$ correctly capturing the DM's true preference.
\end{example}

Examples \ref{ex:swaps_index_alternative} and \ref{ex:swaps_index_alternative_1} highlight the fact that using model-free criteria, such as the swaps index, to measure stochastic rationality can be misleading in certain circumstances. This conclusion continues to hold under an alternative measure of stochastic rationality introduced in \cite{ok2023measuring} (see Appendix \ref{sec:appendix_rat_index}). Consequently, model-free criteria should be applied with caution when analyzing rationality in stochastic choice settings.

%%%%%%%%%%%%%%%%%%%%%%%%%%%%%%%%%%%%%%%%%%%%%%%%%%%%%%%%%%%%%
%%%%%%%%%%%%%%%%% Micro Foundation %%%%%%%%%%%%%%%%%%%%%%%%%%
%%%%%%%%%%%%%%%%%%%%%%%%%%%%%%%%%%%%%%%%%%%%%%%%%%%%%%%%%%%%%

\section{Micro Foundation: Constrained Optimization}\label{sec:micro_foundation}

In this paper, our aim is to develop a model of decision-making that incorporates the psychological notion of dual-system theories. Our framework is deliberately parsimonious, and it naturally raises some questions. How does the agent determine how often to engage System 2? Why is System 1 modeled using the Luce framework? We address these issues in this section. Inspired by the idea that stochastic choice can result from an optimization problem \citep{machina1985stochastic, Fudenberg_Iijima_Strzalecki_2015_ECMA,cerreia2019deliberately,chambers2025weighted}, we show that the DST model is the solution to a simple constrained optimization problem of a boundedly rational agent operating under dual-system thinking.

The optimization problem is grounded in two leading theories from psychology about how dual-system thinking operates. The \textit{cognitive-miser theory} asserts that individuals tend to exert minimal cognitive and processing effort when making decisions. Accordingly, automatic processes with low computational expense are often activated first and the brain is engaged only when these automatic
processes fail.\footnote{\cite{hull2001science} claims that: ``the rule that human beings seem to follow is to engage the brain only when all else fails, and usually not even then'' (p. 37), which is agreed by \cite{richerson2008not}: ``In effect, all animals are under stringent selection pressure to be as stupid as they can get away with" (p. 135).} Based on the cognitive-miser theory, the \textit{default-interventionist theory} of System 1-System 2 posits that the default is to use System 1's processes whenever possible. System 2 acts as a check-and-intervention mechanism and is often engaged only to override failure by System 1 when its solution to the problem at hand is not sufficiently good \cite[p.69]{stanovich2009distinguishing}. In this account, the cognitive weight $\alpha$ in the DST model represents the frequency at which System 2 intervenes and overrides System 1's decision.

As System 1 is automatic, intuitive, heuristic, and low in cognitive effort, it is natural to assume that System 1 makes choices based on the ease of selection and seeks to minimize the cost of choosing. Assume that each option $x$ is associated with a parameter $w(x)>0$ that measures how easy it is to be selected when System 1 is engaged. Oftentimes $w(x)$ corresponds to the salience or familiarity of $x$. A higher value of $w(x)$ implies that the option is more salient or more familiar, making it easier and less costly to choose. Naturally, the inverse of the salience value and/or familiarity value, $1/w(x)$, can be interpreted as the cognitive (and processing) cost of selecting $x$. Consider a non-singleton menu $S$. System 1's objective is to minimize an expected additive quadratic cost: 
\[
\underset{\rho(.,S)\in \Delta(S)}{\min}\sum_{x\in S}\frac{\rho(x,S)^2}{w(x)}.
\]
In System 1's cost minimization objective, the cost of choosing $x$ increases with its probability of being chosen, $\rho(x,S)$, and decreases with its salience or familiarity value, $w(x)$.

System 2, by contrast, is deliberate and analytical. When it is activated, it opts for the option with the highest utility with certainty. Let $x_S^*$ denote such an option in choice set $S$. System 2 cares about the frequency with which the best option is chosen and monitors the decision made by System 1 through a constraint:
\[
\underbrace{\vphantom{\frac{w(x^*_S)}{w(S)}}
\rho(x^*_S,S)}_{\text{System 2 only cares about the best option}} \ge \underbrace{\frac{w(x^*_S)}{w(S)}}_{\text{choice frequency of $x_S^*$ by System 1}} + \underbrace{\vphantom{\frac{w(x^*_S)}{w(S)}} m_S.}_{\text{the magnitude of required improvement}}
\]
This constraint acts as an overriding mechanism through which System 2 intervenes in the decision-making process if System 1's choice is not sufficiently good. In the constraint, $m_S\ge 0$ measures the degree to which System 2 wants to improve on System 1's choice regarding the selection probability of the best option in $S$. A larger value of $m_S$ indicates greater desired improvement, implying that System 1's choice is relatively unsatisfactory. Conversely, a small value of $m_S$ indicates more alignment between the two systems. If $m_S=0$ then System 1's choice is sufficiently satisfactory and System 2 need not interfere. In summary, the agent solves
\[
\underset{\rho(.,S)\in \Delta(S)}{\min}\sum_{x\in S}\frac{\rho(x,S)^2}{w(x)} \quad \text{ subject to } \quad  
\rho(x^*_S,S) \ge \frac{w(x^*_S)}{w(S)} + m_S. 
\]
The Lagrangian of the agent's cost-minimization problem is
\[
\mathcal{L}=\sum_{x\in S}\frac{\rho(x,S)^2}{w(x)}-\lambda_1\bigg(\sum_{x\in S} \rho(x,S)-1\bigg)-\lambda_2\bigg(\rho(x^*_S,S)-\frac{w(x^*_S)}{w(S)} - m_S \bigg),
\]
where $\lambda_1$ and $\lambda_2\ge 0$ are the Lagrange multipliers associated with condition $\sum_{x\in S} \rho(x,S)=1$ and System 2's constraint, respectively. Focusing on interior solutions with $\rho(x,S)\in (0,1)$ for all $x\in S$, the first-order conditions (FOCs) are
\begin{eqnarray*}
    \frac{\partial \mathcal{L}}{\partial \rho(x,S)}&=& 0 \, = \, \frac{2\rho(x,S)}{w(x)}-\lambda_1 \quad \text{for all } x\in S, x\ne x^*_S,\\
    \frac{\partial \mathcal{L}}{\partial \rho(x^*_S,S)}&=& 0 \, = \, \frac{2\rho(x^*_S,S)}{w(x^*_S)}-\lambda_1-\lambda_2,\\
    \lambda_2\bigg[\rho(x^*_S,S)-\frac{w(x^*_S)}{w(S)}- m_S\bigg] &=& 0 \\
    \sum_{x\in S} \rho(x,S)&=&1
\end{eqnarray*}
Solving these FOCs for the optimal choice yields the DST model
\[
\rho(x,S)=\alpha\cdot \mathbbm{1}(\text{$x$ is best in $S$})+(1-\alpha)\cdot \frac{w(x)}{w(S)} \quad \text{for all $x\in S$},
\]
where the degree of intervention by System 2 is given by
\begin{equation}\label{eq:alpha}
    \alpha=m_S\cdot \frac{1}{1-w(x_S^*)/w(S)}.
\end{equation}
Equation (\ref{eq:alpha}) shows how $\alpha$ naturally arises in the DST model. First, the degree of intervention by System 2 depends positively on the magnitude of improvement $m_S$. If the desired improvement is higher (lower) then System 2 needs to interfere more often (less often). Particularly, in the limiting case when no improvement is needed as System 1's decision is already sufficiently good, i.e., $m_S=0$, System 2 need not interfere, i.e., $\alpha=0$. Second, the factor $1/[1-w(x_S^*)/w(S)]$ can be interpreted as the cost of intervention by System 2. From System 2's point of view, term $1-w(x_S^*)/w(S)$ measures the amount of ``mistakes" made by System 1. When this quantity is high, it is natural to expect the intervention by System 2 to be less costly as corrections are easy; in which case, the intervention is more efficient: the same degree of intervention ($\alpha$) leads to higher improvement in choice probabilities (higher $m_S$). Conversely, if System 1 already performs well and makes few mistakes, improving upon it is inevitably costly; consequently, the intervention becomes less efficient.

We have shown above that the agent's optimal choice follows the DST model in a single decision problem. To obtain the representation for all decision problems where $\alpha$ is invariant across menus, the magnitudes of improvement must satisfy an independence property. Mathematically, equation (\ref{eq:alpha}) implies that $\alpha$ is menu-independent if and only if
\begin{equation}\label{eq:constant alpha}
\frac{m_S}{m_{S'}}= \frac{1-w(x_S^*)/w(S)}{1-w(x_{S'}^*)/w(S')} \quad \text{ for all menus $S$ and $S'$.}
\end{equation}
This property states that the relative size of desired improvements is equal to the relative size of System 1's mistakes. In decision problems where System 1 makes many (few) mistakes, the required improvements must be large (may be small) so that the degree of intervention remains the same across all choice problems. 

To conclude this section, note that when condition (\ref{eq:constant alpha}) fails to hold, in which case the magnitudes of improvement do not necessarily align with the extent of System 1's mistakes, the agent's optimal choice corresponds to a generalization of the DST model where the cognitive weight varies across menus. We formally investigate this case in Section \ref{sec:menu-dependent alpha}.

%%%%%%%%%%%%%%%%%%%%%%%%%%%%%%%%%%%%%%%%%%%%%%%%%%%%%%%%%%%%%
%%%%%%%%%%%%%%%%% Random Product Availability %%%%%%%%%%%%%%%
%%%%%%%%%%%%%%%%%%%%%%%%%%%%%%%%%%%%%%%%%%%%%%%%%%%%%%%%%%%%%

\section{Generalizations}\label{sec:extensions}

In this section, we extend the DST framework along two dimensions that arise naturally in real-world applications. First, we allow random product availability, adapting the parametric approach of \cite{Brady-Rehbeck_2016_ECMA}. Since this can be interpreted as random awareness by the agent, we can view this extension as another source of bounded rationality. We provide both characterization and identification results for this extension. Second, we consider menu-dependent cognitive weights, which greatly expand the explanatory power of the model. In Appendix \ref{sec:appendix_heterogeneity}, we also consider heterogeneity across agents in terms of preferences, cognitive capacity, and Luce weights.

%%%%%%%%%%%%%%%%%%%%%%%%%%%%%%%%%%%%%%%%%%%%%%%%%%%%%%%%%%%%%
%%%%%%%%%%%%%%%%%%%%%%%% New Section %%%%%%%%%%%%%%%%%%%%%%
%%%%%%%%%%%%%%%%%%%%%%%%%%%%%%%%%%%%%%%%%%%%%%%%%%%%%%%%%%%%%
\subsection{Random Product Availability}\label{sec:DST-PA}

In this subsection, we assume that the options available to the decision maker vary randomly and are unobserved by the analyst---capturing either {\it{random availability}} or {\it{random awareness}}. To model this randomness, we adopt the parametric framework of \cite{Brady-Rehbeck_2016_ECMA}. Each set $T \in \mathcal{D} \cup \{\emptyset\}$ is feasible with probability $\pi(T)$, where $\pi$ is a probability distribution on $\mathcal{D} \cup \{\emptyset\}$. Given a choice set $S \in \mathcal{D}$, any subset $T \subseteq S$ is available to the decision maker with probability

\[
\frac{\pi(T)}{\sum_{A\subseteq S}\pi(A)} \quad \text{ for all } T \subseteq S .
\]

To accommodate the possibility that no product is available, \cite{Brady-Rehbeck_2016_ECMA} introduces a default option $a^*$.  The decision maker chooses the default only when no option is available ($T=\emptyset$). Let $X^* = X \cup \{a^*\}$ and $S^* = S \cup \{a^*\}$ for all $S \in \mathcal{D}\cup \{\emptyset\}$. A random choice function $\rho$ in this environment is a mapping from $X^* \times \big(\mathcal{D} \cup \{\emptyset\}\big)$ to $[0,1]$ such that $\sum_{x\in S^*} \rho(x,S)=1$ for all $S\in \mathcal{D} \cup \{\emptyset\}$, $\rho(x,S)=0$ if $x \notin S^*$, and $\rho(x,S)\in[0,1]$ otherwise. Formally, our model with random product availability is defined as follows.

\begin{definition} \label{defn:DST-PA}
A random choice function has a DST representation with random product availability (DST-PA) if there exist a linear order $\succ$, a constant $\alpha \in (0,1)$, a weight function $w\colon X\to \mathbb{R}_{++}$, and a full-support probability distribution $\pi$ over $\mathcal{D}\cup \{\emptyset\}$ such that for all $x\in S \in \mathcal{D}$,
\[
\rho_{\alpha,\succ,w,\pi}(x,S)
= \sum_{\substack{T\subseteq S\\ x\in T}}
\frac{\pi(T)}{\sum_{A\subseteq S}\pi(A)}
\Bigg[\alpha\, \mathbbm{1}(\text{$x$ is $\succ$-best in $T$})+(1-\alpha)\, \frac{w(x)}{w(T)} \Bigg]
\]
and $\rho(a^*,S)=\dfrac{\pi(\emptyset)}{\sum_{A\subseteq S}\pi(A)}$ for all $S\in \mathcal{D} \cup \{\emptyset\}$.
\end{definition}

This model extends the standard DST framework to settings where the set of products actually available to the decision maker is itself uncertain. Before arriving at the store, the decision maker may not know which subset of the menu $S$ will be feasible---for instance, some items may be out of stock or simply not noticed by the consumer. Once the agent is in the store and facing $T$, she behaves according to the DST model. The uncertainty is captured by a probability distribution $\pi$ over all possible feasible sets, i.e., over $\mathcal{D}\cup \{\emptyset\}$. The analyst does not observe this realization. Consequently, the observed choice probability for an option $x$ in menu $S$ is obtained by averaging over all feasible subsets $T \subseteq S$ that contain $x$, with each subset weighted by its availability probability under $\pi$. We will show that the model gains additional explanatory power while preserving uniqueness.

\medskip
\noindent \textbf{Identification.} Suppose $\rho$ has a DST-PA representation. We show that all primitives of the model are uniquely identified. For $\pi$, note that $\displaystyle \rho(a^*,S)=\frac{\pi(\emptyset)}{\sum_{A\subseteq S}\pi(A)}$ for all $S \in \mathcal{D}\cup \{\emptyset\}$ implies $\displaystyle \frac{\pi(S)}{\pi(\emptyset)}=\sum_{A\subseteq S}\frac{(-1)^{|S\setminus A|}}{\rho(a^*,A)}$ by M\"obius inversion (see \cite{shafer1976mathematical} and Lemma \ref{lemma:Mobius} in the Appendix). Using the fact that $\pi$ constitutes a probability distribution over $\mathcal{D}\cup \{\emptyset\}$ to obtain 
\[
\pi(\emptyset)= \Bigg[\sum_{S\in \mathcal{D}\cup \{\emptyset\}} \sum_{A\subseteq S}\frac{(-1)^{|S\setminus A|}}{\rho(a^*,A)}\Bigg]^{-1} \quad \text{and} \quad \pi(S)= \pi(\emptyset) \sum_{A\subseteq S}\frac{(-1)^{|S\setminus A|}}{\rho(a^*,A)}.
\]
To identify the preference $\succ$, the deliberation weight $\alpha$, and the Luce weighting function $w$, we construct an auxiliary object $\rho^* \colon X\times \mathcal{D} \to \mathbbm{R}_{+}$ from the observed choice rule $\rho$, which we refer to as the \emph{associated random choice function}. Intuitively, $\rho^*$ captures the agent's choices as if all products in the menu were always available. We show that $\rho^*$ itself admits a DST representation, with parameters $(\alpha^*, \succ^*, w^*)$. Moreover, there exists a unique mapping between the primitives of $\rho$ and those of $\rho^*$, which allows us to recover the parameters of $\rho$ from the corresponding parameters of $\rho^*$.

To construct $\rho^*$, let $O(x,S)=\rho(x,S)/\rho(a^*,S)$ be the odds ratio of the probability of choosing $x$ over the default option $a^*$ when the menu is $S$. Mathematically,
\[
O(x,S)=\sum_{T:T\subseteq S, x\in T} \frac{\pi(T)}{\pi(\emptyset)} \Bigg[\alpha \mathbbm{1}\text{($x$ is $\succ$-best in $T$)} +(1-\alpha) \frac{w(x)}{w(T)}\Bigg].
\]
Define $\rho^*$ as
\begin{equation}\label{eq:associated random choice function}
\rho^*(x,S)=\Bigg[\underbrace{\underset{T\subseteq S}{\sum}\frac{(-1)^{|S\setminus T|}}{\rho(a^*,T)}}_\text{$\pi(S)/\pi(\emptyset)$}\Bigg]^{-1}\underset{T:T\subseteq S, x\in T}{\sum}(-1)^{|S\setminus T|}O(x,T) \quad \text{ for all $x\in S$ and $S\in \mathcal{D}$.}
\end{equation}
The \textit{associated random choice function} $\rho^*$ is constructed from $\rho$ via M\"obius inversion \citep{shafer1976mathematical}. This inversion formula provides a bijective correspondence that allows us to recover $\rho$ from $\rho^*$ and conversely. In the statement of M\"obius inversion, $\rho^*$ and $\rho$ are often referred to as the  point function and cumulative function, respectively.

Proposition \ref{prop:associated choice function} shows that the \textit{associated random choice function} $\rho^*$ constructed above admits a DST representation. Furthermore, it shares the same primitives as $\rho$.

\begin{prop}\label{prop:associated choice function}
If $\rho$ has a DST-PA representation $(\alpha,\succ,w,\pi)$, then $\rho^*$ has a DST representation $(\alpha^*,\succ^*,w^*)$. Furthermore, $(\alpha,\succ,w)=(\alpha^*,\succ^*,w^*)$ and $\pi$ is unique.
\end{prop}

Proposition \ref{prop:associated choice function} allows the researcher to uniquely identify the primitives of a DST-PA. To appreciate this uniqueness result, note that there are four moving pieces of a DST-PA: (1) the relative importance of the deliberate-thinking system $\alpha$, (2) the preference $\succ$, (3) the Luce weight function $w$, and (4) the product variation rule $\pi$. Nonetheless, Proposition \ref{prop:associated choice function} states that the model is uniquely identified. The proof of Proposition \ref{prop:associated choice function} uses M\"obius inversion.

\medskip
\noindent  \textbf{Characterization.} We show that a DST-PA is characterized by conditions that characterize DST together with one additional property. This new condition is a Block-Marschak inequality applied to the default option; it is formally stated in Axiom \ref{axiom:BM}. Note that this axiom is identical to the Increasing Feasible Odds Axiom in \cite{Brady-Rehbeck_2016_ECMA}.

\begin{axiom}[Block-Marschak inequality on the default option] \label{axiom:BM}
\[
\underset{A\subseteq S}{\sum}\frac{(-1)^{|S\setminus A|}}{\rho(a^*,A)}>0 \quad \text{ for all $S\in \mathcal{D}$.}
\]
\end{axiom}
We are now ready to state a characterization for the DST-PA model.

\begin{theorem}\label{thm:characterization DST-PA} $\rho$ has a DST-PA representation if and only if it satisfies Axiom \ref{axiom:BM} and its associated random choice function, $\rho^*$, has a DST representation.
\end{theorem}

An interesting special case of this model is the parametric model of \cite{Manzini-Mariotti_2014_ECMA}. In this model, each option $x$ is feasible with probability $\phi(x) \in (0,1)$, and the availabilities of distinct options are independent. A set $T$ is feasible when every item in $T$ is feasible and every option outside $T$ is not available. The probability of $T$ being feasible under choice set $S$ is thus $\prod_{x\in T}\phi(x) \prod_{y\in S\setminus T}(1-\phi(y))$. When no options are available, the decision maker selects the default option with certainty.

Under the independent product variation, $\rho$ satisfies Axiom \ref{axiom:mido}, which \cite{Brady-Rehbeck_2016_ECMA} call menu-independent default option. This axiom states that the relative probability of selecting the default option remains invariant across different choice sets.

\begin{axiom}[Menu independent default option] \label{axiom:mido} For all $(x,S,S')$ such that $x\in S\cap S'$ and $S,S'\in \mathcal{D}$
\[
\frac{\rho(a^*,S)}{\rho(a^*,S\setminus x)}=\frac{\rho(a^*,S')}{\rho(a^*,S'\setminus x)}.
\]
\end{axiom}

Proposition \ref{prop:DST-PA-MM} states a characterization for our model with independent product variation.

\begin{prop}\label{prop:DST-PA-MM} $\rho$ has a DST-PA representation with independent product variation if and only if it satisfies Axioms \ref{axiom:BM}-\ref{axiom:mido} and its associated random choice function, $\rho^*$, has a DST representation.
\end{prop}

%%%%%%%%%%%%%%%%%%%%%%%%%%%%%%%%%%%%%%%%%%%%%%%%%%%%%%%%%%%%%%%%%%%%%%%%%
%%%%%%%%%%%%%%%%%%%%%%%%%%% A NEW SECTION %%%%%%%%%%%%%%%%%%%%%%%%%%%%%
%%%%%%%%%%%%%%%%%%%%%%%%%%%%%%%%%%%%%%%%%%%%%%%%%%%%%%%%%%%%%%%%%%%%%%%%%
\subsection{Menu-Dependent Cognitive Weights}\label{sec:menu-dependent alpha}

In some contexts, it is natural to allow the relative importance of the two cognitive systems to depend on the menu itself. When alternatives are very similar or the menu is large, carefully evaluating every option can be costly, making the decision maker more inclined toward intuitive thinking. Conversely, when the cost of examining all available options is low, the decision maker is more likely to rely on deliberate thinking.\footnote{The same reasoning applies when the decision maker relies on memory when making choices: as the menu size increases, the probability that the decision maker can recall all options diminishes, inducing a greater tendency to randomize.}

In this subsection, we study a generalization of our model in which the cognitive weight varies across choice sets. Formally, $\rho$ has a DST representation with menu-dependent cognitive weights (d-DST) if there exists $\alpha_S \in (0,1)$ for each choice set $S$ such that\[
\rho(x,S)=\alpha_S\cdot \mathbbm{1}(\text{$x$ is $\succ$-best in $S$})+(1-\alpha_S)\cdot \frac{w(x)}{w(S)} \quad \text{for all $x\in S$}, \text{ for all $S\in \mathcal{D}$}.
\]
We call $(\{\alpha_S\}_S, \succ, w)$ a d-DST representation of $\rho$. As shown in Section \ref{sec:micro_foundation}, d-DST may arise from the solution to a constrained optimization problem. Note that in a d-DST we impose no restriction on how $\alpha_S$ may vary with the menu $S$, thereby preserving the full generality of the model.

Allowing for menu-dependent cognitive weights substantially enhances the explanatory power of our model. To elaborate, when there are three options in the grand set, Remark \ref{rem:d-DST} shows that d-DST is capable of explaining any positive random choice function except those admitting a Luce representation.

\begin{remark}\label{rem:d-DST} Suppose $X$ has three alternatives and $\rho$ is positive. Then $\rho$ admits a d-DST representation if and only if it violates IIA. 
\end{remark}

The Luce model lies outside the scope of d-DST, as shown in Remark \ref{rem:d-DST}, because a d-DST necessarily violates IIA. To see this, whenever $\rho$ has a d-DST representation with preference $x\succ y \succ z$, it follows that $\frac{\rho(y,\xyz)}{\rho(z,\xyz)} < \frac{\rho(y,\yz)}{\rho(z,\yz)}$ regardless of the cognitive weights $\{\alpha_S\}_S$. Consequently, $\rho$ cannot have a Luce representation.

While permitting the cognitive weights to vary across menus improves the model's explanatory power, it incurs a trade-off in terms of tractability and identification. Generally, a dataset may have multiple d-DST representations, complicating both identification and interpretation of parameters. To understand the challenges in recovering the model's primitives, consider a simple example with $X=\xyz$. Suppose choices from binary and tripleton menus satisfy the following conditions:
\begin{itemize}
    \item Choices from the tripleton menu are uniform: $\rho(x,\xyz)=\rho(y,\xyz)=\rho(z,\xyz)$,
    \item Choices from pairs are given by $\rho(x,\xy)=\rho(y,\yz)=\rho(z,\xz)=3/4$. 
\end{itemize}
The data $\rho$ given above violates IIA; hence, it admits a d-DST representation following Remark \ref{rem:d-DST}. In fact, this $\rho$ has at least three d-DST representations with three different preference orderings: $x\succ_1 y \succ_1 z$, $y\succ_2 z \succ_2 x$, and $z\succ_3 x \succ_3 y$.\footnote{With preference $\succ_1$, $\rho$ has one d-DST representation with the following parameters: $w(x)=1/9$, $w(y)=w(z)=4/9$, $\alpha_{\xy}=11/16$, $\alpha_{\xz}=1/16$, $\alpha_{\yz}=1/2$, and $\alpha_{\xyz}=1/4$. For preferences $\succ_2$ and $\succ_3$, one can set similar parameters to obtain a d-DST representation for $\rho$.} These preference relations form a cycle and yield conflicting orderings for every pair of alternatives. Consequently, without additional insight into the decision maker's underlying preference, the analyst cannot draw any error-free inference about the DM's preference ordering.

To facilitate identification and characterization in d-DST, we assume that the two systems are consistent. As \cite{kahneman2003maps} argues, in certain environments System 1 can be powerful and accurate, and therefore consistent with System 2, despite operating automatically and intuitively. This consistency assumption has important implications for identification and characterization in d-DST, as it allows us to uniquely recover the preference of the DM.

To elaborate, suppose $\rho$ has a consistent d-DST representation $(\{\alpha_S\}_S, \succ, w)$. Then for every choice set $S$ and for all $x,y\in S$, it is straightforward to verify that $\rho(x,S)> \rho(y,S)$ if and only if $x\succ y$. Consequently, the DM's preference can be uniquely recovered by comparing choice probabilities from the same menu. Appendix \ref{sec:Appendix-Identification} utilizes this uniqueness result to provide a characterization of d-DST. It shows that d-DST is characterized by four axioms: one axiom concerning properties of the revealed preference (completeness, transitivity, antisymmetry) and the other three related to the standard IIA property. Appendix \ref{sec:Appendix-Identification} also provides additional results on d-DST when the two systems are not consistent.

%%%%%%%%%%%%%%%%%%%%%%%%%%%%%%%%%%%%%%%%%%%%%%%%%%%%%%%%%%%%%%%%%%%%%%%
%%%%%%%%%%%%%%%%%%%%%%%%%%%%%%%%% New Section %%%%%%%%%%%%%%%%%%%%%%%%%
%%%%%%%%%%%%%%%%%%%%%%%%%%%%%%%%%%%%%%%%%%%%%%%%%%%%%%%%%%%%%%%%%%%%%%%

\section{Literature Review}\label{sec:compare}

In this section we relate DST to three prominent classes of models: (1) models that study dual-system theories within an economic framework, (2) models that generalize the Luce rule, and (3) models that are related to random utility. Comparing our model with existing models highlights the distinct mechanisms through which we depart from other frameworks and clarifies the broader landscape of stochastic choice models.

Regarding the first set of studies, \cite{ilut2023economic} and \cite{cerigioni2021dual} assume that decision makers follow dual-process thinking and investigate their choices in dynamic environments. \cite{ilut2023economic} formalize the default-interventionist theory of System 1-System 2 interaction within a structural macroeconomic framework, which differs from our decision-theoretic choice environment. Meanwhile, \cite{cerigioni2021dual} studies deterministic choice and assumes that System 2 is activated only when the decision problem at hand is sufficiently different from past problems. This paper focuses on stochastic choice; our environment is not necessarily dynamic, and System 2 intervenes to override System 1's decision in the optimization problem.

Regarding Luce-based models, there are several generalizations of the Luce rule in the literature. Recent studies include the Luce model with limited consideration (LRLC) of \cite{Ahumada_Ulku_2018}, the perception-adjusted Luce model (PALM) of \cite{Echenique_Saito_Tserenjigmid_2018}, the general Luce model (GL) of \cite{Echenique2019GeneralModel}, the threshold Luce model (TL) of \cite{horan2021stochastic}, the order-dependent Luce model (ODLM) of \cite{tserenjigmid2021order}, and the focal Luce model (FLM) of \cite{kovach2022focal}.

The LRLC, GL, and TL models extend the Luce rule to accommodate zero choice probabilities. In these frameworks, an option is chosen with zero probability if it lies outside the decision maker's consideration set (LRLC model), or is dominated by other options, where the dominance relationship is modeled by a strict partial order (GL model) or a semi-order (TL model). The decision maker follows the standard Luce rule otherwise. Our DST model features a positive stochastic choice function in which choice probabilities of the best options violate the IIA property (see Proposition \ref{prop:identification_1}). Consequently, DST and the LRLC/GL/TL models are disjoint.

The ODLM, PALM, and our model incorporate an underlying order into the standard Luce framework. Still, the ODLM/PALM and our model are both conceptually and behaviorally distinct. Conceptually, in ODLM or PALM the underlying order represents the physical or perceptual sequence through which the decision maker observes options. This ordering yields menu-dependent utility (as in ODLM) or position-dependent priority (as in PALM). By contrast, the DST model is not sequential: the underlying order represents option utilities, which are assumed to be menu-independent. Behaviorally, PALM introduces an outside option to the environment, which allows the model to explain various behavioral phenomena \cite[p.70]{Echenique_Saito_Tserenjigmid_2018}. The outside option is absent from our framework, making PALM and DST not directly comparable. Regarding ODLM, the model satisfies an IIA property that $\rho(x,\xyz)/\rho(y,\xyz)=\rho(x,\xy)/\rho(y,\xy)$ whenever the underlying order places $x$ before $y$ and $y$ before $z$. Our model, however, necessarily violates this property under any ordering among $x,y,z$ (see Proposition \ref{prop:identification_1}). Consequently, our model and ODLM are not merely distinct but entirely disjoint.

In the FLM, options that are focal receive additional boosts to their utilities. Similar to the ODLM, the binary FLM, which is one of the specifications studied in the paper, satisfies an IIA property that $\rho(x,\xyz)/\rho(y,\xyz)=\rho(x,\xy)/\rho(y,\xy)$ whenever $D_{zx,y}>0$ and $D_{zy,x}>0$.\footnote{In binary FLM, in any binary menu, the two options are assumed to be equally focal.} Again, our Proposition \ref{prop:identification_1} suggests that this property cannot occur in the DST model under any ordering among $x,y,z$. Hence, the DST model and the binary FLM are also disjoint.

Finally, regarding random utility models (RUMs), our model is a special case because it is a convex combination of preference maximization and the Luce model. Note that RUMs are not uniquely identified in general; the DST model belongs to a class of RUMs with unique identification. Three special cases of RUMs that are relevant to our framework are the dual-RUM of \cite{manzini2018dual}, the weighted linear discrete choice (WL) of \cite{chambers2025weighted}, and the Luce alignment model (LAM) of \cite{suleymanov2026revealed}. In a dual-RUM, the decision maker is endowed with two preferences and randomizes between them when making a choice. Meanwhile, in the LAM, the decision maker randomizes between two Luce rules. The dual-RUM can be interpreted as the limiting case of our DST model; this is because the Luce choice rule under automatic thinking can approximate any preference maximization under an appropriate parameterization \citep{mcfadden2000mixed}.

In the WL model, choice probability is a sum of two components: a base probability following a Luce choice rule and a comparative probability transfer measuring differences in the utilities of options. The WL and our model are distinct. The WL model satisfies moderate stochastic transitivity. Our model, as shown in Section \ref{sec:anomaly}, can violate moderate stochastic transitivity. Hence, our model is not nested in WL.

%%%%%%%%%%%%%%%%%%%%%%%%%%%%%%%%%%%%%%%%%%%%%%%%%%%%%%%%%%%%%
%%%%%%%%%%%%%%%%%%%%%%%% New Section %%%%%%%%%%%%%%%%%%%%%%
%%%%%%%%%%%%%%%%%%%%%%%%%%%%%%%%%%%%%%%%%%%%%%%%%%%%%%%%%%%%%

\section{Conclusion}\label{sec:conclusion}
We have introduced the Dual-System Thinking model to incorporate the dual-process theories in psychology into an economic framework. Our formulation has intuitive and simple behavioral and micro foundations. We show that the model enjoys several theoretically and empirically desirable properties. It is also useful for empirical estimation in discrete choice analysis and for informing firms' decision-making. Our work can be extended in several directions. One avenue is to explore alternative structures and formulations of the two systems. For instance, the two systems can be modeled by the axiomatic properties they satisfy. In our framework, System 1's choice is characterized by the Independence of Irrelevant Alternatives, whereas System 2's choice is characterized by the Weak Axiom of Revealed Preference. Imposing other properties on the choices made by the two systems may yield new insights into human decision-making processes.

\bigskip
\bigskip
\bibliographystyle{ecta}
\bibliography{references1.bib}

\clearpage
\noindent \Large{\textbf{APPENDIX}}
\normalsize
\begin{appendices}

%%%%%%%%%%%%%%%%%%%%%%%%%%%%%%%%%%%%%%%%%%%%%%%%%%%%%%%%%%%%%
%%%%%%%%%%%%%%% Alternative Objective Function %%%%%%%%%%%%%
%%%%%%%%%%%%%%%%%%%%%%%%%%%%%%%%%%%%%%%%%%%%%%%%%%%%%%%%%%%%

\section{Identification of $\alpha$ and Luce Weights in the DST Model}\label{sec:Appendix_identification_DST}

In this appendix we provide explicit formulas of $\alpha$ and $w$ in the DST model using choices from binary and tripleton menus. Take $x,y,z$ pairwise distinct and suppose the revealed preference is $x\succ_R y \succ_R z$. The identification of the revealed preference follows the method described in Section \ref{sec:identification} in the main body of the paper. The DST representation of $\rho$ implies 
\begin{eqnarray*}
    \rho(y,\xyz)=(1-\alpha)\frac{w(y)}{w(\xyz)} \quad \text{ and } \quad \rho(z,\xyz)=(1-\alpha)\frac{w(z)}{w(\xyz)},
\end{eqnarray*}
and 
\[
\rho(z,\yz)=(1-\alpha)\frac{w(z)}{w(\yz)}.
\]
It follows that $\alpha$ can be identified as
\[
\alpha=1-\frac{\rho(z,\yz)[\rho(z,\xyz)+\rho(y,\xyz)]}{\rho(z,\xyz)}.
\]
The weight for each alternative is inferred from choices in binary menus and given by
\[
\begin{aligned}
w(z) &= \Bigg[\sum_{y: y\ne z}\frac{1-\alpha}{\rho(z,\yz)}-|X|+2\Bigg]^{-1} && \text{if $z$ is the least preferred option in $X$,}\\
w(y) &= w(z)\Bigg[\frac{1-\alpha}{\rho(z,\yz)}-1\Bigg] && \text{if $z$ is the least preferred option in $X$ and $y\succ z$.}
\end{aligned}
\]
The formula for $w(z)$ uses the assumption that $\sum_{x\in X} w(x)=1$.

%%%%%%%%%%%%%%%%%%%%%%%%%%%%%%%%%%%%%%%%%%%%%%%%%%%%%%%%%%%%%%
%%%%%%%%%%%%%%%%%%%%%%% This is the new section %%%%%%%%%%%%%%
%%%%%%%%%%%%%%%%%%%%%%%%%%%%%%%%%%%%%%%%%%%%%%%%%%%%%%%%%%%%%%

\newpage
\section{Estimating the DST Model in Various Datasets}\label{sec:appendix_datasets}
We evaluate the DST model on a diverse set of choice data described in Table \ref{tab:datasets}. These datasets have different characteristics. As in Section \ref{sec:empirical_estimation} in the main body of the paper, for each dataset, we estimate (multinomial) logit, nested logit, independent probit, covariance probit, and DST models using maximum likelihood. We compare model fit using the adjusted $R^2$ and the McFadden pseudo $R^2$. Across the heterogeneous collection of datasets in Table \ref{tab:datasets}, the DST model is able to outperform other leading models in discrete choice estimation, indicating its usefulness for empirical analysis.

\begin{table}[ht!]
    \centering
    \footnotesize
    \setlength{\tabcolsep}{3pt} % Reduce space between columns (default is 6pt)
    \begin{tabular}{@{}p{3cm}|c|ccccc@{}}
    \hline
    \hline
    Data & \multicolumn{1}{c} {Data} & \multicolumn{5}{|c}{Goodness-of-fit statistics (adjusted $R^2$, McFadden pseudo $R^2$, \# of free parameters)} \\
    \cline{2-7}
    & ($|X|$, $|\mathcal{D}|$, $|S|$, $N$) & Logit &  Ind. Probit & Nested Logit & Cov. Probit & DST \\
         \hline
      Intertemporal choice   & &&&& \\
      $\quad$ RR (2000)   & (4, 6, 2, 88) & (0.45, 0.07, 3) & (0.44, 0.07, 3) & (0.87, 0.12, 5) & (0.89, 0.14, 8) & (0.95, 0.14, 4) \\
     $\quad$ MM (2006) & (4, 11, 2-4, 102) & (0.94, 0.34, 3) &  (0.93, 0.34, 3) & (0.98, 0.36, 5) & (0.99, 0.37, 8) & (0.99, 0.37, 4)  \\ 

    Riskless choice & &&&& \\
    $\quad$ LR (2015)  & (3, 4, 2-3, 40-118) & (0.76, 0.12, 2) & (0.75, 0.12, 2) & (0.66, 0.12, 3) & (0.62, 0.12, 3) & (0.85, 0.13, 3) \\
    $\quad$ (choices by frogs) & &&&& & \\
    %t\'ungara 
   Choice under risk  & &&&& & \\
    $\quad$ Caliari (2023) & (4, 11, 2-4, 145) & (0.72, 0.02, 3) & (0.72, 0.02, 3) & (0.73, 0.03, 5) & (0.71, 0.03, 8) & (0.75, 0.03, 4) \\
    $\quad$ AB (2021) & (9, 36, 2, 87) & (0.87, 0.09, 8) & (0.86, 0.09, 8) & (0.95, 0.10, 10-13) & (0.98, 0.10, 43) & (0.95, 0.10, 9) \\
    \hline
    \hline
    \end{tabular} 
    \caption{Performance of DST model on various datasets across different choice contexts}
    \label{tab:datasets}
        \vspace{0.3cm}
    \parbox{\textwidth}{\footnotesize \justifying
    \noindent \textit{Notes:} Datasets: RR (2000) - \cite{roelofsma2000intransitive}; MM (2006) - \cite{manzini2006two}; LR (2015) - \cite{lea2015irrationality}; Caliari (2023) - \cite{caliari2023behavioural}; AB (2021) - \cite{apesteguia2021separating}; . In ``Data" column ($|X|$, $|\mathcal{D}|$, $|S|$, $N$), the first, second, third, and fourth entries represent the number of options, the number of menus, the size of menus, and the number of participants, respectively. All aggregate data are positive. For the dataset in \cite{apesteguia2021separating}, the number of nests when estimating nested logit is either 2, 3, 4, or 5.}
\end{table}

%%%%%%%%%%%%%%%%%%%%%%%%%%%%%%%%%%%%%%%%%%%%%%%%%%%%%%%%%%%%%%
%%%%%%%%%%%%%%%%%%%%%%% This is the new section %%%%%%%%%%%%%%
%%%%%%%%%%%%%%%%%%%%%%%%%%%%%%%%%%%%%%%%%%%%%%%%%%%%%%%%%%%%%%

\newpage
\section{A Linear Program to Solve the Optimal List Design Problem}\label{sec:Appendix_linear_programming}
In this appendix, we develop a 0-1 linear program for algorithmically solving the optimal list design problem presented in the main body of the paper (Section \ref{sec:application}) when the platform's objective function, $f_\tau$, corresponds to the expected payoff. To formalize the programming, let the set of products be $S=[s]=\{1,2,\dots,s\}$ for some $s=|S|$. Let $A_{j,k} \equiv \mathbbm{1}_{\{l(j)=k\}}$ and $B_{j,k} \equiv \mathbbm{1}_{\{l(j)=k \text{ and } j=\argmax_{m\in S} u^l(m)\}}$ for any $j,k\in [s]$. In other words, $A_{j,k}=1$ if product $j$ is placed at position $k$ under list $l$, and $B_{j,k}=1$ if product $j$ is placed at position $k$ and is also the product with the highest perceived utility. Let
\[
C_{j,k}=\tau_j\sum_{i=1}^K\eta_i (1-\alpha_i)\frac{w_i(k)}{\sum_{k'\in [s]}w_i(k')} 
\]
be the population-weighted payoff of product $j$ under System $1$ when it is placed in the $k$th position, and $D_j=\tau_j\sum_{i=1}^K\eta_i \alpha_i$ the population-weighted payoff of product $j$ under System $2$. We assume that the platform observes the customers' choices and the distribution of customer types, so that $C_{j,k}$ and $D_j$ are observed. The platform's optimization problem can be formulated equivalently as
\[
\max_{(A_{j,k})\in \{0,1\}^{s^2} \text{ and } (B_{j,k}) \in \{0,1\}^{s^2}} \bigg \{\sum_{(j,k)\in [s]^2} C_{j,k} A_{j,k} + \sum_{(j,k)\in [s]^2} D_j B_{j,k} \bigg\}
\]
subject to the constraints
\begin{align}
  &\sum_{k\in[s]} A_{j,k} = 1 \quad\text{for all } j\in[s]\quad\text{and}\quad
   \sum_{j\in[s]} A_{j,k} = 1 \quad\text{for all } k\in[s], \tag{C1}\label{C1}\\
  &\sum_{(j,k)\in[s]^2} B_{j,k} = 1, \tag{C2}\label{C2}\\
  &B_{j,k} \le A_{j,k} \quad\text{for all } (j,k)\in[s]^2, \tag{C3}\label{C3}\\
  &B_{j,k} \ge A_{j,k}- \sum_{(m,q)\in [s]^2:\, u(m)+\beta(q)>u(j)+\beta(k)} A_{m,q} \quad\text{for all } (j,k)\in[s]^2.\tag{C4}\label{C4}
\end{align}
Constraint (\ref{C1}) ensures that the mapping between products and positions in the list is one-to-one. Constraint (\ref{C2}) follows from our assumption that the ranking among perceived utilities is strict under any list, so there exists a unique product with the highest perceived utility. Constraint (\ref{C3}) derives from the definitions of $A_{j,k}$ and $B_{j,k}$. The final constraint guarantees that the unique product $j$ with $B_{j,k}=1$ is the one with the highest perceived utility.

\newpage
%%%%%%%%%%%%%%%%%%%%%%%%%%%%%%%%%%%%%%%%%%%%%%%%%%%%%%%%%%%%%
%%%%%%%%%%%%%%%%% Rat Index %%%%%%%%%%%%%%%%%
%%%%%%%%%%%%%%%%%%%%%%%%%%%%%%%%%%%%%%%%%%%%%%%%%%%%%%%%%%%%%
\section{Alternative Measure of Stochastic Rationality}\label{sec:appendix_rat_index}
In this appendix, we discuss the relationship between the cognitive weight in our DST model and the stochastic rationality index in \cite{ok2023measuring}. In their approach, stochastic choice data are first ``approximated" by a deterministic choice correspondence. The rationality of the stochastic behavior is then evaluated via the rationality of this choice correspondence (see \cite{ribeiro2017regular} for a characterization of a rational choice correspondence). \cite{ok2023measuring} utilize the notion of Fishburn correspondences. For each random choice function $\rho$ and parameter $\lambda\in (0,1]$, define the $\lambda$-choice correspondence induced by $\rho$ as
\[
c_{\lambda,\rho}(S) = \Big\{x \in S \,|\,\rho(x,S) \ge \lambda \max_{y\in S}\rho(y,S)\Big\} \quad \text{ for all $S\in \mathcal{D}$.}
\]
Given choice set $S$, $c_{\lambda,\rho}(S)$ contains all elements whose choice probability is at least a $\lambda$-fraction of the maximum choice probability in $S$. Parameter $\lambda$ represents the degree of selectivity. A small value of $\lambda$ allows options with low choice probabilities to still qualify as valid choices from the menu (and included in $c_{\lambda,\rho}(S)$), whereas a large value of $\lambda$ treats these options as mistakes. 

Using the induced choice correspondence, \cite{ok2023measuring} define a rationality index for each RCF $\rho$ as follows
\[
I_{rationality}(\rho)= 1- \text{Leb}\big(\{\lambda \in (0, 1]: c_{\lambda,\rho} \text{ is not rational}\}\big),
\]
where Leb() is the Lebesgue measure. Roughly speaking, $I_{rationality}(\rho)$ quantifies the size of the set of $\lambda$-values for which the induced choice correspondence $c_{\lambda,\rho}$ is rational. Consequently, $\rho$ is more rational when $I_{rationality}(\rho)$ is higher. Similar to the swaps index, $I_{rationality}()$ is also a model-free measure of stochastic rationality.

In a special case of our DST model when the Luce weights are the same for all alternatives, \cite{ok2023measuring} show that $I_{rationality}$ is a U-shaped function of $\alpha$. Hence, in two DST models that differ only in the cognitive weight, $I_{rationality}$ may conclude that the one with a higher level of deliberate thinking is less rational. Example \ref{ex:rationality_index} illustrates this point further in the case in which System 1 and System 2 are inconsistent.

\begin{example}\label{ex:rationality_index} Suppose $X=\xyz.$ Let $x\succ y \succ z$. Suppose $w(x)=0.01, w(y)=0.39$, and $w(z)=0.60$, so that the ranking induced by the Luce weight is different from the preference ranking. Let $\rho$ and $\rho'$ have DST representations $(\alpha=0.55, \succ, w)$ and $(\alpha'=0.40, \succ, w)$, respectively. The rationality indexes of $\rho$ and $\rho'$ are given by $I_{rationality}(\rho)=0.52$ and $I_{rationality}(\rho')=0.57$. Hence, although the propensity to think deliberately under $\rho$ is higher than that under $\rho'$, the rationality index concludes that $\rho'$ is more rational than $\rho$. 
\end{example}

%%%%%%%%%%%%%%%%%%%%%%%%%%%%%%%%%%%%%%%%%%%%%%%%%%%%%%%%%%%%%
%%%%%%%%%%%%%%%%% Identification %%%%%%%%%%%%%%%%%
%%%%%%%%%%%%%%%%%%%%%%%%%%%%%%%%%%%%%%%%%%%%%%%%%%%%%%%%%%%%%
\bigskip
\section{Additional Results in d-DST} \label{sec:Appendix-Identification}
\subsection{Characterization of d-DST under consistency}
This section presents the characterization of d-DST under consistency. Based on the identification of the preference stated in the main body of the paper (Section \ref{sec:menu-dependent alpha}), we introduce a preference relation $\succ_{R,d^c}$ on $X$. For all $x,y \in X$, let
\begin{equation}\label{eq:revealed_preference_consistency}
  x\succ_{R,d^c} y \quad \text{ if } \quad \rho(x,S)> \rho(y,S) \quad \text{ for all $S\in \mathcal{D}$ such that $S\supseteq \xy$.}  
\end{equation}
The binary relation says that $x\succ_{R,d^c} y$ if, for all choice sets containing both $x$ and $y$, the decision maker selects $x$ with higher probability than $y$. The relation $\succ_{R,d^c}$ is the revealed preference of the DM if the DM admits a d-DST representation. The first axiom states that this revealed preference must constitute a linear order over $X$.

\begin{axiom}\label{axiom:0-dDST-Consistency} The revealed preference relation $\succ_{R,d^c}$ is complete, transitive, and asymmetric. 
\end{axiom}

The second and third axioms state that the best-ranked option in a choice set enjoys a higher relative choice probability. They capture the gain of an option when it is the best in a menu. Axiom \ref{axiom:1-dDST} considers the case in which the best option in the two menus are different, whereas Axiom \ref{axiom:1.5-dDST} considers the case in which they are the same.

\begin{axiom}\label{axiom:1-dDST} Suppose the revealed preference is $\succ_{R,d^c}$. For all $(x,y,S,S')$ such that $x$ and $y$ are distinct, $x,y\in S\cap S'$, and $S,S'\in \mathcal{D}$, and $x=b_{\succ_{R,d^c}}(S)\ne b_{\succ_{R,d^c}}(S')$,
\[
\frac{\rho(x,S)}{\rho(y,S)}> \frac{\rho(x,S')}{\rho(y,S')}.
\]
\end{axiom}

\begin{axiom}\label{axiom:1.5-dDST} Suppose the revealed preference is $\succ_{R,d^c}$. For all $(x,y,z,S,S')$ such that $x,y,z$ are pairwise distinct, $x,y\in S\cap S'$, $z\in S'$ and $S,S'\in \mathcal{D}$, and $x=b_{\succ_{R,d^c}}(S)= b_{\succ_{R,d^c}}(S')$,
\[
\frac{\rho(x,S)}{\rho(y,S)} > \frac{\rho(z,S')}{\rho(y,S')}
\]
\end{axiom}

The last axiom is a transitive version of the standard IIA axiom. It reads as follows.

\begin{axiom}[Transitive IIA]\label{axiom:2-dDST} Suppose the revealed preference is $\succ_{R,d^c}$. For all $(x,y,z,S,S',S'')$ such that $x,y,z$ are not necessarily distinct, $x,y\in S\in \mathcal{D}$, and $y,z\in S'\in \mathcal{D}$, and $x,z\in S'' \in \mathcal{D}$ and $x,y,z$ are not the best-ranked options in $S,S',S''$ according to $\succ_{R,d^c}$,
\[
\frac{\rho(x,S)}{\rho(y,S)} \frac{\rho(y,S')}{\rho(z,S')}= \frac{\rho(x,S'')}{\rho(z,S'')}.
\]
\end{axiom}
One implication of Axiom \ref{axiom:2-dDST} is the usual IIA condition that $\rho(x,S)/\rho(y,S)=\rho(x,S')/\rho(y,S')$ whenever neither $x$ nor $y$ is the best options in $S$ or $S'$. The requirement $x,y\ne b_{\succ_{R,d^c}}(S), b_{\succ_{R,d^c}}(S')$ is necessary for this IIA condition to hold because Luce's IIA breaks down when either $x$ or $y$ is the best element in a choice set.

Proposition \ref{prop:d-DST_characterization_consistency} states a characterization of d-DST under consistency.

\begin{prop} \label{prop:d-DST_characterization_consistency} $\rho$ has a consistent d-DST representation if and only if it is positive and satisfies Axioms \ref{axiom:0-dDST-Consistency}-\ref{axiom:2-dDST}.
\end{prop}

\subsection{d-DST without consistency}

When the two systems in a d-DST are inconsistent with each other, the analysis quickly becomes intractable, as shown in Section \ref{sec:menu-dependent alpha} in the main body of the paper. To keep it simple, we assume that the best option in $X$ is always observable. This assumption is justified in relevant contexts where options are naturally ranked by an ordering that reflects the DM's preferences. For instance, foods can be ordered by caloric content, cars by emissions or fuel efficiency, insurance policies by deductibles, restaurants by customer ratings, stocks by expected return, and academic journals by prestige. Note that we require only that the top option in $X$ be observed; the analyst need not observe the full ranking. This accommodates situations in which the best option is exogenously known, while the relative ordering of the remaining options remains ambiguous. For instance, when publishing a macroeconomics paper, while distinctions between journals such as the Journal of Monetary Economics and American Economic Journal: Macroeconomics may be subtle and a precise ranking between the two remains unclear, there is broad consensus that Econometrica is superior to both. 

In what follows, let $x^*$ be the observed best option in $X$. We present identification in and characterization results for d-DST under inconsistency below.

\smallskip
\noindent \textbf{Identification.} Suppose $\rho$ has a d-DST representation $(\{\alpha_S\}_S, \succ, w)$. For arbitrary $y,z\ne x^*$, the d-DST representation of $\rho$ implies that
\[
\frac{\rho(y,\yz)}{\rho(z,\yz)}>\frac{\rho(y,\{x^*,y,z\})}{\rho(z,\{x^*,y,z\})} \quad \text{ if and only if $y\succ z$}. 
\]
Consequently, the DM's preference can be uniquely recovered based on the difference in relative choice probabilities between binary and tripleton menus. The cognitive weights $\{\alpha_S\}$ and Luce weights $w$, however, are not unique in general.

\medskip
\noindent \textbf{Characterization.} Based on the identification strategy, we introduce a preference relation $\succ_{R,d}$ on $X$. For all $y\in X\setminus x^*$, let $x^*\succ_{R,d} y$. For all $y,z\ne x^*$, let
\begin{equation}\label{eq:revealed_preference_inconsistency}
y\succ_{R,d} z \quad \text{ if } \quad \frac{\rho(y,\yz)}{\rho(z,\yz)}>\frac{\rho(y,\{x^*,y,z\})}{\rho(z,\{x^*,y,z\})}.
\end{equation}
The preference relation states that $y\succ_{R,d} z$ if removing $x^*$ from the choice set benefits the relative choice probability of $y$. It is the revealed preference of the DM.

Proposition \ref{prop:d-DST_characterization} below states a characterization of d-DST under inconsistency; it is similar to the characterization of d-DST under consistency stated in Proposition \ref{prop:d-DST_characterization_consistency}. The only difference is that Axiom \ref{axiom:1.5-dDST} is no longer needed, and the revealed preference used in Axioms \ref{axiom:0-dDST-Consistency}, \ref{axiom:1-dDST}, and \ref{axiom:2-dDST} is $\succ_{R,d}$, as defined in (\ref{eq:revealed_preference_inconsistency}), instead of $\succ_{R,d^c}$, as defined in (\ref{eq:revealed_preference_consistency}).

\begin{prop}\label{prop:d-DST_characterization} $\rho$ has a d-DST representation in which $x^*$ is the most-preferred alternative if and only if it is positive and satisfies Axioms \ref{axiom:0-dDST-Consistency}, \ref{axiom:1-dDST}, and \ref{axiom:2-dDST}, with the revealed preference defined in (\ref{eq:revealed_preference_inconsistency}). 
\end{prop}

%%%%%%%%%%%%%%%%%%%%%%%%%%%%%%%%%%%%%%%%%%%%%%%%%%%%%%%%%%%%%%%%%%
%%%%%%%%%%%%%%%%% This is a new section %%%%%%%%%%%%%%%%%%%%%%%%%%
%%%%%%%%%%%%%%%%%%%%%%%%%%%%%%%%%%%%%%%%%%%%%%%%%%%%%%%%%%%%%%%%%%
\newpage
\section{DST in a Heterogeneous Population}\label{sec:appendix_heterogeneity}
In our baseline DST, we assume no variation in the model's primitives. While justifiable in a single-agent context, this assumption can be overly restrictive when the observed behavior represents choices by a population. This appendix extends our baseline DST model to allow for heterogeneity across individuals. For simplicity and tractability, we assume that $\mathcal{D}=\mathcal{X}$ in this subsection, i.e., the analyst observes choice behavior from all non-empty menus.\footnote{With limited data, the analysis quickly becomes cumbersome without offering any additional insights. Our aim is to keep the analysis simple and avoid complications when studying DST with heterogeneity.} We first define a DST model with heterogeneity as follows.

\begin{definition}[h-DST] \label{defn:heterogeneous_DST} A random choice function $\rho$ has a heterogeneous DST representation (h-DST) if there exist a collection of triples $\{(\alpha_i, \succ_i,w_i)\}_{i=1}^k$ and a probability vector $\{\eta_i\}_{i=1}^k$ on this collection such that $\rho(x,S)=\sum_{i=1}^k \eta_i \rho_{\alpha_i, \succ_i,w_i}(x,S)$ for all $x\in S$ and $S\in \mathcal{X},$ where, for each $i$, each random choice function $\rho_{\alpha_i, \succ_i,w_i}$ has a DST representation $(\alpha_i, \succ_i,w_i)$.
\end{definition}

In an h-DST, there are $k$ individuals indexed by $1,\dots, k$. Individual $i$ has a share $\eta_i\in (0,1]$ in the population, with $\sum_{i=1}^k \eta_i=1$. Each individual behaves according to a DST and has a representation $(\alpha_i, \succ_i,w_i)$. The analyst does not observe individual choice behavior, but she observes the aggregate choice in the population, which is the weighted average of the choices of all individuals. Note that besides the interpersonal interpretation, h-DST also applies to intrapersonal behavior. Under this interpretation, each $\rho_{\alpha_i, \succ_i,w_i}$ corresponds to the choice behavior of the same DM at a different time, mood, or context.

Each \textit{behavioral type} in an h-DST is represented by a triple $(\alpha_i,\succ_i, w_i)$. Accordingly, h-DST accommodates unrestricted heterogeneity both in the relative influence of the two systems and in how each system operates. In other words, individuals in an h-DST may differ in their degree of deliberation, in how deliberate thinking functions, and in how intuitive thinking randomizes. Consequently, the randomness of choices in an h-DST can originate from variability in cognitive weights, preference orderings, or Luce weights, or any combination thereof.

\smallskip

\newcommand{\PDST}{\mathrm{P}_{\mathrm{DST}}}
\newcommand{\clPDST}{\mathrm{cl}\bigl(\mathrm{P}_{\mathrm{DST}}\bigr)}
\newcommand{\clconvPDST}{\mathrm{cl}\bigl(\mathrm{conv}(\mathrm{P}_{\mathrm{DST}})\bigr)}

\newcommand{\PhDST}{\mathrm{P}_{\mathrm{h\text{-}DST}}}
\newcommand{\clPhDST}{\mathrm{cl}\bigl(\mathrm{P}_{\mathrm{h\text{-}DST}}\bigr)}

\newcommand{\Ph}{\mathrm{P}_{\mathrm{h\text{-}con\text{-}DST}}}
\newcommand{\clPh}{\mathrm{cl}\bigl(\mathrm{P}_{\mathrm{h\text{-}con\text{-}DST}}\bigr)}

\newcommand{\PRUM}{\mathrm{P}_{\mathrm{RUM}}}
\newcommand{\bdRUM}{\mathrm{bd}\bigl(\mathrm{P}_{\mathrm{RUM}}\bigr)}
\newcommand{\clRUM}{\mathrm{cl}\bigl(\mathrm{P}_{\mathrm{RUM}}\bigr)}
\newcommand{\intRUM}{\mathrm{int}\bigl(\mathrm{P}_{\mathrm{RUM}}\bigr)}

In what follows, let $\PDST$, $\PhDST$, and $\PRUM$ denote the sets of all random choice functions that admit, respectively, a DST representation, an h-DST representation, and a random utility representation. By definition, $\PhDST$ is the convex hull of $\PDST$. Proposition \ref{prop:h-DST} presents the connections between the three sets. In Proposition \ref{prop:h-DST}, for any set $A$, $\text{cl}(A)$ denotes the closure of the set $A$ and $\text{conv}(A)$ denotes its convex hull.\footnote{All topological concepts are defined on the relative topology in the affine hull of random choice functions.}

\begin{prop}\label{prop:h-DST} $\clconvPDST = \clPhDST = \PRUM$. Additionally, $\PhDST$ does not intersect the facets of $\PRUM$.
\end{prop}

Three remarks follow Proposition \ref{prop:h-DST}. First, Proposition \ref{prop:h-DST} implies that every h-DST admits a RUM representation. This result extends the Block and Marschak's finding that each Luce model is a RUM \citep{block_marschak1960}. Our result follows from the facts that (1) deterministic preference maximization and the Luce model are well-known special cases of RUMs, (2) DST is a convex combination of the two, (3) $\PhDST$ is the convex hull of $\PDST$, and (4) $\PRUM$ is a convex set (the Block-Marschak theorem). 

Second, Proposition \ref{prop:h-DST} indicates that $\PhDST$ is dense in $\PRUM$ and thus any RUM can be approximated by some h-DST to any degree of accuracy. This result comes from the facts that $\PRUM$, by definition, is the convex hull of the set of all deterministic rational choices, and any deterministic rational choice can be approximated by a DST model (by setting $\alpha$ sufficiently close to 1). The approximation of RUM by h-DST is similar to approximation results by mixed logit \citep{mcfadden2000mixed} and attribute rule models \citep{Gul_Natenzon_Pesendorfer_2014_ECMA}. 

Finally, as $\PRUM$ admits an equivalence characterization as a polytope defined by (the converse of) Block-Marschak inequalities \citep{falmagne1978representation}, it is natural to ask how $\PhDST$ ``fits" within that polytope. Proposition \ref{prop:h-DST} suggests that h-DST does not intersect the facets of this polytope but instead remains entirely within its interior. This result follows from the assumption that each cognitive system is activated with a strictly positive probability in the DST model.

\newcommand{\PAR}{\mathrm{P}_{\mathrm{AR}}}
\newcommand{\clPAR}{\mathrm{cl}\bigl(\mathrm{P}_{\mathrm{AR}}\bigr)}

\newcommand{\PLuce}{\mathrm{P}_{\mathrm{Luce}}}
\newcommand{\clconvPLuce}{\mathrm{cl}(\mathrm{conv}(\mathrm{P}_{\mathrm{Luce}}))}
\newcommand{\convPLuce}{\mathrm{conv}(\mathrm{P}_{\mathrm{Luce}})}

\textbf{Relationship between h-DST and the attribute rule models.} In the attribute rule (AR) model of \cite{Gul_Natenzon_Pesendorfer_2014_ECMA}, options possess attributes and choice proceeds in two stages. The decision maker first draws a salient (or relevant) attribute according to a logit formula and then selects an alternative that exhibits that attribute according to a second logit formula.\footnote{\cite{Gul_Natenzon_Pesendorfer_2014_ECMA} study a rich domain that differs from our setup. For comparison, the attribute rule we refer to here is defined on a finite domain: there is a finite set of options with each option having a finite set of attributes.} Conceptually, DST and the attribute rule model are distinct as there is no underlying order in the latter. Behaviorally, let $\PAR$ and $\PLuce$ denote the set of all random choice functions that admit an attribute rule and a Luce choice rule, respectively. Recall that $\PhDST$ denotes the set of all random choice functions that admit an h-DST representation. Remark \ref{rem:Attribute Rule} shows that the closures of $\PAR$, $\PhDST$, and $\convPLuce$ are all equivalent. This result holds because they are all equivalent to $\PRUM$, which follows directly from Proposition \ref{prop:h-DST} in our paper and Theorem 3 in \cite{Gul_Natenzon_Pesendorfer_2014_ECMA}.

\begin{remark} \label{rem:Attribute Rule} $\clPhDST=\clPAR=\clconvPLuce$.
\end{remark}

Even though the closures of $\PhDST$ and $\PAR$ coincide, note that $\PhDST$ and $\PAR$ are distinct. For instance, to see that $\PAR\not\subseteq \PhDST$, observe that some attribute rules lie on the facets of $\PRUM$ and therefore lie outside $\PhDST$, by Proposition \ref{prop:h-DST}. As an example, consider $X = \xyz$ and let the set of attributes be $\{a_1, a_2\}$. Suppose options $x$ and $z$ possess only attribute $a_1$, which has value $1$, while option $y$ possesses only attribute $a_2$, which also has value $1$. Let $\rho_{AR}$ be the attribute rule induced by these parameters. Then $\rho_{AR}(y, \{x,y\}) = \rho_{AR}(y, \{x,y,z\})=1/2$. It follows that $\rho_{AR}\not\in \PhDST$ because $\rho(y, \{x,y\})> \rho(y, \{x,y,z\})$ for any $\rho$ having a DST representation.

%%%%%%%%%%%%%%%%%%%%%%%%%%%%%%%%%%%%%%%%%%%%%%%%%%%%%%%%%%%%%
%%%%%%%%%%%%%%%%%%%%%%%%% Omited Proofs%%%%%%%% %%%%%%%%%%%%%
%%%%%%%%%%%%%%%%%%%%%%%%%%%%%%%%%%%%%%%%%%%%%%%%%%%%%%%%%%%%%
\newpage

\section{Omitted Proofs}

%%%%%%%%%%%%%%%%%%%%%%%%%%%%%%%%%%%%%%%%%%%%%%%%%%%%%%%%%%%%%
%%%%%%%%%%%%%%%%%%%%%%%%% New Proof %%%%%%%%%%%%%%%%%%%%%%%%
%%%%%%%%%%%%%%%%%%%%%%%%%%%%%%%%%%%%%%%%%%%%%%%%%%%%%%%%%%%%%

\subsection{Proof of Propositions \ref{prop:identification} and \ref{prop:identification_1}} 

For Proposition \ref{prop:identification_1}, suppose $\rho$ has a DST representation $(\alpha,\succ,w)$. Consider $x,y,z$ pairwise distinct such that $x\succ y \succ z$. Using the DST representation of $\rho$, it is routine to verify that $D_{xy,z} > 0, D_{yz,x} < 0$, and $D_{zx,y} < 0$. By the definition of $\succ_R$, we have $x\succ_R y$ and $y\succ_R z$ and $x\succ_R z$. Hence, $\succ_R$ and $\succ$ are identical on every triple $(x,y,z)$. It follows that $\succ_R$ and $\succ$ are identical. Consequently, $\succ_R$ represents $\rho$. 

For Proposition \ref{prop:identification}, Appendix \ref{sec:Appendix_identification_DST} shows that $\alpha$ and $w$ in the DST model are uniquely identified given a preference relation. Hence, it is sufficient to show that the preference in a DST representation is unique, which follows directly from Proposition \ref{prop:identification_1}, as the revealed preference $\succ_R$ is unique. $\blacksquare$

%%%%%%%%%%%%%%%%%%%%%%%%%%%%%%%%%%%%%%%%%%%%%%%%%%%%%%%%%%%%%%%%%%%
%%%%%%%%%%%%%%%%%%%%%%%%%% A new section %%%%%%%%%%%%%%%%%%%%%%%%%%%
%%%%%%%%%%%%%%%%%%%%%%%%%%%%%%%%%%%%%%%%%%%%%%%%%%%%%%%%%%%%%%%%%%%

\subsection{Proof of Theorem \ref{thm:characterization}}
The only-if part is straightforward. For the if part, we will first prove that $\rho$ has a representation at binary choice sets, then show that $\rho$ has a representation in larger menus. Enumerate options in $X$ as $x_1\succ_R x_2\succ_R\dots \succ_R x_n$, with $n=|X|$.

Let $D\in \mathcal{D}$ be an arbitrary choice set with at least three alternatives. Enumerate alternatives in $D$ as $y_1,\dots,y_d$ with $d=|D|\ge 3$ and $y_1\succ_R \dots \succ_R y_d$. Axiom \ref{axiom: Constant Gain} implies that there exists a constant $\gamma$ such that 
\[
A(y_d,y_i,D)= \gamma \rho(y_i,D) \quad \text{for all $i\ne 1,d$}.
\]
Let $\gamma_D$ be a constant satisfying
\[
\rho(y_d,D)=\gamma_D\rho(y_d,D\setminus y_1)\rho(D\setminus y_1,D). 
\]
Note that $\gamma_D>0$ as $\rho$ is positive (Axiom \ref{axiom:positivity}). We start with the following lemma.

\begin{lemma}\label{lema:thm-proof-1} $\gamma_D=\gamma$ for all $D\in \mathcal{D}$ with at least three alternatives. Consequently, $\gamma>0$. 
\end{lemma}
\begin{proof} By definition of $\gamma$, $A(y_d,y_i,D)= \gamma \rho(y_i,D)$ for all $i\ne 1,d$. It follows
\[
    \sum_{i=2}^{d-1} A(y_d,y_i,D) = \gamma \sum_{i=2}^{d-1} \rho(y_i,D).
\]
Additionally, we have $A(y_{d-1},y_d,D) = \gamma \rho(y_d,D)$ by Axiom \ref{axiom: Constant Gain}. Therefore, 
\begin{equation}\label{eq:thm-proof-0}
\sum_{i=2}^{d-1} A(y_d,y_i,D) + A(y_{d-1},y_d,D) = \gamma \sum_{i=2}^{d} \rho(y_i,D) = \gamma \rho(D\setminus y_1,D).
\end{equation}
By definition of $\gamma_D$
\begin{equation}\label{eq:thm-proof-0-0}
\gamma_D\rho(D\setminus y_1,D)=\frac{\rho(y_d,D)}{\rho(y_d,D\setminus y_1)}=1-A(y_d,y_1,D).
\end{equation}
From (\ref{eq:thm-proof-0}) and (\ref{eq:thm-proof-0-0})
\[
(\gamma_D-\gamma)\rho(D\setminus y_1,D) = 1-A(y_d,y_1,D)-\sum_{i=2}^{d-1} A(y_d,y_i,D) - A(y_{d-1},y_d,D)=1-\sum_{x\in D} A(l_{\succ_R}(D\setminus x),x,D)=0,
\]
where the last equation comes from Axiom \ref{axiom:5}. As $\rho$ is a positive random choice function, it follows that $\gamma_D=\gamma$. As $\gamma_D>0$, we have $\gamma>0$.
\end{proof}

Back to the main proof. Following the identification strategies, define the weight function as
\[
w(x_n) = \Bigg[\sum_{i=1}^{n-1}\frac{1}{\gamma \rho(x_n,\{x_i,x_n\})}-(n-2)\Bigg]^{-1} \text{ and } w(x_i) = w(x_n)\Bigg[\frac{1}{\gamma \rho(x_n,\{x_i,x_n\})}-1\Bigg] \text{ $\forall i<n$.}
\]
Note that by definition $w(X)=1$. The proof proceeds in four steps. We will $\gamma>1$ so there exists $\alpha\in (0,1)$ such that $\gamma=1/(1-\alpha)$ at the end of the proof. 

\medskip
\noindent \textit{\textbf{\underline{Step 1:}}} We show that $w(x_n)$ is well-defined and $w(x_i)>0$ for all $i\ge 2$. To show $w(x_n)$ is well-defined, it is sufficient to show that $\gamma \rho(x_n,\{x_i,x_n\})<1$ for all $i=2,\dots,n-1$. For any $i \in \{2, \ldots, n-1\}$, applying Lemma \ref{lema:thm-proof-1} to the menu $D = \{x_1, x_i, x_n\}$ (where $x_1 \succ_R x_i \succ_R x_n$) yields $\rho(x_n, D) = \gamma \rho(x_n, \{x_i, x_n\}) \rho(\{x_i, x_n\}, D)$. As $\rho$ is positive, $\rho(\{x_i, x_n\}, D) > \rho(x_n, D)$. It follows that $\gamma \rho(x_n, \{x_i, x_n\}) < 1$, which is equivalent to $\frac{1}{\gamma \rho(x_n, \{x_i, x_n\})} > 1$. Because there are $n-2$ such terms, their sum is strictly greater than $n-2$. Since the remaining term, $\frac{1}{\gamma \rho(x_n, \{x_1, x_n\})}$, is positive, the total bracketed expression in the definition of $w(x_n)$ is strictly positive, making $w(x_n)$ well-defined and positive. The fact that $w(x_i)>0$ for all $i \in \{2, \ldots, n-1\}$ follows directly from $\frac{1}{\gamma \rho(x_n, \{x_i, x_n\})} > 1$ shown above.

\medskip
\noindent \textit{\textbf{\underline{Step 2:}}} We show that $\rho$ has a representation at all binary choice sets. First, suppose the choice set is $\{x_i,x_n\}$, for some $i\ne n$. By definition of the weight function
\[
w(x_i) = w(x_n)\Bigg[\frac{1}{\gamma \rho(x_n,\{x_i,x_n\})}-1\Bigg] \Rightarrow \rho(x_n,\{x_i,x_n\})=\gamma^{-1}\frac{w(x_n)}{w(x_i)+w(x_n)}.
\]
It follows that $\rho(x_i,\{x_i,x_n\})=1-\gamma^{-1}+\gamma^{-1}\frac{w(x_i)}{w(x_i)+w(x_n)}.$ Now, suppose the choice set is $\{x_i,x_j\}$ for some $(i,j)$ such that $i\ne j$ and $i,j\ne n$. Without loss of generality, suppose $x_i\succ_R x_j \succ_R x_n$. Then
\begin{eqnarray}
    \gamma \rho(x_n,\{x_i,x_j,x_n\})=\frac{\rho(x_j,\{x_i,x_j\})-\rho(x_j,\{x_i,x_j,x_n\})}{\rho(x_j,\{x_i,x_j\})}&=&1-\frac{\rho(x_j,\{x_i,x_j,x_n\})}{\rho(x_j,\{x_i,x_j\})} \label{eq:thm-proof-1}\\
    \gamma \rho(x_j,\{x_i,x_j,x_n\})=\frac{\rho(x_n,\{x_i,x_n\})-\rho(x_n,\{x_i,x_j,x_n\})}{\rho(x_n,\{x_i,x_n\})}&=&1- \frac{\rho(x_n,\{x_i,x_j,x_n\})}{\rho(x_n,\{x_i,x_n\})} \label{eq:thm-proof-2}\\
    \gamma \rho(x_n,\{x_n, x_j\})\rho(\{x_n,x_j\},\{x_i,x_j,x_n\}) &=&\rho(x_n,\{x_i,x_j,x_n\}) \label{eq:thm-proof-3}
\end{eqnarray}
Equations (\ref{eq:thm-proof-1}) and (\ref{eq:thm-proof-2}) come from Axiom \ref{axiom: Constant Gain} and the definition of $\gamma$. Equation (\ref{eq:thm-proof-3}) results from Lemma \ref{lema:thm-proof-1}. Using that $\rho(x_n,\{x_j,x_n\})=\gamma^{-1}\frac{w(x_n)}{w(x_j)+w(x_n)}$ as shown above, equation (\ref{eq:thm-proof-3}) simplifies to 
\begin{equation}\label{eq:thm-proof-4}
\frac{\rho(x_j,\{x_i,x_j,x_n\})}{\rho(x_n,\{x_i,x_j,x_n\})}=\frac{w(x_j)}{w(x_n)} \Rightarrow \rho(x_j,\{x_i,x_j,x_n\})=\rho(x_n,\{x_i,x_j,x_n\})\frac{w(x_j)}{w(x_n)}.
\end{equation}
Plug (\ref{eq:thm-proof-4}) into (\ref{eq:thm-proof-2}) and using the fact that $\rho(x_n,\{x_i,x_n\})=\gamma^{-1}\frac{w(x_n)}{w(x_i)+w(x_n)}$, we have
\[
\gamma \rho(x_n,\{x_i,x_j,x_n\})\frac{w(x_j)}{w(x_n)} = 1- \frac{\rho(x_n,\{x_i,x_j,x_n\})}{\gamma^{-1}\frac{w(x_n)}{w(x_i)+w(x_n)}}
\]
which implies 
\[
\rho(x_n,\{x_i,x_j,x_n\})=  \gamma^{-1}\frac{w(x_n)}{w(x_i)+w(x_j)+w(x_n)}
\]
Plug this into (\ref{eq:thm-proof-4}) and get 
\[
\rho(x_j,\{x_i,x_j,x_n\})=  \gamma^{-1}\frac{w(x_j)}{w(x_i)+w(x_j)+w(x_n)}
\]
Finally, from (\ref{eq:thm-proof-1}) we have 
\[
\rho(x_j,\{x_i,x_j\})= \frac{\rho(x_j,\{x_i,x_j,x_n\})}{1-\gamma \rho(x_n,\{x_i,x_j,x_n\})}=\gamma^{-1}\frac{w(x_j)}{w(x_i)+w(x_j)},
\]
where the last equation uses the representations of $\rho(x_j,\{x_i,x_j,x_n\})$ and $\rho(x_n,\{x_i,x_j,x_n\})$ obtained earlier. Now, $\rho(x_i,\{x_i,x_j\})=1-\rho(x_j,\{x_i,x_j\})=1-\gamma^{-1}+\gamma^{-1}\frac{w(x_i)}{w(x_i)+w(x_j)}$. This shows that $\rho$ has the representation at all binary choice sets.

\noindent \textit{\textbf{\underline{Step 3:}}} Consider arbitrary $D\in \mathcal{D}$ with at least three alternatives. Suppose $\rho$ has the representation at all $D'$ such that $D'\subset D$ and $D'\ne D$. Note that all of these $D'$ are in $\mathcal{D}$ by richness. We will show that $\rho$ also has a representation at $D$. Suppose the elements in $D$ are $y_1,\dots,y_d$ with $d=|D|\ge 3$ and $y_1\succ_R \dots \succ_R y_d$. By Axiom \ref{axiom: Constant Gain} and definition of $\gamma$
\begin{eqnarray*}
    \gamma \rho(y_d,D)&=&A(y_{d-1}, y_d, D)=1-\frac{\rho(y_{d-1},D)}{\rho(y_{d-1},D\setminus y_d)} =1-\frac{\rho(y_{d-1},D)}{\gamma^{-1}\frac{w(y_{d-1})}{w(D\setminus y_d)}} \\
    \gamma \rho(y_{d-1},D)&=&A(y_d, y_{d-1}, D)=1-\frac{\rho(y_d,D)}{\rho(y_d,D\setminus y_{d-1})}=1-\frac{\rho(y_d,D)}{\gamma^{-1}\frac{w(y_d)}{w(D\setminus y_{d-1})}}, 
\end{eqnarray*}
where the last equation in each line comes from the representations of $\rho$ at $D\setminus y_{d-1}$ and $D\setminus y_d$ by induction hypothesis. Using these two equations above to solve for $\rho(y_d,D)$, noting that $w(D)=\sum_{y\in D}w(y)$, we have
\[
\gamma \rho(y_d,D)=1-\frac{\gamma^{-1}\bigg[1-\frac{\rho(y_d,D)}{\gamma^{-1}\frac{w(y_d)}{w(D\setminus y_{d-1})}}\bigg]}{\gamma^{-1}\frac{w(y_{d-1})}{w(D\setminus y_d)}} \Rightarrow \rho(y_d,D)=\gamma^{-1}\frac{w(y_d)}{w(D)}.
\]
Now 
\[
\gamma \rho(y_i,D)=A(y_d, y_i, D)=1-\frac{\rho(y_d,D)}{\rho(y_d,D\setminus y_i)} =1-\frac{\rho(y_d,D)}{\gamma^{-1}\frac{w(y_d)}{w(D\setminus y_i)}} \quad \text{ for all $i=2,\dots,d-1$}.
\]
The first equation comes from Axiom \ref{axiom: Constant Gain} and definition of $\gamma$. The second equation uses the representation of $\rho$ at $D\setminus y_i$ by induction hypothesis. Plug the representation of $\rho(y_d,D\setminus y_i)$ obtained earlier into the equation above to get
\[
\rho(y_i,D)=\gamma^{-1}\frac{w(y_i)}{w(D)}\quad \text{ for all $i=2,\dots,d-1$}.
\]
Now, $\rho(y_1,D)=1-\sum_{i=2}^{d}\rho(y_i,D)=1-\gamma^{-1}+\gamma^{-1}\frac{w(y_1)}{w(D)}$. Hence, $\rho$ has the representation at $D$. By induction, $\rho$ has the representation at all $D\in \mathcal{D}$.

\medskip
\noindent \textbf{\underline{\textit{Step 4:}}} We show that $\gamma=\frac{1}{1-\alpha}$ for some $\alpha\in (0,1)$ and $w(x_1)>0$. Note that $x_1\succ_R x_2\succ_R x_3$. Using the representation of $\rho$ shown in Steps 2-3
\begin{eqnarray*}
D_{x_1x_2,x_3}&=&\frac{\rho(x_1,\{x_1,x_2,x_3\})}{\rho(x_2,\{x_1,x_2,x_3\})}-\frac{\rho(x_1,\{x_1,x_2\})}{\rho(x_2,\{x_1,x_2\})}\\&=&\frac{1-\gamma^{-1}}{\gamma^{-1}}\Bigg[\frac{w(x_1)+w(x_2)+w(x_3)}{w(x_2)}-\frac{w(x_1)+w(x_2)}{w(x_2)}\Bigg] =\frac{1-\gamma^{-1}}{\gamma^{-1}}\frac{w(x_3)}{w(x_2)}
\end{eqnarray*}
By definition of $\succ_R$, $x_1\succ_R x_2\succ_R x_3$ implies that $D_{x_1x_2,x_3}>0$. Lemma \ref{lema:thm-proof-1} says that $\gamma>0$. By the representation of $\rho$ in Steps 1-2 and the fact that $\rho$ is positive (Axiom \ref{axiom:positivity}),
\[
\frac{w(x_3)}{w(x_2)}=\frac{\rho(x_3,\{x_1,x_2,x_3\})}{\rho(x_2,\{x_1,x_2,x_3\})}>0.
\]
Hence, $1-\gamma^{-1}$ must be positive. Let $1-\gamma^{-1}=\alpha>0$. Note that $1-\gamma^{-1}<1$ as $\gamma>0$. It follows $\alpha\in (0,1)$. To show that $w(x_1)>0$, by strict regularity (Axiom \ref{ax:Strict_regularity}),
\[
\rho(x_3,\{x_1,x_2,x_3\})=\gamma^{-1}\frac{w(x_3)}{w(x_1)+w(x_2)+w(x_3)} <\gamma^{-1}\frac{w(x_3)}{w(x_2)+w(x_3)}=\rho(x_3,\{x_2,x_3\}),
\]
which implies that $w(x_1)>0$, as $w(x_2)>0$ and $w(x_3)>0$ shown in Step 1. This completes the proof. $\blacksquare$

%%%%%%%%%%%%%%%%%%%%%%%%%%%%%%%%%%%%%%%%%%%%%%%%%%%%%%%%%%%%%
%%%%%%%%%%%%%%%%%%%%%%%%% New Proof %%%%%%%%%%%%%%%%%%%%%%%%
%%%%%%%%%%%%%%%%%%%%%%%%%%%%%%%%%%%%%%%%%%%%%%%%%%%%%%%%%%%%%

\subsection{Proof of Proposition \ref{prop:characterization-Luce}}
The necessary part is straightforward. Suppose $\rho$ has a Luce representation with $w\colon X\to \mathbb{R}_{++}$ being the weight function. Then 
\[
A(x,y,S)=\frac{\rho(x,S\setminus y)-\rho(x,S)}{\rho(x,S\setminus y)}=1-\frac{\rho(x,S)}{\rho(x,S\setminus y)}=1-\frac{w(S\setminus y)}{w(S)}=\frac{w(y)}{w(S)}.
\]
Fix an arbitrary ranking $\succ$ over $X$. Suppose $S=\{x_1,x_2,\dots,x_K\}$ such that $x_1\succ x_2\succ\dots\succ x_K$. Then there are $|S|=K$ ordered pairs $(x,y)$ such that $x=l_{\succ}(S\setminus y)$. Those pairs are $(x_K,x_i)$ for $i=1,2,\dots,K-1$ and $(x_{K-1},x_K)$. Hence, 
\[
\sum_{(x,y): x=l_{\succ}(S\setminus y)} A(x,y,S) = \sum_{i=1}^{K-1} A(x_K,x_i,S) + A(x_{K-1},x_K,S)= \sum_{i=1}^{K-1} \frac{w(x_i)}{w(S)} + \frac{w(x_K)}{w(S)} =  \sum_{i=1}^{K} \frac{w(x_i)}{w(S)}=1.
\]

For the sufficiency, suppose that for each arbitrary linear order $\succ$ over alternatives in $X$, we have $\displaystyle \sum_{(x,y): x=l_{\succ}(S\setminus y)} A(x,y,S)=1$ for all choice set $S$ with at least three alternatives. Suppose $S=\{x_1,x_2,\dots,x_K\}$. We will show 
\[
\frac{\rho(x_i,S\setminus x_t)}{\rho(x_i,S)}=\frac{\rho(x_j,S\setminus x_t)}{\rho(x_j,S)} \quad \text{ for all $(i,j,t)$ pairwise distinct.}
\]
Fix $(i,j,t)$. Consider two rankings $\succ_1$ and $\succ_2$ such that
\[
\begin{cases}
    x_t =l_{\succ_1}(S)$ and $x_i= l_{\succ_1}(S\setminus x_t) \\
    x_t =l_{\succ_2}(S)$ and $x_j= l_{\succ_2}(S\setminus x_t)\\
\end{cases}
.
\]
By definition, the worst option in $S$ is the same in rankings 1 and 2, but the second-worst option is different. Under ranking $\succ_1$, the ordered pairs $(x,y)$ such that $x=l_{\succ_1}(S\setminus y)$ are $(x_t,x_p)$ for $p=1,2,\dots,K$ and $p\ne t$ and $(x_i,x_t)$. Using Axiom \ref{axiom:Luce-2}
\[
1=\sum_{(x,y): x=l_{\succ_1}(S\setminus y)} A(x,y,S)= \sum_{p=1, p\ne t}^{K} A(x_t,x_p,S) + A(x_i,x_t,S).
\]
Under ranking $\succ_2$, the ordered pairs $(x,y)$ such that $x=l_{\succ_2}(S\setminus y)$ are $(x_t,x_q)$ for $q=1,2,\dots,K$, $q\ne t$ and $(x_j,x_t)$. Using Axiom \ref{axiom:Luce-2}
\[
1=\sum_{(x,y): x=l_{\succ_2}(S\setminus y)} A(x,y,S)= \sum_{q=1, q\ne t}^{K} A(x_t,x_q,S) + A(x_j,x_t,S).
\]
Hence, it follows that $A(x_i,x_t,S)=A(x_j,x_t,S)$. Using the definition of $A(x,y,S)$, the above equation is equivalent to
\[
\frac{\rho(x_i,S\setminus x_t)-\rho(x_i,S)}{\rho(x_i,S\setminus x_t)}=\frac{\rho(x_j,S\setminus x_t)-\rho(x_j,S)}{\rho(x_j,S\setminus x_t)} \Rightarrow \frac{\rho(x_i,S)}{\rho(x_i,S\setminus x_t)}=\frac{\rho(x_j,S)}{\rho(x_j,S\setminus x_t)}
\]
As $i,j,t$ are chosen arbitrarily, for all $x,y,z\in S$ pairwise distinct, we have 
\[
\frac{\rho(x,S)}{\rho(x,S\setminus z)}=\frac{\rho(y,S)}{\rho(y,S\setminus z)}.
\]
Since the equation above holds for all choice sets $S$ with at least three alternatives, it follows that $\rho$ satisfies IIA and thus has a Luce representation. $\blacksquare$

%%%%%%%%%%%%%%%%%%%%%%%%%%%%%%%%%%%%%%%%%%%%%%%%%%%%%%%%%%%%%%%%%%%
%%%%%%%%%%%%%%%%%%%%%%%%%% A new section %%%%%%%%%%%%%%%%%%%%%%%%%%%
%%%%%%%%%%%%%%%%%%%%%%%%%%%%%%%%%%%%%%%%%%%%%%%%%%%%%%%%%%%%%%%%%%%

\subsection{Proof of Proposition \ref{prop:comparative static-1}}

\noindent \underline{Part (i):} Suppose $\rho$ FOSD$_{\succ}$ $\rho'$. Consider a binary menu $\xy$ with $x\succ y$. By definition, $\rho$ FOSD$_{\succ}$ $\rho'$ implies
\[
\rho(x,\xy)=\alpha+(1-\alpha)\frac{w(x)}{w(x)+w(y)}\ge \rho'(x,\xy)=\alpha'+(1-\alpha')\frac{w(x)}{w(x)+w(y)},
\]
which is equivalent to $(\alpha-\alpha')\frac{w(y)}{w(x)+w(y)}\ge 0$. Consequently, $\alpha\ge \alpha'$. Now, suppose $\alpha\ge \alpha'$. Let $S\in \mathcal{D}$ be an arbitrary menu. Enumerate elements in $S$ as $x_1\succ \dots \succ x_s$, where $s=|S|$. We want to show that for all $i=1,\dots, s$
\begin{eqnarray*}
\sum_{j=1}^i \rho(x_j,S) \ge \sum_{j=1}^i \rho'(x_j,S) \Leftrightarrow  \alpha + (1-\alpha) \frac{\sum_{j=1}^i w(x_j)}{w(S)} &\ge& \alpha' + (1-\alpha') \frac{\sum_{j=1}^i w(x_j)}{w(S)} \\
\Leftrightarrow (\alpha-\alpha') \bigg[1-\frac{\sum_{j=1}^i w(x_j)}{w(S)} \bigg] &\ge& 0,
\end{eqnarray*}
which holds because $\alpha\ge \alpha'$ and $w(S)\ge \sum_{j=1}^i w(x_j)$.

\noindent \underline{Part (ii):} Suppose $\rho$ FOSD$_{\succ}$ $\rho'$. Consider an arbitrary binary menu $\xy$. Suppose $x\succ y$. By definition, $\rho$ FOSD$_{\succ}$ $\rho'$ implies
\[
\rho(x,\xy)=\alpha+(1-\alpha)\frac{w(x)}{w(x)+w(y)}\ge \rho'(x,\xy)=\alpha+(1-\alpha)\frac{w'(x)}{w'(x)+w'(y)},
\]
which is equivalent to $\frac{w(x)}{w(x)+w(y)}\ge \frac{w'(x)}{w'(x)+w'(y)}$ or $\frac{w(x)}{w(y)}\ge \frac{w'(x)}{w'(y)}$. As $x$ and $y$ are chosen arbitrarily, it follows that $\frac{w(x)}{w(y)}\ge \frac{w'(x)}{w'(y)}$ whenever $x\succ y$. Now, suppose $\frac{w(x)}{w(y)}\ge \frac{w'(x)}{w'(y)}$ whenever $x\succ y$. Let $S\in \mathcal{D}$ be an arbitrary menu. Enumerate elements in $S$ as $x_1\succ \dots \succ x_s$, where $s=|S|$. We want to show that for all $i=1,\dots, s$
\begin{eqnarray*}
\sum_{j=1}^i \rho(x_j,S) \ge \sum_{j=1}^i \rho'(x_j,S) \Leftrightarrow  \alpha + (1-\alpha) \frac{\sum_{j=1}^i w(x_j)}{w(S)} &\ge& \alpha + (1-\alpha) \frac{\sum_{j=1}^i w'(x_j)}{w'(S)} \\
\Leftrightarrow (1-\alpha) \bigg[\frac{\sum_{j=1}^i w(x_j)}{w(S)}-\frac{\sum_{j=1}^i w'(x_j)}{w'(S)} \bigg] &\ge& 0.
\end{eqnarray*}
As $1>\alpha$, it is sufficient to show that the difference inside the square bracket is nonnegative. This is equivalent to
\begin{eqnarray*}
    \frac{\sum_{j=1}^i w(x_j)}{w(S)}&\ge& \frac{\sum_{j=1}^i w'(x_j)}{w'(S)}  \\
w'(S) \sum_{j=1}^i w(x_j)&\ge& w(S) \sum_{j=1}^i w'(x_j) \\
    \bigg[\sum_{k=1}^i w'(x_k)+\sum_{k=i+1}^s w'(x_k)\bigg] \sum_{j=1}^i w(x_j)&\ge& \bigg[ \sum_{k=1}^i w(x_k)+\sum_{k=i+1}^s w(x_k) \bigg]\sum_{j=1}^i w'(x_j) \\
        \sum_{k=1}^i w'(x_k)\sum_{j=1}^i w(x_j) +\sum_{k=i+1}^s w'(x_k) \sum_{j=1}^i w(x_j)&\ge& \sum_{k=1}^i w(x_k)\sum_{j=1}^i w'(x_j) + \sum_{k=i+1}^s w(x_k) \sum_{j=1}^i w'(x_j) \\
        \sum_{k=i+1}^s w'(x_k) \sum_{j=1}^i w(x_j)&\ge& \sum_{k=i+1}^s w(x_k) \sum_{j=1}^i w'(x_j) \\
        \sum_{j=1}^i \sum_{k=i+1}^s w(x_j) w'(x_k)&\ge&  \sum_{j=1}^i \sum_{k=i+1}^s w'(x_j)w(x_k) \\
        \sum_{j=1}^i \sum_{k=i+1}^s [w(x_j) w'(x_k)-w'(x_j)w(x_k)] &\ge& 0.
\end{eqnarray*}
Whenever $j<k$, we have $x_j\succ x_k$ and it follows that $\frac{w(x_j)}{w'(x_j)} \ge \frac{w(x_k)}{w'(x_k)}$. Consequently, $w(x_j) w'(x_k)-w'(x_j)w(x_k)\ge 0$. Therefore, $\sum_{j=1}^i \sum_{k=i+1}^s [w(x_j) w'(x_k)-w'(x_j)w(x_k)] \ge 0$. This completes our proof of the Proposition. $\blacksquare$

%%%%%%%%%%%%%%%%%%%%%%%%%%%%%%%%%%%%%%%%%%%%%%%%%%%%%%%%%%%%%%%%%%%
%%%%%%%%%%%%%%%%%%%%%%%%%% A new section %%%%%%%%%%%%%%%%%%%%%%%%%%%
%%%%%%%%%%%%%%%%%%%%%%%%%%%%%%%%%%%%%%%%%%%%%%%%%%%%%%%%%%%%%%%%%%%

\subsection{Proof of Proposition \ref{prop:comparative static-2}}
Suppose $\rho$ FOSD$_{\succ_b}$ $\rho'$. Take arbitrary $x,y$ such that $x\succ_b y$ and $x\succ' y$. Consider binary menu $\xy$. By definition, $\rho$ FOSD$_{\succ_b}$ $\rho'$ implies
\[
\rho(x,\xy)\ge \rho'(x,\xy)=\alpha+(1-\alpha)\frac{w(x)}{w(x)+w(y)}> (1-\alpha)\frac{w(x)}{w(x)+w(y)}.
\]
Consequently, it must be the case that $x\succ y$ because $\rho(x,\xy)=(1-\alpha)\frac{w(x)}{w(x)+w(y)}$ if $y\succ x$. Hence, for arbitrary $x,y$ such that $x\succ_b y$ and $x\succ' y$, we have $x\succ y$. Consequently, $\{\succ',\succ\}$ is single-crossing with respect to $\succ_b$.

Now, suppose $\{\succ',\succ\}$ is single-crossing with respect to $\succ_b$. Let $S\in \mathcal{D}$ be an arbitrary menu. Enumerate elements in $S$ as $x_1\succ_b \dots \succ_b x_s$, where $s=|S|$. We want to show that for all $i=1,\dots, s$
\begin{equation}\label{eq:prop-1}
\sum_{j=1}^i \rho(x_j,S) \ge \sum_{j=1}^i \rho'(x_j,S).
\end{equation}
Let $x_p$ be the $\succ$-best in $S$ and $x_q$ the $\succ'$-best in $S$, for some $p,q\in \{1,\dots,s\}$. If $q<p$ then we have $x_q\succ_b x_p$ (by enumeration) and $x_q\succ' x_p$ (by definition of $x_q$). The single-crossing property then implies that $x_q\succ x_p$, which is a contradiction to the definition of $x_p$. Therefore, $q\ge p$. If $p = q$, the two choice distributions are identical, and equation (\ref{eq:prop-1}) holds trivially with equality. When $p\ne q$, to show (\ref{eq:prop-1}), note that $\rho(x_j,S)=\rho'(x_j,S)=(1-\alpha)\frac{w(x_j)}{w(S)}$ for all $j\ne p,q$ as $x_j$ is neither $\succ$-best nor $\succ'$-best in $S$. Additionally, $\rho(x_p,S)-\rho'(x_p,S)=\rho'(x_q,S)-\rho(x_q,S)= \alpha$, which also implies $\rho(x_p,S)+\rho(x_q,S)=\rho'(x_p,S)+\rho'(x_q,S)$. From here, consider the following two cases.

\smallskip
\textbf{\textit{\underline{Case 1:}}} $i\ge q$ or $i< p$. If $i<p$, then $p,q\not \in \{1,.,i\}$. Consequently, $\sum_{j=1}^i \rho(x_j,S)=\sum_{j=1}^i \rho'(x_j,S)$ because $\rho(x_j,S)=\rho'(x_j,S)$ for all $j\ne p,q$. When $i\ge q$, then $p,q \in \{1,.,i\}$.  Consequently, $\sum_{j=1}^i \rho(x_j,S)=\sum_{j=1}^i \rho'(x_j,S)$ because (1) $\rho(x_j,S)=\rho'(x_j,S)$ for all $j\ne p,q$, and (2) $\rho(x_p,S)+\rho(x_q,S)=\rho'(x_p,S)+\rho'(x_q,S)$ shown above.

\smallskip
\textbf{\textit{\underline{Case 2:}}} $q>i\ge p$. Then $p\in \{1,.,i\}$ but $q\not \in \{1,.,i\}$. In this case, 
\begin{eqnarray*}
\sum_{j=1}^i \rho(x_j,S)-\sum_{j=1}^i \rho'(x_j,S)&=&\sum_{j=1}^i [\rho(x_j,S)- \rho'(x_j,S)]\\
&=&\rho(x_p,S)- \rho'(x_p,S)+ \sum_{j=1, j\ne p}^i [\rho(x_j,S)- \rho'(x_j,S)] =\alpha>0.
\end{eqnarray*}
This completes our proof of the Proposition. $\blacksquare$

%%%%%%%%%%%%%%%%%%%%%%%%%%%%%%%%%%%%%%%%%%%%%%%%%%%%%%%%%%%%%%%%%%%
%%%%%%%%%%%%%%%%%%%%%%%%%% A new section %%%%%%%%%%%%%%%%%%%%%%%%%%%
%%%%%%%%%%%%%%%%%%%%%%%%%%%%%%%%%%%%%%%%%%%%%%%%%%%%%%%%%%%%%%%%%%%
\subsection{Proof of Remark \ref{rem:SST_satisfied}} Suppose $\rho$ has a DST representation $(\alpha,\succ,w)$ with $(\succ,w)$ satisfying Consistency. Take $x,y,z$ pairwise distinct such that $\rho(x,\xy)\ge 1/2$ and $\rho(y,\yz)\ge 1/2$. As $(\succ,w)$ satisfies Consistency, $\rho(x,\xy)\ge 1/2$ implies $x\succ y$ and $\rho(y,\yz)\ge 1/2$ implies $y\succ z$. Hence, $x\succ y\succ z$. Therefore,  
\[
\rho(x,\xz)=\alpha + \frac{(1-\alpha)w(x)}{w(x)+w(z)}; \rho(x,\xy)=\alpha + \frac{(1-\alpha)w(x)}{w(x)+w(y)}; \rho(y,\yz)=\alpha + \frac{(1-\alpha)w(y)}{w(y)+w(z)}.
\]
As $(\succ,w)$ satisfies Consistency, $x\succ y\succ z$ implies that $w(x)>w(y)>w(z)$. Using $w(x)>w(y)>w(z)$, it is straightforward to verify that $\rho(x,\xz)>\rho(x,\xy)$ and $\rho(x,\xz)>\rho(y,\yz)$. Therefore, $\rho(x,\xz)>\max \{\rho(x,\xy),\rho(y,\yz)\}$, which implies that $\rho$ satisfies strong stochastic transitivity. $\blacksquare$

%%%%%%%%%%%%%%%%%%%%%%%%%%%%%%%%%%%%%%%%%%%%%%%%%%%%%%%%%%%%%%%%%%%
%%%%%%%%%%%%%%%%%%%%%%%%%% Prof of Proposition %%%%%%%%%%%%%%%%%%%%
%%%%%%%%%%%%%%%%%%%%%%%%%%%%%%%%%%%%%%%%%%%%%%%%%%%%%%%%%%%%%%%%%%%
\subsection{Proof of Proposition \ref{prop:optimal_list_1}} Slightly abusing the notation, for arbitrary list $l$, let $\rho^l_{i,1}$ and $\rho^l_{i,2}$ denote the choice of customer $i$ under System $1$ and System $2$, respectively.

\noindent \textbf{\underline{Part (i):}} Let $x_a$ be the product with highest perceived utility under the optimal list $l^*$. Suppose $l^*(x_a)=m$. We first show $m\le a$. Proof by contradiction. Suppose $l^*(x_a)=m>a$. This implies that there exists product $x_j$ with $j>a$ such that $l^*(x_j)<l^*(x_a)$. Swap the positions of products $x_j$ and $x_a$ to obtain a new list $l'$. We prove $\rho^{l'}_i \, \, \text{FOSD}_{\succ_\tau} \, \, \rho^{l^*}_i$ by showing (1) $\rho^{l^*}_i(x_a,S)< \rho^{l'}_i(x_a,S)$, (2) $\rho^{l^*}_i(x_j,S)> \rho^{l'}_i(x_j,S)$, and (3) $\rho^{l^*}(x_t,S) = \rho^{l'}(x_t,S)$ for all $t\ne j,a$. 

As the position of product $x_t$ with $t\ne a,j$ remains unchanged, and product $x_a$ is listed earlier in $l'$ than in $l^*$, $x_a$ remains to be the product with highest perceived utility in the new list. Thus,
\begin{eqnarray*}
    \rho^{l^*}_i(x_a,S)=\alpha_i \rho^{l^*}_{i,2} (x_a,S) + (1-\alpha_i)\rho^{l^*}_{i,1}(x_a,S) &=& \alpha_i + (1-\alpha_i)\rho^{l^*}_{i,1}(x_a,S) \\
    &<& \alpha_i + (1-\alpha_i)\rho^{l'}_{i,1}(x_a,S) \\
    &=& \alpha_i \rho^{l'}_{i,2}(x_a,S) + (1-\alpha_i)\rho^{l'}_{i,1}(x_a,S) = \rho^{l'}_i(x_a,S).
\end{eqnarray*}
The inequality comes from the fact that $x_a$ is listed earlier in $l'$ compared to in $l^*$; therefore, its salience value is higher and its choice probability is higher in $l'$ when System 1 operates. By using a similar logic, we have $\rho^{l^*}_i(x_j,S)> \rho^{l'}_i(x_j,S)$ (because choice of $x_j$ comes from System 1 only and $x_j$ is listed earlier in $l^*$ than in $l'$) and $\rho^{l^*}_i(x_t,S) = \rho^{l'}_i(x_t,S)$ for all $t\ne j,a$ (because choice of $x_t$ comes from System 1 only and its position is unchanged in $l'$). It follows that $\rho^{l'}_i \, \, \text{FOSD}_{\succ_\tau} \, \, \rho^{l^*}_i$. As this happens to every customer $i$, and $\rho^{l^*}$ and $\rho^{l'}$ are population averages, we have $\rho^{l'} \, \, \text{FOSD}_{\succ_\tau} \, \, \rho^{l^*}$. Consequently, $f_\tau(\rho^{l'})>f_\tau(\rho^{l^*})$ because $f_\tau$ satisfies monotonicity. This implies that the platform's utility strictly increases under the new list, contradicting the optimality of $l^*$. Hence, $l^*(x_a)=m\le a$.

Now, we show that $l^*(x_i)=i$ for all $i$ such that $1\le i\le m-1$. There is nothing to show when $m=1$, i.e., when product $x_a$ is placed at the top of the optimal list. Suppose $m\ge 2$, which implies $a\ge 2$ as $a\ge m$. We prove the claim for $i=1$; the other values of $i$ are similar. Proof by contradiction. Suppose that $m\ge 2$ and $l^*(x_1)> 1=l^*(x_k)$ for some value of $k>1$. Swap the positions of $x_1$ and $x_k$ to obtain a new list $l''$. Note that under the new list, the product with the highest perceived utility is either unchanged or becomes $x_1$. This is because swapping $x_1$ and $x_k$ increases the perceived utility of $x_1$, decreases that of $x_k$, and leaves the perceived utility of every other $x_t$ ($t\ne 1,k$) unchanged. Hence, 
\begin{align*}
\rho^{l^*}_i(x_1,S)&<\rho^{l''}_i(x_1,S);\quad
\rho^{l^*}_i(x_a,S)\ge\rho^{l''}_i(x_a,S)\\
\rho^{l^*}_i(x_k,S)&>\rho^{l''}_i(x_k,S);\quad
\rho^{l^*}_i(x_t,S)=\rho^{l''}_i(x_t,S)\quad\text{for all }t\neq 1,a,k.
\end{align*}
It follows that $\rho^{l''}_i \, \, \text{FOSD}_{\succ_\tau} \, \, \rho^{l^*}_i$. Consequently, $\rho^{l''} \, \, \text{FOSD}_{\succ_\tau} \, \, \rho^{l^*}$. It follows that $f_\tau(\rho^{l''})>f_\tau(\rho^{l^*})$. Hence, the new list generates higher utility for the platform, contradicting the optimality of $l^*$. Therefore, if $m\ge 2$ then $l^*(x_1)=1$. 

Showing that $l^*(x_i)=i+1$ for all $i$ such that $m\le i \le a-1$ is similar (by strong induction) so we omit a formal proof. 

\smallskip
\noindent \textbf{\underline{Part (ii):}} First we show that $l^*(x_h)\ge m$. Note that if $l^*(x_h)<m=l^*(x_a)$ then $x_h$ and $x_a$ are different products and $x_h$ is listed earlier than $x_a$. In this case, $x_a$ cannot be the product with the highest perceived utility because the observed utility of $x_h$ is higher than that of $x_a$. This is a contradiction to the definition of $x_a$. Therefore, $l^*(x_h)\ge m$. 

Now suppose there exist products $x_b$ and $x_c$ such that $l^*(x_b)>l^*(x_c)> l^*(x_h)$ and $b<c$ (implying $\tau_{x_b}>\tau_{x_c}$). Swap the positions of $x_b$ and $x_c$ to obtain a new list $l'''$. Note that the swap does not change the product with the highest perceived utility. This is because the perceived utility of a product listed from $x_h$ onward is at most that of $x_h$, which is weakly less than that of $x_a$ by definition of $x_a$. Hence, using logic similar to that in the proof of part (i), we have
\begin{align*}
\rho^{l^*}_i(x_b,S)&<\rho^{l'''}_i(x_b,S);\quad \rho^{l^*}_i(x_c,S)>\rho^{l'''}_i(x_c,S);\quad  
\rho^{l^*}_i(x_t,S)=\rho^{l'''}_i(x_t,S)\quad\text{for all }t\neq b,c.
\end{align*}
Consequently, $\rho_i^{l'''} \, \, \text{FOSD}_{\succ_\tau} \, \, \rho_i^{l^*}$ and it follows $\rho^{l'''} \, \, \text{FOSD}_{\succ_\tau} \, \, \rho^{l^*}$. Hence, $f_\tau(\rho^{l'''})>f_\tau(\rho^{l^*})$. This implies that the new list generates higher utility for the platform, contradicting the optimality of $l^*$. Therefore, $l^*(x)<l^*(y)$ whenever $\tau_x>\tau_y$ and $\min\{l^*(x),l^*(y)\} > l^*(x_h)$. 

The last thing to show is that $a\le h$. Suppose not, so that $a>h$. Then
$\tau_{x_h}>\tau_{x_a}$. By the previous argument,
$l^*(x_h)\ge m=l^*(x_a)$; since $x_h\neq x_a$, we have
$l^*(x_h)>m$. Swap the positions of
$x_h$ and $x_a$. The same swap argument as above applies: under the new
list, $x_h$ becomes the perceived-utility maximizer, its System-1
salience increases, and the salience of $x_a$ decreases, while the salience of all other
products is unchanged. Thus, for every customer $i$, choice probability
is shifted from $x_a$ to $x_h$. Since $\tau_{x_h}>\tau_{x_a}$, choice under the new list first-order stochastically dominates that choice under $l^*$ with respect
to $\succ_\tau$. By monotonicity of $f_\tau$, the platform has a strictly higher utility under the new list. This contradicts the
optimality of $l^*$; therefore $a\le h$. This completes our proof. $\blacksquare$

%%%%%%%%%%%%%%%%%%%%%%%%%%%%%%%%%%%%%%%%%%%%%%%%%%%%%%%%%%%%%%%%%%%
%%%%%%%%%%%%%%%%%%%%%%%%%% Prof of Proposition %%%%%%%%%%%%%%%%%%%%
%%%%%%%%%%%%%%%%%%%%%%%%%%%%%%%%%%%%%%%%%%%%%%%%%%%%%%%%%%%%%%%%%%%
\subsection{Proof of Proposition \ref{prop:swaps index}} Suppose $\rho$ and $\rho'$ have consistent DST representations $(\alpha,\succ,w)$ and $(\alpha',\succ,w')$, respectively.  

\noindent \underline{\textit{Part (ii):}} The proof uses Theorem 1 in \cite{apesteguia2015measure}. Let $P$ be an arbitrary linear order over $X$. In our framework, given that $\sigma\colon \mathcal{D}\to (0,1)$ is assumed to have full support, Theorem 1 in \cite{apesteguia2015measure} states that if $xPy$ implies that $\sum_{S\supseteq \xy} \sigma(S)\rho(x,S)>\sum_{S\supseteq \xy} \sigma(S)\rho(y,S)$, then $P$ is the unique swaps preference of $\rho$. Hence, to show that $\succ_{swaps}(\rho)\, \equiv \, \succ$, it is sufficient to prove that $x\succ y$ implies that $\sum_{S\supseteq \xy} \sigma(S) \rho(x,S)>\sum_{S\supseteq \xy} \sigma(S) \rho(y,S)$. Showing that $\succ_{swaps}(\rho')\, \equiv \, \succ$ is similar so we omit a formal proof. Take an arbitrary $S\supseteq \xy$. As $x\succ y$ and $(\succ,w)$ satisfies Consistency by assumption, it follows that $w(x)>w(y)$. Also, $x\succ y$ implies that $y$ cannot be the $\succ$-best element in $S$. Hence,
\[
\rho(x,S)\ge (1-\alpha)\frac{w(x)}{w(S)} > (1-\alpha)\frac{w(y)}{w(S)} = \rho(y,S).
\]
Consequently, $\sum_{S\supseteq \xy} \sigma(S) \rho(x,S)>\sum_{S\supseteq \xy} \sigma(S) \rho(y,S)$.

\noindent \underline{\textit{Part (i):}} Using the fact that $\succ_{swaps}(\rho)\, \equiv \, \succ$, we have
\[
I_{swaps}(\rho)=I(\rho,\succ)=\sum_{S\in \mathcal{D}}\sum_{x\in S} \sigma(S) \rho(x,S) |\{y: y\in S, y\succ x\}|
\]
Term $\sigma(S)\rho(x,S) |\{y: y\in S, y\succ x\}|$ vanishes when $x$ is the $\succ$-best option in $S$ because the set $\{y: y\in S, y\succ x\}$ is empty. Recall that the $\succ$-best option in $S$ is denoted by $b_{\succ}(S)$. Then  
\begin{eqnarray*}
I_{swaps}(\rho)&=&\sum_{S\in \mathcal{D}} \sum_{x\in S; x\ne b_{\succ}(S)} \sigma(S) \rho(x,S) |\{y: y\in S, y\succ x\}| \\
&=&\sum_{S\in \mathcal{D}} \sum_{x\in S; x\ne b_{\succ}(S)} \sigma(S)\frac{(1-\alpha)w(x)}{w(S)}|\{y: y\in S, y\succ x\}|,
\end{eqnarray*}
where the second equation comes from the DST representation of $\rho$ (note that $x$ is not the $\succ$-best in $S$). Similarly, 
\[
I_{swaps}(\rho')=\sum_{S\in \mathcal{D}} \sum_{x\in S; x\ne b_{\succ}(S)} \sigma(S) \frac{(1-\alpha')w'(x)}{w'(S)}|\{y: y\in S, y\succ x\}|
\]
When $w=w'$, it is straightforward that $I_{swaps}(\rho)<I_{swaps}(\rho')$ if and only if $\alpha>\alpha'$. 

Now, suppose $\frac{w(x)}{w'(x)}\ge\frac{w(y)}{w'(y)}$ whenever $x\succ y$ and suppose $\alpha>\alpha'$. We prove that $I_{swaps}(\rho)<I_{swaps}(\rho')$. As $1-\alpha<1-\alpha'$, it is sufficient to show that for a fixed menu $S\in \mathcal{D}$
\[
\sum_{x\in S; x\ne b_{\succ}(S)} \frac{w(x)}{w(S)}|\{y: y\in S, y\succ x\}|\le \sum_{x\in S; x\ne b_{\succ}(S)}\frac{w'(x)}{w'(S)}|\{y: y\in S, y\succ x\}|.
\]
Suppose $S=\{x_1,x_2,\dots,x_s\}$ with $s=|S|$ and $x_1\succ x_2 \succ \dots \succ x_s$. The inequality above is equivalent to
\begin{eqnarray*}
  \frac{\sum_{i=1}^s w(x_i) (i-1)}{w(S)} &\le& \frac{\sum_{i=1}^s w'(x_i) (i-1)}{w'(S)}  \\
  \Leftrightarrow w'(S) \sum_{i=1}^s w(x_i) (i-1) &\le & w(S) \sum_{i=1}^s w'(x_i) (i-1)  \\
   \Leftrightarrow \sum_{(i,j): i<j}(i-j)[w(x_i)w'(x_j)-w'(x_i)w(x_j)] &\le& 0.
\end{eqnarray*}
We have $x_i\succ x_j$ whenever $i<j$ and it follows that $\frac{w(x_i)}{w'(x_i)} \ge \frac{w(x_j)}{w'(x_j)}$. Consequently, $w(x_i)w'(x_j)-w'(x_i)w(x_j)\ge 0$. Therefore, $\sum_{(i,j): i<j}(i-j)[w(x_i)w'(x_j)-w'(x_i)w(x_j)] \le 0$. This completes our proof of the Proposition. $\blacksquare$

%%%%%%%%%%%%%%%%%%%%%%%%%%%%%%%%%%%%%%%%%%%%%%%%%%%%%%%%%%%%%%%%%%%
%%%%%%%%%%%%%%%%%%%%%%%%%% A new section %%%%%%%%%%%%%%%%%%%%%%%%%%%
%%%%%%%%%%%%%%%%%%%%%%%%%%%%%%%%%%%%%%%%%%%%%%%%%%%%%%%%%%%%%%%%%%%

\subsection{Proof of Proposition \ref{prop:associated choice function}}
First, suppose $\rho$ has a DST-PA representation. We show that 
\begin{alignat*}{3}
    \rho^*(x,S)&= \alpha+(1-\alpha)\frac{w(x)}{\sum_{y\in S} w(y)} &&\text{ , if $x\in S \in \mathcal{D}$ and $x$ is $\succ$-best in $S$} \\
    \rho^*(y,S)&= (1-\alpha)\frac{w(y)}{\sum_{y\in S} w(y)} &&\text{ , if $y\in S \in \mathcal{D}$ and $y$ is not $\succ$-best in $S$}
\end{alignat*}
so $\rho^*$ has a DST representation $(\alpha,\succ,w)$. The proofs of Proposition \ref{prop:associated choice function} and Theorem \ref{thm:characterization DST-PA} use the following result from \cite{shafer1976mathematical}.

\begin{lemma} [M\"obius Inversion] \label{lemma:Mobius} Suppose $\Theta$ is a finite set. Suppose $f$ and $g$ are functions on $2^{\Theta}$. Then $\displaystyle f(A)=\underset{B\subseteq A}\sum g(B) \text{ for all $A\subseteq \Theta$} \Leftrightarrow g(A)=\underset{B\subseteq A}\sum (-1)^{|A\setminus B|} f(B)$ for all $A\subseteq \Theta$.
\end{lemma}

In what follows, we consider a menu $S$ with at least two elements. Additionally, for an arbitrary element $t\in S$, let $S_t$ be the set of all subsets of $S$ that contain $t$. Also, let $S_{t,\text{best}}$ be the set of all subsets of $S$ in which $t$ is the best-ranked item.

\medskip
\noindent \textbf{\underline{Step 1:}} Suppose $x$ is the $\succ$-best element in $S$. Then $x$ is also the $\succ$-best element in any subset of $S$ that includes $x$. To use Lemma \ref{lemma:Mobius}, let $\Theta=S\setminus x$. By definitions of $x$ and $\Theta$, $x$ is the best-ranked item in $A\cup x$ for all $A\subseteq \Theta$. For each $A\subseteq \Theta$ and $B\subseteq \Theta$, let $f(A)=O(x,A\cup x)$ and $g(B)=\frac{1}{\pi(\emptyset)}\Big[\alpha\pi(B\cup x)+(1-\alpha)\pi(B\cup x)\frac{w(x)}{w(B\cup x)}\Big]$. For all $A\subseteq \Theta$, we have
\begin{eqnarray*}
f(A)=O(x,A\cup x)&=&\frac{1}{\pi(\emptyset)} \Bigg[\alpha \underset{T\in (A\cup x)_{x,\text{best}}}{\sum} \pi(T) +(1-\alpha) \underset{T\in (A\cup x)_{x}}{\sum} \pi(T)\frac{w(x)}{w(T)}\Bigg] \\
&=& \frac{1}{\pi(\emptyset)}\Bigg[\alpha\sum_{B:B\subseteq A}\pi(B\cup x)+(1-\alpha)\sum_{B:B\subseteq A}\pi(B\cup x)\frac{w(x)}{w(B\cup x)}\Bigg]\\
&=& \frac{1}{\pi(\emptyset)}\sum_{B:B\subseteq A} \Bigg[\alpha\pi(B\cup x)+(1-\alpha)\pi(B\cup x)\frac{w(x)}{w(B\cup x)}\Bigg]\\
&=& \sum_{B:B\subseteq A} g(B),
\end{eqnarray*}
The first equation comes from the definition of $O(x,A\cup x)$. The second equation results from two facts. First, $x$ is best in $S$ so $x$ is best in $T\subseteq A\cup x$ if and only if $T=B\cup x$ for some $B\subseteq A$. Second, any subset of $A\cup x$ that includes $x$ can be written as $B\cup x$ for some $B\subseteq A$. The last equation comes from the definition of the $g$ function. Now, Lemma \ref{lemma:Mobius} applies and it implies that
\begin{equation}\label{eq:Mobius}
g(A)=\underset{B\subseteq A}\sum (-1)^{|A\setminus B|} f(B) \quad \text{ for all $A\subseteq \Theta$.}
\end{equation}
Replace $A$ by $S\setminus x$ in the equation above, equation (\ref{eq:Mobius}) becomes
\begin{eqnarray*}
    g(S\setminus x)&=&\underset{B\subseteq (S\setminus x)}\sum (-1)^{|(S\setminus x)\setminus B|} f(B) \\
\Rightarrow  \frac{1}{\pi(\emptyset)}\Big[\alpha\pi(S)+(1-\alpha)\pi(S)\frac{w(x)}{w(S)}\Big]&=&\underset{B:B\subseteq (S\setminus x)}\sum (-1)^{|(S\setminus x)\setminus B|} O(x,B\cup x) \\
\Leftrightarrow  \frac{\pi(S)}{\pi(\emptyset)}\Big[\alpha+(1-\alpha)\frac{w(x)}{w(S)}\Big]&=&\underset{B:B\subseteq (S\setminus x)}\sum (-1)^{|(S\setminus x)\setminus B|} O(x,B\cup x) \\
\Leftrightarrow  \frac{\pi(S)}{\pi(\emptyset)}\Big[\alpha+(1-\alpha)\frac{w(x)}{w(S)}\Big]&=&\underset{T:T\in S_x}\sum (-1)^{|S\setminus T|} O(x,T),
\end{eqnarray*}
where the last equation uses the fact that for each $T\in S_x$, then $T\setminus x$ is a subset of $S\setminus x$. Here $S_x$ is the set of subsets of $S$ that contain $x$. Now, using the definition of the \textit{associated random choice function} $\rho^*$, we have
\[
    \Rightarrow \rho^*(x,S)=\frac{\pi(\emptyset)}{\pi(S)}\underset{T:T\in S_x}\sum (-1)^{|S\setminus T|} O(x,T)=\alpha+(1-\alpha)\frac{w(x)}{w(S)}.
\]

\medskip
\noindent \textbf{\underline{Step 2:}} Suppose $y$ is not $\succ$-best in $S$. Let $x$ be the $\succ$-best in $S$. It follows $x\succ y$. Consider an arbitrary $A\subseteq S\setminus \xy$. By definition
\begin{eqnarray*}
    O(y,A\cup \xy)&=&\frac{1}{\pi(\emptyset)}\Bigg[\alpha \underset{T\in (A\cup \xy)_{y,\text{best}}}{\sum} \pi(T) +(1-\alpha) \underset{T\in (A\cup \xy)_{y}}{\sum} \pi(T)\frac{w(y)}{w(T)}\Bigg] \\
    O(y,A\cup y) &=&\frac{1}{\pi(\emptyset)}\Bigg[\alpha \underset{T\in (A\cup y)_{y,\text{best}}}{\sum} \pi(T) +(1-\alpha) \underset{T\in (A\cup y)_{y}}{\sum} \pi(T)\frac{w(y)}{w(T)}\Bigg]
\end{eqnarray*}
We simplify the difference $O(y,A\cup \xy)-O(y,A\cup y)$. Note that $y$ cannot be the $\succ$-best item if $x$ is available. This implies that $(A\cup \xy)_{y,\text{best}}=(A\cup y)_{y,\text{best}}$. It follows 
\begin{equation}\label{eq: term 1}
 \alpha \underset{T\in (A\cup \xy)_{y,\text{best}}}{\sum} \pi(T)=\alpha \underset{T\in (A\cup y)_{y,\text{best}}}{\sum}\pi(T)   
\end{equation}
Note that any subset of $A\cup y$ that includes $y$ is also a subset of $A\cup \xy$. Meanwhile, a subset of $A\cup \xy$ that includes $x$ cannot be a subset of $A\cup y$. Therefore, we can write
\begin{eqnarray}\label{eq: term 2}
\underset{T\in (A\cup\xy)_{y}}{\sum} \pi(T)\frac{w(y)}{w(T)}&=&\underset{T\in (A\cup y)_{y}}{\sum} \pi(T)\frac{w(y)}{w(T)}+\underset{T:T\in (A\cup \xy)_y; x\in T}{\sum} \pi(T)\frac{w(y)}{w(T)} \nonumber \\
&=&\underset{T\in (A\cup y)_{y}}{\sum} \pi(T)\frac{w(y)}{w(T)}+\underset{B:B\subseteq A}{\sum} \pi(B\cup \xy)\frac{w(y)}{w(B\cup \xy)},
\end{eqnarray}
where equation (\ref{eq: term 2}) comes from the fact that for each $T\subseteq (A\cup \xy)_y$ such that $x\in T$, there exists a unique $B\subseteq A$ such that $T=B\cup \xy$. Using equations (\ref{eq: term 1}) and (\ref{eq: term 2}), we have
\begin{equation}\label{eq:thm}
 O(y,A\cup \xy) -O(y,A\cup y) =\frac{1}{\pi(\emptyset)} \Bigg[(1-\alpha) \underset{B:B\subseteq A}{\sum}\pi(B\cup \xy)\frac{w(y)}{w(B\cup \xy)} \Bigg].
\end{equation}

To use Lemma \ref{lemma:Mobius}, let $\Theta^1=S\setminus \xy$. By definition of $\Theta^1$, $x$ is the best-ranked item in $A\cup \xy$ for all $A\subseteq \Theta^1$. For each $A\subseteq \Theta^1$ and $B\subseteq \Theta^1$, let $f^1(A)=O(y,A\cup \xy)-O(y,A\cup \y)$ and $g^1(B)=\frac{1}{\pi(\emptyset)}\Big[(1-\alpha)\pi(B\cup \xy)\frac{w(y)}{w(B\cup \xy)}\Big]$. Then equation (\ref{eq:thm}) implies
\[
f^1(A)=\sum_{B:B\subseteq A} g^1(B)\quad \text{ for all $A\subseteq \Theta^1$}.
\]
Applying M\"obius Inversion (Lemma \ref{lemma:Mobius}), we have
\[
g^1(A)=\sum_{B:B\subseteq A} (-1)^{|A\setminus B|} f^1(B)\quad \text{ for all $A\subseteq \Theta^1$}.
\]
Using definitions of $g^1$ and $f^1$, the equation above is equivalent to
\begin{equation} \label{eq:ydiff}
\frac{(1-\alpha)\pi(A\cup \xy)\frac{w(y)}{w(A\cup \xy)}}{\pi(\emptyset)}=\underset{B:B\subseteq A}{\sum}(-1)^{|A\setminus B|}[O(y,B\cup \xy)-O(y,B\cup y)] \quad \forall A\subseteq \Theta^1.
\end{equation}
Replace $A=S\setminus \xy$ in the equation above, equation (\ref{eq:ydiff}) becomes
\begin{eqnarray}
\frac{(1-\alpha)\pi(S)\frac{w(y)}{w(S)}}{\pi(\emptyset)}&=&\underset{B:B\subseteq (S\setminus \xy)}{\sum}(-1)^{|(S\setminus \xy)\setminus B|}[O(y,B\cup \xy)-O(y,B\cup y)]  \nonumber \\
&=& \underset{T:T\in S_{xy}}{\sum}(-1)^{|S\setminus T|}[O(y,T)-O(y,T\setminus x)], \label{eq:ydiff3}
\end{eqnarray}
where $S_{xy}$ in equation (\ref{eq:ydiff3}) is the set of all subsets of $S$ that include both $x$ and $y$. This equation uses the fact that for all $T\in S_{xy}$, there exists a unique $B\subseteq S\setminus \xy$ such that $T=B\cup \xy$. We handle the sum $\underset{T:T\in S_{xy}}{\sum}(-1)^{|S\setminus T|}[O(y,T)-O(y,T\setminus x)]$ by using the following Lemma.

\begin{lemma} \label{lemma:identity}For two arbitrary elements $z,t$ and any set $S\supseteq \{z,t\}$ 
\begin{equation} \label{eq:identity}
       \underset{T: T\in S_t}{\sum}(-1)^{|S\setminus T|}O(t,T)= \underset{T:T\in S_{zt}}{\sum}(-1)^{|S\setminus T|}[O(t,T)-O(t,T\setminus z)],
\end{equation}
where $S_t$ is the set of all subsets of $S$ that contain $t$ and $S_{zt}$ is the set of all subsets of $S$ that contain both $z$ and $t$.
\end{lemma}

\begin{proof} Partition the set of all subsets of $S$ that contain $t$ into two subsets. One includes all subsets that include $z$ and the other includes all subsets that do not include $z$. The former is $S_{zt}$ and the latter is $(S\setminus z)_t$. Mathematically, $S_t=S_{zt}\cup (S\setminus z)_t$ and $S_{zt}\cap (S\setminus z)_t=\emptyset$. Thus,
\[
\underset{T: T\in S_t}{\sum}(-1)^{|S\setminus T|}O(t,T)=\underset{T: T\in S_{zt}}{\sum}(-1)^{|S\setminus T|}O(t,T)+\underset{T: T\in (S\setminus z)_t}{\sum}(-1)^{|S\setminus T|}O(t,T)
\]
For every $T\in (S\setminus z)_t$, there exists a unique $T^\prime \in S_{zt}$ such that $T=T^\prime \setminus z$. Additionally, for every $T^\prime \in S_{zt}$, there exists unique $T\in (S\setminus z)_t$ such that $T=T^\prime \setminus z$. Therefore, 
\[
\underset{T: T\in (S\setminus z)_t}{\sum}(-1)^{|S\setminus T|}O(t,T)=-\underset{T':T'\in S_{zt}}{\sum}(-1)^{|S\setminus T'|}O(t,T'\setminus z).
\]
Hence,
\begin{eqnarray*}
    \underset{T: T\in S_t}{\sum}(-1)^{|S\setminus T|}O(t,T)&=&\underset{T: T\in S_{zt}}{\sum}(-1)^{|S\setminus T|}O(t,T)-\underset{T':T'\in S_{zt}}{\sum}(-1)^{|S\setminus T'|}O(t,T'\setminus z)\\
    &=& \underset{T:T\in S_{zt}}{\sum}(-1)^{|S\setminus T|}[O(t,T)-O(t,T\setminus z)].
\end{eqnarray*}
This proves the identity. $\blacksquare$

\smallskip
\noindent Now, from equations (\ref{eq:ydiff3}) and (\ref{eq:identity}), 
\begin{alignat*}{2}
    \frac{\pi(S)}{\pi(\emptyset)}(1-\alpha)\frac{w(y)}{w(S)}&=\underset{T:T\in S_y}{\sum}(-1)^{|S\setminus T|}O(y,T) &&\text{ , if $y$ is not $\succ$-best in $S$}\\
    \Rightarrow \rho^*(y,S)&= (1-\alpha)\frac{w(y)}{w(S)} &&\text{ , if $y$ is not $\succ$-best in $S$,}
\end{alignat*}
where the second equation comes from the definition of $\rho^*$. This completes the proof of the Proposition. $\blacksquare$ 
\end{proof}

%%%%%%%%%%%%%%%%%%%%%%%%%%%%%%%%%%%%%%%%%%%%%%%%%%%%%%%%%%%%%
%%%%%%%%%%%%%%%%%%%%%%%%% New Proof %%%%%%%%%%%%%%%%%%%%%%%%
%%%%%%%%%%%%%%%%%%%%%%%%%%%%%%%%%%%%%%%%%%%%%%%%%%%%%%%%%%%%%

\subsection{Proof of Theorem \ref{thm:characterization DST-PA}}
For the only-if part, see the proof of Proposition \ref{prop:associated choice function} for details. For the if part, suppose there exists a DST $(\alpha,\succ,w)$ representation of $\rho^*$. This means 
\begin{alignat*}{3}
    \rho^*_{\alpha,\succ,w}(x,S)&= (1-\alpha)\frac{w(x)}{w(S)}+\alpha &&\text{ , if $x\in S$ and $x$ is $\succ$-best in $S$} \\
    \rho^*_{\alpha,\succ,w}(y,S)&= (1-\alpha)\frac{w(y)}{w(S)} &&\text{ , if $y\in S$ and $y$ is not $\succ$-best in $S$}
\end{alignat*}
Define mapping $\pi\colon \mathcal{D}\cup\{\emptyset\} \to \mathbb{R}_{++}$ as
\[
    \pi(\emptyset)=\Bigg[\underset{S\in \mathcal{D}\cup \{\emptyset\}}{\sum}\underset{T\subseteq S}{\sum}\frac{(-1)^{|S\setminus T|}}{\rho(a^*,T)}\Bigg]^{-1} \quad \text{ and } \quad\pi(S)=\pi(\emptyset)\underset{T\subseteq S}{\sum}\frac{(-1)^{|S\setminus T|}}{\rho(a^*,T)} \quad \text{ for all $S\in \mathcal{D}$.}
\]
By Block-Marschak inequality on the default option (Axiom \ref{axiom:BM}), $\pi(S)>0$ for all $S \in \mathcal{D}\cup\{\emptyset\}$. Thus, by construction $\pi(S)\in (0,1)$ for all $S\in \mathcal{D}\cup\{\emptyset\}$ and it is straightforward to see that $\underset{S\in \mathcal{D}\cup\{\emptyset\}}{\sum}\pi(S)=1$. Additionally,
\begin{equation}\label{eq:characterization 1}
    \frac{\pi(S)}{\pi(\emptyset)}= \underset{T\subseteq S}{\sum}\frac{(-1)^{|S\setminus T|}}{\rho(a^*,T)} \quad \text{for all $S\in \mathcal{D}\cup\{\emptyset\}$}.
\end{equation}
and by M\"obius inversion,
\begin{equation}\label{eq:characterization 2}
\frac{1}{\rho(a^*,S)}=\frac{\sum_{T:T\subseteq S}\pi(T)}{\pi(\emptyset)} \quad \text{for all $S\in \mathcal{D}\cup\{\emptyset\}$}.
\end{equation}
By assumption, $\rho^*$ has a DST representation $(\alpha,\succ,w)$. Hence, 
\begin{eqnarray}
  \rho^*(x,S)&= \alpha+(1-\alpha)\frac{w(x)}{w(S)} \quad &\text{ if $x\in S$ and $x$ is $\succ$-best in $S\in \mathcal{D}$,} \label{eq:characterization 3} \\
    \rho^*(y,S)&= (1-\alpha)\frac{w(y)}{w(S)} \quad &\text{ if $y\in S$ and $y$ is not $\succ$-best in $S\in \mathcal{D}$.} \label{eq:characterization 4}
\end{eqnarray}
Using the definition of $\rho^*$
\begin{equation}\label{eq:characterization 5}
\rho^*(z,S)=\Bigg[\underset{T\subseteq S}{\sum}\frac{(-1)^{|S\setminus T|}}{\rho(a^*,T)}\Bigg]^{-1}\underset{T:T\in S_z}{\sum}(-1)^{|S\setminus T|}O(z,T) \quad \text{ for all $z\in S\in \mathcal{D}$.}
\end{equation}
Using equations (\ref{eq:characterization 1}), (\ref{eq:characterization 3}), (\ref{eq:characterization 4}), and (\ref{eq:characterization 5}), we have
\begin{alignat*}{2}
 \underset{T:T\in S_x}{\sum}(-1)^{|S\setminus T|}O(x,T)&= \frac{\pi(S)}{\pi(\emptyset)}\Bigg[\alpha+(1-\alpha)\frac{w(x)}{w(S)}\Bigg] &&\text{ if $x\in S$ and $x$ is $\succ$-best in $S\in \mathcal{D}$}\\
 \underset{T:T\in S_y}{\sum}(-1)^{|S\setminus T|}O(y,T)&= \frac{\pi(S)}{\pi(\emptyset)}\Bigg[(1-\alpha)\frac{w(y)}{w(S)}\Bigg] &&\text{ if $y\in S$ and $y$ is not $\succ$-best in $S\in \mathcal{D}$}
\end{alignat*}
From here, follow the two steps below.

\noindent \textit{\underline{\textbf{Step 1: Representation of choice of $x$ in $S$.}}} Note that if $x$ is $\succ$-best in $S$ then it is also $\succ$-best in every set $T\in S_x$. Thus, using a M\"obius-inversion argument similar to that in the proof of Proposition \ref{prop:associated choice function}, we have
\begin{alignat*}{2}
   O(x,S)&=\frac{1}{\pi(\emptyset)}\underset{T:T\in S_x}{\sum}\Bigg[\alpha\pi(T)+(1-\alpha)\pi(T)\frac{w(x)}{w(T)}\Bigg] &&\text{ if $x\in S$ and $x$ is $\succ$-best in $S$},\\
\Rightarrow \rho(x,S)&=\frac{1}{\underset{T:T\subseteq S}{\sum}\pi(T)}\Bigg[\alpha\underset{T:T\in S_x}{\sum}\pi(T)+(1-\alpha)\underset{T:T\in S_x}{\sum}\pi(T)\frac{w(x)}{w(T)}\Bigg] &&\text{ if $x\in S$ and $x$ is $\succ$-best in $S$}.
\end{alignat*}

\noindent \textit{\underline{\textbf{Step 2: Representation of choice of $y$ in $S$.}}} The M\"obius inversion technique used in Step 1 is not directly applicable to the situation in which $y$ is not $\succ$-best in $S$ because there exists at least one $T\in S_y$ such that $y$ is $\succ$-best in $T$ (for example, $T=\{y\}$). When $y$ is not $\succ$-best in $S$, $\exists x \in S$ such that $x\succ y$. The implication is that $y$ is also not $\succ$-best in every set $T\in S_{xy}$. Using Lemma \ref{lemma:identity} to write
\[
        \underset{T:T\in S_{xy}}{\sum}(-1)^{|S\setminus T|}[O(y,T)-O(y,T\setminus x)]=\underset{T: T\in S_y}{\sum}(-1)^{|S\setminus T|}O(y,T)=\frac{\pi(S)}{\pi(\emptyset)}\Bigg[(1-\alpha)\frac{w(y)}{w(S)}\Bigg]
\]
and applying the M\"obius inversion similar to that in the proof of Proposition \ref{prop:associated choice function},
\begin{equation}\label{eq:yrelationship}
        O(y,S)-O(y,S\setminus\x)=\frac{1}{\pi(\emptyset)} \Bigg[(1-\alpha) \underset{T:T\in S_{xy}}{\sum}\pi(T)\frac{w(y)}{w(T)} \Bigg].
\end{equation}
Using this relationship, it can be shown by induction that when $y$ is not $\succ$-best in $S$
\begin{equation}\label{eq:yform}
    O(y,S)=\frac{1}{\pi(\emptyset)}\Bigg[\underset{T:T\in S_{y,\text{best}}}{\sum}\alpha\pi(T)+(1-\alpha)\underset{T: T\in S_y}{\sum}\pi(T)\frac{w(y)}{w(T)}\Bigg], 
\end{equation}
so $\rho(y,S)=\rho_{\alpha,\succ,w,\pi}(y,S)$.

\noindent \underline{\textit{Step 2.1: Binary menus.}} Consider a binary menu $\xy$ with $x\succ y$. From Step 1, we know that
\begin{equation}\label{eq:characterization 6}
O(x,\xy)=\frac{1}{\pi(\emptyset)}\Bigg[\alpha(\pi(x)+\pi(\xy))+(1-\alpha)\Bigg(\pi(x)+\pi(\xy)\frac{w(x)}{w(x)+w(y)}\Bigg)\Bigg].
\end{equation}
By definition
\begin{eqnarray*}
O(x,\xy)+O(y,\xy)&=&\frac{\rho(x,\xy)+\rho(y,\xy)}{\rho(a^*,\xy)}=\frac{1-\rho(a^*,\xy)}{\rho(a^*,\xy)} \\
\Rightarrow O(y,\xy)&=& \frac{1-\rho(a^*,\xy)}{\rho(a^*,\xy)}- O(x,\xy) \\
&=&\frac{1}{\pi(\emptyset)}\Bigg[\alpha\pi(y)+(1-\alpha)\Bigg(\pi(y)+\pi(\xy)\frac{w(y)}{w(x)+w(y)}\Bigg)\Bigg],
\end{eqnarray*}
where the last equation uses (\ref{eq:characterization 2}) and (\ref{eq:characterization 6}). The last equation above implies that $O(y,\xy)$ has the form in equation (\ref{eq:yform}).

\noindent \underline{\textit{Step 2.2: Menus of size $3$ or above.}} Now, consider a menu $S\in \mathcal{D}$ with at least three alternatives where $y$ is not $\succ$-best in $S$. Suppose that $O(y,A)$ has the form in equation (\ref{eq:yform}) for any proper subset $A$ of $S$. Let $x\in S$ be an arbitrary alternative that is $\succ$-better than $y$. 
Note that $O(y,S\setminus\x)$ has the form in equation (\ref{eq:yform}) by assumption. Thus, using (\ref{eq:yrelationship})
\begin{eqnarray*}
    O(y,S)&=&O(y,S\setminus\x)+\frac{1}{\pi(\emptyset)} \Bigg[(1-\alpha) \underset{T:T\in S_{xy}}{\sum}\pi(T)\frac{w(y)}{w(T)} \Bigg]\\
    &=&\frac{\alpha}{\pi(\emptyset)}\underset{T:T\in S_{y,\text{best}}}{\sum}\pi(T)+\frac{1-\alpha}{\pi(\emptyset)}\underset{T:T\in (S\setminus x)_{y}}{\sum}\pi(T)\frac{w(y)}{w(T)} + \frac{1-\alpha}{\pi(\emptyset)} \underset{T:T\in S_{xy}}{\sum}\pi(T)\frac{w(y)}{w(T)}\\
    &=&\frac{\alpha}{\pi(\emptyset)}\underset{T:T\in S_{y,\text{best}}}{\sum}\pi(T)+\frac{1-\alpha}{\pi(\emptyset)}\underset{T:T\in S_{y}}{\sum}\pi(T)\frac{w(y)}{w(T)}.
\end{eqnarray*}
Here, the second equation uses the assumption that $O(y,S\setminus\x)$ follows (\ref{eq:yform}) and $(S\setminus x)_{y,\text{best}}=S_{y,\text{best}}$ since $x\succ y$. The last equation uses $S_y=(S\setminus x)_{y}\cup S_{xy}$ and $(S\setminus x)_y\cap S_{xy}=\emptyset$. Therefore, 
\[
\rho(y,S)=\alpha \underset{T:T\in S_{y,\text{best}}}{\sum} \frac{\pi(T)}{\underset{A:A\subseteq S}{\sum}\pi(A)} +(1-\alpha) \underset{T:T\in S_{y}}{\sum} \frac{\pi(T)}{\underset{A:A\subseteq S}{\sum}\pi(A)} \frac{w(y)}{w(T)}
\]
if $y$ is not $\succ$-best in $S$. This completes the proof. $\blacksquare$

%%%%%%%%%%%%%%%%%%%%%%%%%%%%%%%%%%%%%%%%%%%%%%%%%%%%%%%%%%%%%%%%%%%
%%%%%%%%%%%%%%%%%%%%%%%%%% A new section %%%%%%%%%%%%%%%%%%%%%%%%%%%
%%%%%%%%%%%%%%%%%%%%%%%%%%%%%%%%%%%%%%%%%%%%%%%%%%%%%%%%%%%%%%%%%%%
\subsection{Proof of Proposition \ref{prop:DST-PA-MM}} The necessity is straightforward and thus is omitted. For the sufficiency, as $\rho$ satisfies Axiom \ref{axiom:BM} and its associated random choice function $\rho^*$ has a DST representation, Theorem \ref{thm:characterization DST-PA} implies that it has a DST-PA representation. It is then sufficient to prove that there exists $\phi\colon X\to (0,1)$ such that 
$\frac{\pi(T)}{\sum_{A\subseteq S}\pi(A)}
=
\prod_{x\in T}\phi(x)
\prod_{y\in S\setminus T}\bigl(1-\phi(y)\bigr)$ for all $S\in\mathcal D$ and $T\subseteq S$. By Axiom \ref{axiom:mido}, for all $x\in X$ and $S,S'$ such that $x\in S,S'\in \mathcal{D}$, we have
\[
\frac{\rho(a^*,S)}{\rho(a^*,S\setminus x)}=\frac{\rho(a^*,S')}{\rho(a^*,S'\setminus x)}.
\]
In the equation above, let $S'=\{x\}$. Then $\rho(a^*,S'\setminus x)=\rho(a^*,\emptyset)=1$ and it follows that $\rho(a^*,S)=\rho(a^*,S\setminus x)\rho(a^*,\{x\})$ for all pairs $(x,S)$ such that $x\in S\in \mathcal{D}$. Using the representation of $\rho$ that $\rho(a^*,S)=\frac{\pi(\emptyset)}{\sum_{A\subseteq S}\pi(A)}$, the equation $\rho(a^*,S)=\rho(a^*,S\setminus x)\rho(a^*,\{x\})$ is equivalent to
\begin{equation}\label{eq:MM-1}
\pi(\emptyset)\sum_{A\subseteq S}\pi(A)=[\pi(\emptyset)+\pi(\{x\})]\sum_{B\subseteq S\setminus x}\pi(B) \quad \text{ for all $x\in S\in \mathcal{D}$.}
\end{equation}
We will use (\ref{eq:MM-1}) to show
\begin{equation}\label{eq:MM-2}
    \displaystyle \pi(S) \pi(\emptyset)^{|S|-1}=\prod_{x\in S}\pi(\{x\}) \quad \text{ for all } S\in \mathcal{D}
\end{equation}
by induction based on the number of options in $S$. 

\noindent \textbf{\underline{\textit{Step 1:}}} Note that (\ref{eq:MM-2}) trivially holds when $|S|=1$ (and when $S=\emptyset$ by the standard empty-product convention). We show that (\ref{eq:MM-2}) holds when $|S|=2$. In equation (\ref{eq:MM-1}), replace $S=\xy$, we have 
\[
\pi(\emptyset)[\pi(\emptyset)+\pi(\{x\})+\pi(\{y\})+\pi(\xy)]=[\pi(\emptyset)+\pi(\{x\})][\pi(\emptyset)+\pi(\{y\})],
\]
which implies $\pi(\xy)\pi(\emptyset) =\pi(\{x\})\pi(\{y\})$. 

\noindent \textbf{\underline{\textit{Step 2:}}} Consider arbitrary $S\in \mathcal{D}$ with at least three alternatives. Suppose equation (\ref{eq:MM-2}) holds for all subsets $S'$ of $S$ such that $S'\ne S$. We will show that equation (\ref{eq:MM-2}) also holds at $S$. Take $x\in S$. From (\ref{eq:MM-1})
\begin{eqnarray}
\pi(\emptyset)\sum_{A\subseteq S}\pi(A)&=&[\pi(\emptyset)+\pi(\{x\})]\sum_{B\subseteq S\setminus x}\pi(B) \nonumber \\
\pi(\emptyset)\bigg[\sum_{B\subseteq S\setminus x}(\pi(B)+\pi(B\cup x))\bigg]&=&[\pi(\emptyset)+\pi(\{x\})]\sum_{B\subseteq S\setminus x}\pi(B) \nonumber \\   
\pi(\emptyset)\sum_{B\subseteq S\setminus x}\pi(B)+\pi(\emptyset)\sum_{B\subseteq S\setminus x}\pi(B\cup x) &=&\pi(\emptyset)\sum_{B\subseteq S\setminus x}\pi(B)+\pi(\{x\})\sum_{B\subseteq S\setminus x}\pi(B) \nonumber \\
\pi(\emptyset)\sum_{B\subseteq S\setminus x}\pi(B\cup x) &=&\pi(\{x\})\sum_{B\subseteq S\setminus x}\pi(B) \label{eq:MM-3}. 
\end{eqnarray}
Now, let $B'$ be an arbitrary proper subset of $S\setminus x$, i.e., $B' \subseteq S\setminus x$ and $B' \ne S\setminus x$. Then both $B'$ and $B'\cup x$ are proper subsets of $S$. By induction hypothesis, $\pi(B')$ and $\pi(B'\cup x)$ have the form in equation (\ref{eq:MM-2}). This means
\[
\pi(B'\cup x) \pi(\emptyset)^{|B'|}=\prod_{y\in B'\cup x}\pi(\{y\}) \quad \text{ and } \quad \pi(B') \pi(\emptyset)^{|B'|-1}=\prod_{z\in B'}\pi(\{z\}),
\]
which implies $\pi(\emptyset)\pi(B'\cup x)=\pi(\{x\})\pi(B')$. Take the summation across all possible $B'$
\begin{equation}\label{eq:MM-4}
\pi(\emptyset)\sum_{B':B'\subseteq S\setminus x, B'\ne S\setminus x} \pi(B'\cup x)=\pi(\{x\})\sum_{B':B'\subseteq S\setminus x, B'\ne S\setminus x} \pi(B')
\end{equation}
Equations (\ref{eq:MM-3}) and (\ref{eq:MM-4}) then imply that $\pi(\emptyset)\pi(S)=\pi(\{x\})\pi(S\setminus x)$. As $\pi(S\setminus x)$ has the form in equation (\ref{eq:MM-2}) by induction hypothesis, we have
\[
\pi(\emptyset)\pi(S)=\pi(\{x\})\pi(S\setminus x)=\pi(\{x\})\frac{\prod_{y\in S\setminus x} \pi(\{y\})}{\pi(\emptyset)^{|S\setminus x|-1}} \Rightarrow  \pi(S) \pi(\emptyset)^{|S|-1}=\prod_{z\in S}\pi(\{z\}).
\]
Hence, (\ref{eq:MM-2}) also holds at $S$. By induction, (\ref{eq:MM-2}) holds for all non-empty $S\in \mathcal{D}$. 

Now, let $\phi\colon X \rightarrow (0,1)$ such that $\phi(x)=\frac{\pi(\{x\})}{\pi(\{x\})+\pi(\emptyset)}$ for all $x\in X$. Equivalently, $\frac{\phi(x)}{1-\phi(x)}=\frac{\pi(\{x\})}{\pi(\emptyset)}$. Equation (\ref{eq:MM-2}) then implies $ \frac{\pi(T)}{\pi(\emptyset)}=\prod_{x\in T} \frac{\phi(x)}{1-\phi(x)}$ for $T\in \mathcal{D}\cup \{\emptyset\}$, where we adopt the convention that a product taken over the empty set is $1$. Hence, for all non-empty $S\in \mathcal{D}$, we have
\[
\frac{\sum_{T\subseteq S}\pi(T)}{\pi(\emptyset)}= \sum_{T\subseteq S} \prod_{x\in T} \frac{\phi(x)}{1-\phi(x)},
\]
The RHS of the equation above is indeed
\[
\prod_{x\in S} \Bigg[\frac{\phi(x)}{1-\phi(x)}+1\Bigg]= \prod_{x\in S} \frac{1}{1-\phi(x)}=\frac{1}{\prod_{x\in S}(1-\phi(x))}.
\]
Hence, $\frac{\sum_{T\subseteq S}\pi(T)}{\pi(\emptyset)}= \frac{1}{\prod_{x\in S}(1-\phi(x))}$. It follows that for all $T\subseteq S$ and $S\in \mathcal{D}$
\[
\frac{\pi(T)}{\sum_{T'\subseteq S}\pi(T')}=\frac{\pi(T)}{\pi(\emptyset)}\cdot \frac{\pi(\emptyset)}{\sum_{T'\subseteq S}\pi(T')}= \prod_{x\in T} \frac{\phi(x)}{1-\phi(x)} \cdot \prod_{x\in S}(1-\phi(x))=\prod_{x\in T}\phi(x) \prod_{y\in S\setminus T} (1-\phi(y)).
\]
This completes our proof. $\blacksquare$

%%%%%%%%%%%%%%%%%%%%%%%%%%%%%%%%%%%%%%%%%%%%%%%%%%%%%%%%%%%%%%%%%%%
%%%%%%%%%%%%%%%%%%%%%%%%%% Prof of Remark %%%%%%%%%%%%%%%%%%%%%%%%%%%
%%%%%%%%%%%%%%%%%%%%%%%%%%%%%%%%%%%%%%%%%%%%%%%%%%%%%%%%%%%%%%%%%%%
\subsection{Proof of Remark \ref{rem:d-DST}}
The necessity is straightforward. For the sufficiency, suppose $\rho$ is positive and violates IIA. Without loss of generality, suppose 
\[
\frac{\rho(z,\xyz)}{\rho(y,\xyz)}>\frac{\rho(z,\yz)}{\rho(y,\yz)}.
\]
Consider a preference $x\succ y\succ z$. Choose $\alpha_X$ arbitrarily such that $0<\alpha_X<\rho(x,\xyz)$ and 
\begin{equation}\label{eq:d-DST-1}
\min\bigg\{\frac{\rho(x,\xy)}{\rho(y,\xy)}, \frac{\rho(x,\xz)}{\rho(z,\xz)} \bigg\} > \max\bigg\{\frac{\rho(x,\xyz)-\alpha_X}{\rho(y,\xyz)},\frac{\rho(x,\xyz)-\alpha_X}{\rho(z,\xyz)}\bigg\} 
\end{equation}
One can always choose such $\alpha_X$ because the RHS of (\ref{eq:d-DST-1}) goes to $0$ when $\alpha_X\to \rho(x,\xyz)$. Let $w(y)=\frac{\rho(y,\xyz)}{1-\alpha_X}$, $w(z)=\frac{\rho(z,\xyz)}{1-\alpha_X}$, and $w(x)=1-w(y)-w(z)$. Then $\rho$ has a representation at $X$. Let
\[
\alpha_{\xy} = 1-\rho(y,\xy)\frac{w(x)+w(y)}{w(y)};  \alpha_{\xz} = 1-\rho(z,\xz)\frac{w(x)+w(z)}{w(z)}
\]
and 
\[
\alpha_{\yz} = 1-\rho(z,\yz)\frac{w(y)+w(z)}{w(z)}.
\]
Given (\ref{eq:d-DST-1}), it is straightforward to verify that $\alpha_{\xy}\in (0,1)$ and $\alpha_{\xz}\in (0,1)$. The fact that $\alpha_{\yz}\in (0,1)$ comes from our assumption that $\rho(z,\xyz)/\rho(y,\xyz)>\rho(z,\yz)/\rho(y,\yz)$. By definition of $\alpha_{\xy}$, $\alpha_{\xz}$, $\alpha_{\yz}$, $\rho$ has a representation at $\xy$, $\xz$, and $\yz$. Hence, $\rho$ has a d-DST representation.

%%%%%%%%%%%%%%%%%%%%%%%%%%%%%%%%%%%%%%%%%%%%%%%%%%%%%%%%%%%%%%%%%%%
%%%%%%%%%%%%%%%%%%%%%%%%%% Prof of Proposition %%%%%%%%%%%%%%%%%%%%
%%%%%%%%%%%%%%%%%%%%%%%%%%%%%%%%%%%%%%%%%%%%%%%%%%%%%%%%%%%%%%%%%%%
\subsection{Proof of Proposition \ref{prop:d-DST_characterization}}
The necessity is straightforward. For the sufficiency, let $x^*$ and $z^*$ be $\succ_{R,d}$-best and $\succ_{R,d}$-worst in $X$, respectively. Let $w(z^*)=1$. Define
\[
w(y)=w(z^*)\frac{\rho(y,{\{x^*,y,z^*\}})}{\rho(z^*,{\{x^*,y,z^*\}})} \text{ for all $y$ such that $x^*\succ_{R,d} y \succ_{R,d} z^*$}.
\]
Let $w(x^*)$ be arbitrary such that 
\[
0<w(x^*)<\min_{(z,S): \{x^*,z\}\subseteq S, z\ne x^*}\frac{\rho(x^*,S)}{\rho(z,S)}w(z).
\]
We first show 
\begin{equation}\label{eq:dDST-2}
\frac{w(y)}{w(z)}=\frac{\rho(y,\xyz)}{\rho(z,\xyz)} \quad \text{for arbitrary $x,y,z$ such that $x\succ_{R,d} y \succ_{R,d} z$}.
\end{equation}
With arbitrary $x,y,z$ such that $x\succ_{R,d} y \succ_{R,d} z$, for $z$ and $z^*$ not necessarily distinct, we have
\[
\frac{w(y)}{w(z)}=\frac{w(y)}{w(z^*)}\cdot \frac{w(z^*)}{w(z)}=\frac{\rho(y,\{x^*,y,z^*\})}{\rho(z^*,\{x^*,y,z^*\})}\frac{\rho(z^*,\{x^*,z,z^*\})}{\rho(z,\{x^*,z,z^*\})}=\frac{\rho(y,\xyz)}{\rho(z,\xyz)},
\]
where the second equation comes from the definition of $w(y),w(z)$ and the last equation comes from Axiom \ref{axiom:2-dDST}.

Now, let $S\in \mathcal{D}$ be an arbitrary menu with at least two options. Define 
\[
\alpha_S= 1- \rho(l_{\succ_{R,d}}(S),S)\frac{w(S)}{w(l_{\succ_{R,d}}(S))}.
\]
Clearly, $\alpha_S<1$. We want to show $\alpha_S>0$, which is equivalent to $1/\rho(l_{\succ_{R,d}}(S),S)> w(S)/w(l_{\succ_{R,d}}(S))$. Consider the following cases: 

\noindent \textbf{Case 1:} $x^*\in S$. We first show that $\alpha_S$ is positive when there are two elements in $S$. Suppose not. Suppose $\alpha_S\le 0$. Using $\rho(x^*,S)=1-\rho(l_{\succ_{R,d}}(S),S)$ and the definition of $\alpha_S$
\[
\frac{\rho(x^*,S)}{\rho(l_{\succ_{R,d}}(S),S)}= \frac{\alpha_S + (1-\alpha_S)\frac{w(x^*)}{w(S)}}{(1-\alpha_S)\frac{w(l_{\succ_{R,d}}(S))}{w(S)}}\le  \frac{(1-\alpha_S)\frac{w(x^*)}{w(S)}}{(1-\alpha_S)\frac{w(l_{\succ_{R,d}}(S))}{w(S)}}=\frac{w(x^*)}{w(l_{\succ_{R,d}}(S))},
\]
which is a contradiction to the definition of $w(x^*)$.

Now, suppose $S$ has at least three alternatives. As $x^*\in S$, we have $x^*=b_{\succ_{R,d}}(S)$. For all $y\in S$ such that $y\ne l_{\succ_{R,d}}(S), x^*$, using Axiom \ref{axiom:2-dDST}
\begin{equation}\label{eq:dSFT-3}
 \frac{\rho(y,S)}{\rho(l_{\succ_{R,d}}(S),S)}=\frac{\rho(y,\{x^*,y,l_{\succ_{R,d}}(S)\})}{\rho(l_{\succ_{R,d}}(S),\{x^*,y,l_{\succ_{R,d}}(S)\})}=\frac{w(y)}{w(l_{\succ_{R,d}}(S))}.   
\end{equation}
By definition of $w(x^*)$
\begin{equation}\label{eq:dSFT-4}
\frac{\rho(x^*,S)}{\rho(l_{\succ_{R,d}}(S),S)}>\frac{w(x^*)}{w(l_{\succ_{R,d}}(S))}.
\end{equation}
Using equations (\ref{eq:dSFT-3}) and (\ref{eq:dSFT-4})
\begin{eqnarray*}
1+ \frac{\rho(x^*,S)}{\rho(l_{\succ_{R,d}}(S),S)}+ \sum_{y:y\in S, y\ne x^*,l_{\succ_{R,d}}(S)}\frac{\rho(y,S)}{\rho(l_{\succ_{R,d}}(S),S)} > 1+ \frac{w(x^*)}{w(l_{\succ_{R,d}}(S))}+  \sum_{y:y\in S, y\ne x^*,l_{\succ_{R,d}}(S)}\frac{w(y)}{w(l_{\succ_{R,d}}(S))}.
\end{eqnarray*}
Equivalently,
\[
\frac{1}{\rho(l_{\succ_{R,d}}(S),S)} > \frac{w(S)}{w(l_{\succ_{R,d}}(S))}.
\]
Therefore, $\alpha_S>0$. Now, from the definition of $\alpha_S$
\[
\alpha_S=1- \rho(l_{\succ_{R,d}}(S),S)\frac{w(S)}{w(l_{\succ_{R,d}}(S))} \Rightarrow \rho(l_{\succ_{R,d}}(S),S)=(1-\alpha_S)\frac{w(l_{\succ_{R,d}}(S))}{w(S)}.
\]
From here, for all $y\in S$ such that $y\ne x^*, l_{\succ_{R,d}}(S)$, using Axiom \ref{axiom:2-dDST}
\[
\frac{\rho(y,S)}{\rho(l_{\succ_{R,d}}(S),S)}=\frac{\rho(y,\{x^*,y,l_{\succ_{R,d}}(S)\})}{\rho(l_{\succ_{R,d}}(S),\{x^*,y,l_{\succ_{R,d}}(S)\})}=\frac{w(y)}{w(l_{\succ_{R,d}}(S))} \Rightarrow \rho(y,S)=(1-\alpha_S)\frac{w(y)}{w(S)}.
\]
It follows $\rho(x^*,S)=1-\sum_{y\in S, y\ne x^*}\rho(y,S)=\alpha_S+ (1-\alpha_S)\frac{w(x^*)}{w(S)}$. Hence, $\rho$ has a representation at $S$. 

\smallskip
\noindent \textbf{Case 2:} $x^*\not \in S$. Then $x^*\ne b_{\succ_{R,d}}(S)$. Again, for all $y\in S$ such that $y\ne l_{\succ_{R,d}}(S), b_{\succ_{R,d}}(S)$, using Axiom \ref{axiom:2-dDST}
\[
\frac{\rho(y,S)}{\rho(l_{\succ_{R,d}}(S),S)}=\frac{\rho(y,\{b_{\succ_{R,d}}(S),y,l_{\succ_{R,d}}(S)\})}{\rho(l_{\succ_{R,d}}(S),\{b_{\succ_{R,d}}(S),y,l_{\succ_{R,d}}(S)\})}=\frac{w(y)}{w(l_{\succ_{R,d}}(S))}.
\]
By Axiom \ref{axiom:1-dDST}
\[
\frac{\rho(b_{\succ_{R,d}}(S),S)}{\rho(l_{\succ_{R,d}}(S),S)}> \frac{\rho(b_{\succ_{R,d}}(S),\{x^*,b_{\succ_{R,d}}(S),l_{\succ_{R,d}}(S)\})}{\rho(l_{\succ_{R,d}}(S),\{x^*,b_{\succ_{R,d}}(S),l_{\succ_{R,d}}(S)\})}=\frac{w(b_{\succ_{R,d}}(S))}{w(l_{\succ_{R,d}}(S))},
\]
where the equation comes from (\ref{eq:dDST-2}). From here, follow the same steps as in Case 1. 

Finally, by definition of $w$, it is necessary that $\sum_{x\in X}w(x)>1$. By defining $w'(x)=\frac{w(x)}{w(X)}$ for all $x\in X$, $\rho$ has a d-DST representation $((\alpha_S), \succ_{R,d}, w')$. This completes our proof. $\blacksquare$

%%%%%%%%%%%%%%%%%%%%%%%%%%%%%%%%%%%%%%%%%%%%%%%%%%%%%%%%%%%%%%%%%%%
%%%%%%%%%%%%%%%%%%%%%%%%%% Prof of Proposition %%%%%%%%%%%%%%%%%%%%
%%%%%%%%%%%%%%%%%%%%%%%%%%%%%%%%%%%%%%%%%%%%%%%%%%%%%%%%%%%%%%%%%%%

\subsection{Proof of Proposition \ref{prop:d-DST_characterization_consistency}}
The necessity is straightforward. The sufficiency proof is similar to that in the proof of Proposition \ref{prop:d-DST_characterization}. Let $x^*$ and $z^*$ be $\succ_{R,d^c}$-best and $\succ_{R,d^c}$-worst in $X$, respectively. Let $w(z^*)=1$. Define
\[
w(y)=w(z^*)\frac{\rho(y,{\{x^*,y,z^*\}})}{\rho(z^*,{\{x^*,y,z^*\}})} \text{ for all $y$ such that $x^*\succ_{R,d^c} y \succ_{R,d^c} z^*$}.
\]
and
\begin{equation}\label{eq:d-DST_choosing_x*}
 \max_{y\in X\setminus x^*} w(y)<w(x^*)<\min_{(z,S): \{x^*,z\}\subseteq S, z\ne x^*}\frac{\rho(x^*,S)}{\rho(z,S)}w(z).   
\end{equation}
The proof then proceeds as in the proof of Proposition \ref{prop:d-DST_characterization}. The last thing to show is that $w(x)>w(y)$ whenever $x \succ_{R,d^c} y$ so that the two systems are consistent. We proceed in two steps.

\noindent \underline{\textit{Step 1: Showing $w(y)>w(z)$ whenever $x^*\succ_{R,d^c} y \succ_{R,d^c} z$.}} As in the proof of Proposition \ref{prop:d-DST_characterization} (see (\ref{eq:dDST-2})), we have
\begin{equation}\label{eq:dDST-2_consistency}
\frac{w(y)}{w(z)}=\frac{\rho(y,\{x^*,y,z\})}{\rho(z,\{x^*,y,z\})} \quad \text{for arbitrary $y,z$ such that $x^*\succ_{R,d^c} y \succ_{R,d^c} z$}.
\end{equation}
By definition of $\succ_{R,d^c}$, this would imply that $w(y)>w(z)$.

\noindent \underline{\textit{Step 2: Showing that choosing $w(x^*)$ is feasible.}} Let $x^*\succ_{R,d^c} y^* \succ_{R,d^c} \text{all other options}$. It is feasible to choose $w(x^*)$ as in (\ref{eq:d-DST_choosing_x*}) if for arbitrary $z\in S\setminus x^*$ and $S\supseteq \{x^*,z\}$, with $z$ and $y^*$ not necessarily distinct
\[
\frac{w(y^*)}{w(z)}<\frac{\rho(x^*,S)}{\rho(z,S)} \Leftrightarrow \frac{\rho(y^*,\{x^*,y^*,z\})}{\rho(z,\{x^*,y^*,z\})}<\frac{\rho(x^*,S)}{\rho(z,S)},
\]
where the second inequality uses (\ref{eq:dDST-2_consistency}). The second inequality holds due to Axiom \ref{axiom:1.5-dDST} whenever $z\ne y^*$. When $z=y^*$, the inequality follows from the definition of $\succ_{R,d^c}$, as $x^*\succ_{R,d^c} y^*$ by definition of $x^*$ and $y^*$. Steps 1 and 2 then imply that $w(x)>w(y)$ whenever $x \succ_{R,d^c} y$. This completes our proof. $\blacksquare$

%%%%%%%%%%%%%%%%%%%%%%%%%%%%%%%%%%%%%%%%%%%%%%%%%%%%%%%%%%%%%%%%%%%
%%%%%%%%%%%%%%%%%%%%%%%%%% Prof of Proposition %%%%%%%%%%%%%%%%%%%%
%%%%%%%%%%%%%%%%%%%%%%%%%%%%%%%%%%%%%%%%%%%%%%%%%%%%%%%%%%%%%%%%%%%
\subsection{Proof of Proposition \ref{prop:h-DST}} The proof is divided into Lemmata \ref{lema:h-DST-1}-\ref{lema:h-DST-3}.

\begin{lemma}\label{lema:h-DST-1} Any point in $\PRUM$ is arbitrarily close to some point in $\PhDST$.
\end{lemma}

\begin{proof}
Let $\bar{\rho}\in \PRUM$. Then there exists preferences $\succ_i$ and associated weight $\eta_i\in (0,1]$, for $i=1,\dots,q$ such that $\sum_{i=1}^q \eta_i=1$ and
\[
\bar{\rho}(x,S)= \sum_{i=1}^q \eta_i \mathbbm{1}(x \text{ is $\succ_i$-best in $S$}) \quad \text{ for all $x\in S$ and $S\in \mathcal{X}$.}
\]
Fix $\varepsilon>0$ arbitrarily small. Call $w_i\colon X\to \mathbb{R}_{++}$ an approximation of $\succ_i$ if 
\[
\bigg|\frac{w_i(x)}{w_i(S)} -  \mathbbm{1}(x \text{ is $\succ_i$-best in $S$})\bigg|\le \varepsilon \quad \text{ for all $x\in S$ and $S\in \mathcal{X}$.}
\]
One example is $w_i(x)=\exp(\lambda u_i(x))$ for all $x\in X$, where $u_i$ is a utility function representing $\succ_i$ and $\lambda>0$ sufficiently large. Let $\rho'\in \PhDST$ such that
\[
\rho'(x,S)= \sum_{i=1}^q \eta_i \rho_{\alpha,\succ_i,w_i}(x,S) \quad \text{ for all $x\in S$ and $S\in \mathcal{X}$,}
\]
where $\rho_{\alpha,\succ_i,w_i}$ is a random choice function having DST representation $(\alpha,\succ_i,w_i)$ with $w_i$ being an approximation of $\succ_i$. For all $x\in S\in \mathcal{X}$, $|\rho'(x,S)-\bar{\rho}(x,S)|$ is equal to
\begin{eqnarray*}
&&\bigg|\sum_{i=1}^q\eta_i\bigg(\alpha  \mathbbm{1}(x \text{ is $\succ_i$-best in $S$}) + (1-\alpha)\frac{w_i(x)}{w_i(S)} \bigg)- \sum_{i=1}^q \eta_i \mathbbm{1}(x \text{ is $\succ_i$-best in $S$})\bigg| \\
&=&(1-\alpha) \bigg|\sum_{i=1}^q \eta_i \bigg(\frac{w_i(x)}{w_i(S)} - \mathbbm{1}(x \text{ is $\succ_i$-best in $S$})\bigg) \bigg| \\
&\le &(1-\alpha)\varepsilon \sum_{i=1}^q \eta_i = (1-\alpha)\varepsilon < \varepsilon.
\end{eqnarray*}
In the last line, the first inequality follows from the fact that $w_i$ approximates $\succ_i$. As $|\rho'(x,S)-\bar{\rho}(x,S)|<\varepsilon$ for all $x\in S\in \mathcal{X}$, $\rho'$ lies sufficiently close to $\bar\rho$. $\blacksquare$
\end{proof}

\begin{lemma}\label{lema:h-DST-2} $\PhDST$ does not intersect the facets of $\PRUM$. 
\end{lemma}

\begin{proof} Proof is by a contradiction. Suppose $\rho\in \PhDST$ and $\rho$ lies on one of the facets of $\PRUM$. Note $\PRUM$ is a closed polytope characterized by Block-Marschak inequalities \citep{falmagne1978representation}:
\[
  \sum_{S': S^*\subseteq S'\subseteq X} (-1)^{|S'\setminus S^*|}\rho(x,S')\ge 0 \quad \text{ for all $x\in S^*$ and $S^*\in \mathcal{X}$.}  
\]
The fact that $\rho$ lies on one of the facets of $\PRUM$ implies that there exists $(x^*,S^*)$ with $x^*\in S^*\in \mathcal{X}$ such that 
\begin{equation}\label{eq:h-DST-0}
  \sum_{S': S^*\subseteq S'\subseteq X} (-1)^{|S'\setminus S^*|}\rho(x^*,S')=0.  
\end{equation}
As $\rho\in \PhDST$, by definition, there exists weights $(\eta_1,\dots,\eta_k)\in (0,1)^k$ summing to $1$ and RCF $\rho_{\alpha_i,\succ_i,w_i}$ having DST representation $(\alpha_i,\succ_i,w_i)$ such that
\begin{equation}\label{eq:h-DST-1}
\rho(x,S)=\sum_{i=1}^k \eta_i \rho_{\alpha_i,\succ_i,w_i}(x,S) \quad \text{ for all } x\in S, \text{ for all } S\in \mathcal{X}.
\end{equation}
As each $\rho_{\alpha_i,\succ_i,w_i}$ also admits a RUM representation, it satisfies the Block-Marschak inequalities at $(x^*,S^*)$:
\begin{equation}\label{eq:h-DST-2}
  \sum_{S': S^*\subseteq S'\subseteq X} (-1)^{|S'\setminus S^*|}\rho_{\alpha_i,\succ_i,w_i}(x^*,S')\ge 0.
\end{equation}
It follows from (\ref{eq:h-DST-0}), (\ref{eq:h-DST-1}), and (\ref{eq:h-DST-2}) that inequality (\ref{eq:h-DST-2}) must hold with equality. Let $Y = X \setminus S^*$. Any set $S'$ satisfying $S^* \subseteq S' \subseteq X$ can be uniquely represented as $S' = S^* \cup C$ for some subset $C \subseteq Y$. Under this substitution, the term $|S' \setminus S^*|$ becomes $|(S^* \cup C) \setminus S^*| = |C|$. Therefore, the fact that (\ref{eq:h-DST-2}) holds with equality is equivalent to
\begin{equation}\label{eq:h-DST-3}
  \sum_{C\subseteq Y} (-1)^{|C|}\rho_{\alpha_i,\succ_i,w_i}(x^*,S^* \cup C)= 0.
\end{equation}
We will obtain a contradiction to equation (\ref{eq:h-DST-3}) by showing that
\[
\sum_{C\subseteq Y} (-1)^{|C|}\rho_{\alpha_i,\succ_i,w_i}(x^*,S^* \cup C)\ge (1-\alpha_i)w_i(x^*)\sum_{C\subseteq Y}\frac{(-1)^{|C|}}{w_i(S^*)+w_i(C)}>0.
\]
In Step 1 below, we prove the first inequality. The second inequality is proved in Step 2.

\smallskip
\textbf{Step 1:} First, suppose $x^*$ is not the $\succ_i$-best element in $S^*$. Then it is not the best option in $S^*\cup C$. Using the DST representation of $\rho_{\alpha_i,\succ_i,w_i}$, for all $C\subseteq Y$, $\rho_{\alpha_i,\succ_i,w_i}(x^*,S^* \cup C)=(1-\alpha_i)\frac{w_i(x^*)}{w_i(S^*)+w_i(C)}$. It follows
\begin{equation}\label{eq:h-DST-4}
\sum_{C\subseteq Y} (-1)^{|C|}\rho_{\alpha_i,\succ_i,w_i}(x^*,S^* \cup C)= (1-\alpha_i)w_i(x^*)\sum_{C\subseteq Y}\frac{(-1)^{|C|}}{w_i(S^*)+w_i(C)}.
\end{equation}
Now, suppose $x^*$ is the $\succ_i$-best element in $S^*$. Partition $Y=Y^-\cup Y^+$ such that $x^*$ is the $\succ_i$-best element in $S^*\cup C$ for any $C\subseteq Y^-$, but it is not the best element in $S^*\cup C$ for any non-empty $C\subseteq Y^+$. Using the DST representation of $\rho_{\alpha_i,\succ_i,w_i}$,
\begin{alignat*}{3}
    \rho_{\alpha_i,\succ_i,w_i}(x^*,S^* \cup C)&=\alpha_i+ (1-\alpha_i)\frac{w_i(x^*)}{w_i(S^*)+w_i(C)} \quad &&\text{ for all $C\subseteq Y^-$}, \\
    \rho_{\alpha_i,\succ_i,w_i}(x^*,S^* \cup C)&= (1-\alpha_i)\frac{w_i(x^*)}{w_i(S^*)+w_i(C)} \quad &&\text{ for all $C \subseteq Y$ such that $C \cap Y^+ \neq \emptyset$}.
\end{alignat*}
It follows,
\begin{equation}\label{eq:h-DST-5}
\sum_{C\subseteq Y} (-1)^{|C|}\rho_{\alpha_i,\succ_i,w_i}(x^*,S^* \cup C)= \alpha_i \sum_{C\subseteq Y^-}(-1)^{|C|} + (1-\alpha_i)w_i(x^*)\sum_{C\subseteq Y}\frac{(-1)^{|C|}}{w_i(S^*)+w_i(C)}
\end{equation}
When $Y^-=\emptyset$, we have $\alpha_i \sum_{C\subseteq Y^-}(-1)^{|C|}=\alpha_i>0$. When $Y^-\ne \emptyset$, the binomial theorem says that $\alpha_i \sum_{C\subseteq Y^-}(-1)^{|C|}=0$. Therefore, combining (\ref{eq:h-DST-4}) and (\ref{eq:h-DST-5}), we have
\[
\sum_{C\subseteq Y} (-1)^{|C|}\rho_{\alpha_i,\succ_i,w_i}(x^*,S^* \cup C)\ge (1-\alpha_i)w_i(x^*)\sum_{C\subseteq Y}\frac{(-1)^{|C|}}{w_i(S^*)+w_i(C)}.
\]

\smallskip
\textbf{Step 2:} We first show that the following identity holds:
$$
\Gamma= \sum_{C\subseteq Y}\frac{(-1)^{|C|}}{w_i(S^*)+w_i(C)} = \int_0^\infty \exp (-tw_i(S^*)) \prod_{y \in Y} \left(1 - \exp(-t w_i(y))\right) dt.
$$
Since $w_i(x) > 0$ for all $x \in X$, $w_i(S^*) + w_i(C)>0$ for all $C\subseteq Y$. For any positive number $z$, the Laplace transform of the constant function 1 implies $1/z = \int_0^\infty \exp(-tz) dt$. Applying this identity, we obtain:
$$
\Gamma = \sum_{C \subseteq Y} (-1)^{|C|} \int_0^\infty \exp\left(-t  (w_i(S^*) + w_i(C))\right) dt.
$$
Since the summation is over a finite number of terms, we can interchange the order of summation and integration by the linearity of the integral:
\begin{align*}
\Gamma &= \int_0^\infty \sum_{C \subseteq Y} (-1)^{|C|} \exp\left(-t  (w_i(S^*) + w_i(C))\right) dt = \int_0^\infty \sum_{C \subseteq Y} (-1)^{|C|} \exp (-tw_i(S^*)) \exp (-tw_i(C)) dt
\end{align*}
The term $\exp (-tw_i(S^*))$ is independent of the summation variable $C$ and can be factored out of the sum:
$$
\Gamma = \int_0^\infty \exp (-tw_i(S^*)) \left( \sum_{C \subseteq Y} (-1)^{|C|} \exp (-tw_i(C)) \right) dt
$$
We now focus on simplifying the inner sum. The exponential of a sum is the product of exponentials: $\exp (-tw_i(C)) = \prod_{x \in C} \exp(-t w_i(x))$. The inner sum becomes:
$$
\sum_{C \subseteq Y} (-1)^{|C|} \prod_{x \in C} \exp(-t w_i(x)) = \sum_{C \subseteq Y} \prod_{x \in C} [-\exp(-t w_i(x))]=\prod_{y \in Y} [1 - \exp(-t w_i(y))]
$$
Substituting this product back into the integral expression for $\Gamma$ yields:
$$
\Gamma = \int_0^\infty \exp (-tw_i(S^*)) \prod_{y \in Y} \left(1 - \exp(-t w_i(y))\right) dt
$$
This integral representation shows that $\Gamma$ is strictly positive. This is because the domain of integration is $(0, \infty)$ and because the expression inside the integral is strictly positive as (1) $\exp\left(-t w_i(S^*)\right)$ is always strictly positive; and (2) for all $y\in X\setminus S^*$, the condition $w_i(y) > 0$ implies that $tw_i(y) > 0$. It follows that $0 < \exp(-t w_i(y)) < 1$, which ensures that each factor $(1 - \exp(-t w_i(y)))$ is strictly positive.

Steps 1 and 2 together generate a contradiction to equation (\ref{eq:h-DST-3}). Hence, the initial assumption is wrong and $\PhDST$ does not contain any points lying on a facet of $\PRUM$. $\blacksquare$
\end{proof}

\begin{lemma}\label{lema:h-DST-3} $\PRUM=\clPhDST$. 
\end{lemma}

\begin{proof}
To show $\PRUM=\clPhDST$, we show $\PRUM\subseteq\clPhDST$ and $\clPhDST\subseteq \PRUM$. Note that $\PhDST\subseteq \PRUM$. Additionally, $\PRUM$ is a closed polytope. Hence, $\clPhDST\subseteq \PRUM$. For the other direction, Lemma \ref{lema:h-DST-1} shows that any point in $\PRUM$ is arbitrarily close to some point in $\PhDST$. Hence, $\PRUM\subseteq\clPhDST$ by the definition of the closure of a set. From here, we can conclude that $\PRUM=\clPhDST$. $\blacksquare$
\end{proof}

\end{appendices}
\end{document}